\documentclass[a4paper,10pt]{article}
\usepackage[utf8]{inputenc}

\usepackage{blindtext}
\usepackage[english]{babel} 
\usepackage{amssymb, amsmath, amsthm, mathtools}
\usepackage{mathptmx}
\usepackage{geometry}
\usepackage[dvipsnames]{xcolor}
\pagecolor{white}
\usepackage{hyperref}
\usepackage[nottoc]{tocbibind} %Automatically adds the bibliography and/or the index and/or the contents, etc., to the Table of Contents listing
\usepackage{import}
\definecolor{light-gray}{gray}{0.75}
\usepackage[color=light-gray]{todonotes}
\usepackage{indentfirst} 
\usepackage{graphicx}  %The package builds upon the graphics package, providing a key-value interface for optional arguments to the \includegraphics command. This interface provides facilities that go far beyond what the graphics package offers on its own. 
\usepackage[all]{xypic} %for drawing diagrams https://en.wikibooks.org/wiki/LaTeX/Xy-pic
\usepackage{paralist}  %Provides enumerate and itemize environments that can be used within paragraphs to format the items either as running text or as separate paragraphs with a preceding number or symbol.
\usepackage{subfig}
\usepackage{authblk} %For \affil{}

\usepackage{cancel} %To usel \cancel{} in equations
\usepackage{tikz-cd} %For Tikz
\usepackage{float} %Improves the interface for defining floating objects such as figures and tables.

\theoremstyle{plain}
\newtheorem{theorem}{Theorem}[section]
\newtheorem{proposition}[theorem]{Proposition}
\newtheorem{lemma}[theorem]{Lemma}

\theoremstyle{definition}

\newtheorem{example}[theorem]{Example}

\makeatletter
\newsavebox{\@brx}
\newcommand{\llangle}[1][]{\savebox{\@brx}{\(\m@th{#1\langle}\)}%
  \mathopen{\copy\@brx\kern-0.5\wd\@brx\usebox{\@brx}}}
\newcommand{\rrangle}[1][]{\savebox{\@brx}{\(\m@th{#1\rangle}\)}%
  \mathclose{\copy\@brx\kern-0.5\wd\@brx\usebox{\@brx}}}
\makeatother

\newcommand{\extd}{\mathrm{d}}
\newcommand{\id}{\mathrm{id}}

\newcommand{\tens}{\otimes}

\newcommand{\Ad}{\mathrm{Ad}}
\newcommand{\Hom}{\mathrm{Hom}}

\newcommand{\del}{\partial}

\newcommand{\la}{\triangleright}
\newcommand{\ra}{\triangleleft}

\newcommand{\vspan}{\mathrm{span}}  %vectorspace span, \span is already taken

\newcommand{\Vol}{\mathrm{Vol}}
\renewcommand{\>}{\rangle}
\newcommand{\<}{\langle}

\newcommand{\C}{\mathbb{C}}
\newcommand{\N}{\mathbb{N}}
\newcommand{\R}{\mathbb{R}}
\newcommand{\Z}{\mathbb{Z}}

\newcommand{\CO}{\mathcal{O}}
\newcommand{\CD}{\mathcal{D}}
\newcommand{\CA}{\mathcal{A}}
\newcommand{\CG}{\mathcal{G}}

\newcommand{\CC}{\mathcal{C}}

\newcommand{\CL}{\mathcal{L}}
\newcommand{\CU}{\mathcal{U}}
\newcommand{\CZ}{\mathcal{Z}}
\newcommand{\cg}{\mathfrak{g}}

\renewcommand{\imath}{{\mathfrak{i}}}

\title{Finite group gauge theory on graphs and gravity-like modes}

\author[1]{Shahn Majid \thanks{\textit{E-mail}: \href{mailto:s.majid@qmul.ac.uk}{\texttt{s.majid@qmul.ac.uk}}}}
\author[2]{Francisco Simão \thanks{\textit{E-mail}: \href{mailto:f.castelasimao@qmul.ac.uk}{\texttt{f.castelasimao@qmul.ac.uk}}}}
\affil[1,2]{\normalsize School of Mathematical Sciences\par\nopagebreak Queen Mary University of London\par\nopagebreak
Mile End Rd, London E1 4NS}
\date{}

\begin{document}
\maketitle
\begin{abstract} 
We study gauge theory with finite group $G$  on a graph $X$ using noncommutative differential geometry and Hopf algebra methods with $G$-valued holonomies replaced by  gauge fields  valued in a `finite group  Lie algebra' subset of the group algebra $\C G$ corresponding to the complete graph differential structure on $G$.  We show that this richer theory decomposes as a product over the nontrivial irreducible representations $\rho$ with dimension $d_\rho$ of certain noncommutative $U(d_\rho)$-Yang-Mills theories, which we introduce. The Yang-Mills action  recovers the Wilson action for a lattice but now with additional terms. We compute the moduli space $\CA^\times / \CG$ of regular connections modulo gauge transformations on connected graphs $X$. For $G$ Abelian, this is given as expected by phases  associated to fundamental loops but with additional $\R_{>0}$-valued modes on every edge resembling the metric for quantum gravity models on graphs. For nonAbelian $G$, these modes become positive-matrix valued modes. We study the  quantum gauge field theory in the Abelian case in a functional integral approach, particularly for $X$ the finite chain $A_{n+1}$, the $n$-gon $\Z_n$ and the single plaquette $\Z_2\times \Z_2$. We show that, in stark contrast to usual lattice gauge theory, the Lorentzian version is well-behaved, and we identify novel boundary vs bulk effects in the case of the finite chain.  We also consider gauge fields valued in the finite-group Lie algebra corresponding to a general Cayley graph differential calculus on $G$, where we study an obstruction to closure of gauge transformations. \end{abstract}

\section{Introduction}\label{secintro}

Lattice gauge theory (LGT) is a standard and well-established approximation of continuum gauge theory, as for example in lattice QCD~\cite{Smi,KogSus:ham}. Generalised to other graphs $X$, it may also be relevant to networks and topological data analysis. The case of a finite gauge group $G$  is particularly relevant to the Kitaev model for topologically fault tolerant quantum computing~\cite{Kit} as well as of interest as a simplifying assumption in the study of entanglement~\cite{Don:lat} and in the theory of spin networks~\cite{Bae:spi}.  This established approach nevertheless raises the fundamental  issue: to what extent do the various constructions map onto the continuum limit of which lattice theory is meant to be an approximation?

Noncommutative geometry addresses this head on in the following manner. First of all, a graph $X$ is a perfectly good and exact -- but noncommutative -- differential geometry with not only differential 1-forms $\Omega^1$ (spanned by arrows of the graph, and not commuting with functions) but also $\Omega^2$ and higher forms. This  is simply one limit of a more general conception of geometry based on a coordinate algebra $B$, which in another limit $B=C^\infty(M)$ includes classical geometry on a manifold $M$ (and also includes $B$ noncommutative, relevant to quantum spacetime). By embedding everything into a single coherent framework one can more systematically take the continuum limit of an infinite graph, as well as more systematically transfer geometric concepts from the continuum to the graph case. Such a point of view was used in recent works \cite{Ma:squ,ArgMa,BliMa} for baby quantum gravity models on graphs and we refer to \cite{Ma:gra,BegMa} for an introduction to the methods of quantum Riemannian geometry (QRG) used, as well as \cite{Con} for another approach.  Gauge theory on lattices was previously considered using noncommutative geometry in \cite{CasPa:dis, ACI:dis, DMHS:lgt, DMH:fin} with  gauge-fields valued in matrices (but note that the discrete groups in these works refer to $X$ to simplify its noncommutative geometry). 

Secondly,  we now go much further and also consider the gauge group $G$ as itself a noncommutative graph geometry. Here a translation-invariant differential structure corresponds to a Cayley graph structure on $G$ and in the bicovariant case the dual to the space of invariant 1-forms becomes a `finite group Lie algebra' as, for example,  in \cite{MaRie}. For most of the paper we focus on the universal calculus corresponding to the complete graph on $G$, where the finite group Lie algebra is the subspace $\C G^+\subset \C G$ of the group algebra such that the coefficients in the basis of $\C G$ labelled by group elements add up to zero. The theory is a special case of the braided-Lie algebra of any strict quantum group and includes the usual Lie algebra of a Lie group as a special case of that. Accordingly, we take a gauge field to be an element $\alpha\in \C G^+ \tens \Omega^1$, which means to assign an element of $\C G^+$ to every arrow (subject to a relation with the value on the reversed arrow).  This is fundamentally different from usual LGT, where one assigns a group element (the `holonomy' along the arrow) to every arrow,  since we  allow complex linear combinations of group elements at each arrow. This is in the same spirit as boolean states labeling the Hilbert space basis in a quantum computer, which then proceeds with linear combinations. Hence, the big difference in the present work is that $\C G^+$ extends the gauge fields of LGT to a `quantum' version in the quantum computing sense. This will result in a somewhat different theory but has the merit that we have restored a geometric picture in line with the continuum theory where gauge fields need to be  Lie algebra valued. This is also a necessary warm-up  if we want to understand models with Hopf algebra gauge symmetry\cite{BrzMa}, including the Kitaev model where the full gauge symmetry is the quantum double $D(G)$ rather than $G$ itself\cite{Sho,CowMa}.

We study the resulting theory in detail, including a Yang-Mills action ${1\over c}\int (F,F)_2$ built from the curvature of $F(\alpha)$ and a QRG-type metric $(\ ,\ )_2$ on $\Omega^2$ (and with the integral a sum in the graph case). Our first new result is that the gauge theory that emerges based on such $\alpha$ is actually a product of independent noncommutative $U(d_\rho)$-gauge theories, one for each nontrivial irrep $\rho$ of $G$ of dimension $d_\rho$. We introduce the relevant notion of noncommutative $U(d)$-gauge theory in Section~\ref{secUd} as a generalisation of the noncommutative $U(1)$-gauge theory studied in \cite{Ma:cli}\cite{MaRai}. This works on any differential algebra $(B,\Omega,\extd)$ and if $B$ is smooth functions on a manifold, exactly recovers classical $U(d)$-gauge theory. We use the theory in the case of $B$ the algebra of functions on the vertices on a directed graph $X$ and $\Omega^1$ the arrows. Thus, we arrive at a rather surprising picture of finite group gauge theory with the universal calculus on the fibre as a product of continuum-group gauge theories on the same base. Our focus is the base being a directed graph $X$, but the decomposition applies generally, including a continuum base. Thus, for the permutation group $G=S_3$ on three elements, $\C S_3$-gauge theory on the continuum decomposes as a usual $U(2)\times U(1)$ gauge theory, and similarly $\C S_4$-gauge theory as  $U(3)\times U(3)\times U(2)\times U(1)$, i.e. not far from the Standard Model. Section~\ref{secpre} also explains other preliminary background from QRG and extends this to metrics on $\Omega^2$.  For a directed graph there are four canonical choices for $\Omega^2$ which we call $\Omega^2_{max}, \Omega^2_{med},\Omega^2_{med'}$ and $\Omega^2_{min}$, cf.\cite{BegMa}. Section \ref{secdiffG} compares our approach to the metric on $\Omega^2$ with the limited known results \cite{MaRai,Ma:hod}  for Hodge $\star$ on a finite group or quantum group with its natural $\Omega^2_{wor}$. This is typically a further quotient of $\Omega_{min}$. 

Section \ref{sec:GraphUniCalcGaugeGroup} then describes our formulation of gauge theory on graphs and lattices $X$ as it naturally arises in this context. Section~\ref{secgraphgauge} describes gauge transformations and curvature of $\alpha$ as well as the above decomposition into noncommutative $U(d_\rho)$-gauge theories in the graph case. Section~\ref{sechol} makes a shift from $\alpha$ to what we call `extended holonomies' $u$ to better make contact with LGT. Section~\ref{secwilson} shows the Yang-Mills action for each nontrivial irrep $\rho$ and different choices of $\Omega^2$, see \eqref{eq:CasesExteriorAlgebraGraph}, recovering the usual Wilson action~\cite{Wil:con} in LGT plus additional terms. That we see terms similar to the Wilson action is in line with results in  \cite{CasPa:dis, ACI:dis, DMHS:lgt, DMH:fin}, but the context and details of our results are different. In particular, we do not only consider $\Omega^2_{wor}$ as normally used when the base $X$ is itself a group lattice (although we do briefly cover this case, in Secttion~\ref{secXG}). The other major differences are that we allow any metric on $X$ and, as mentioned, our gauge fields are `Lie algebra' valued. In all cases, there are certain terms which one does not see in usual LGT as they are constants in the case where $u$ is group valued (as opposed to a linear combination $u\in \C G$). 

Section~\ref{sec:mod} then proceeds to analyse the moduli space $\CU^\times/\CG$ of regular connections (in extended holonomy form) modulo gauge transformations. In the Abelian case, the theory amounts to two `entangled'  theories, one of which looks like LGT for copies of $U(1)$ in that there is a phase at each arrow which can be mostly gauged away so that the moduli space becomes a phase for  for each fundamental loop and each group generator. The other looks a  bit like a graph metric, namely an element of $\R_{>0}$ at every edge.  We call these `gravity like modes' but the action is typically far from the QRG Ricci scalar. Theorem~\ref{thm:MasslessScalar} says that for $X$ a tree graph, however, the theory is equivalent to a massless positive-valued (or `tropical') scalar field field much as for quantum gravity on $\Z$ in \cite{ArgMa}. Section~\ref{secnona}  looks briefly at the case of $G$ nonAbelian, now with Hermitian positive elements of $M_{d_\rho}(\C)$ as the `gravity-like modes'.  

Having defined the Yang-Mills action, Section~\ref{secQFT} takes a look at the resulting QFT in both the Euclidean and Lorentzian cases (where there is an $\imath$ in front of the action). We take a functional-integral approach and focus on Abelian $G$, and  $X$ the finite chain $A_{n+1}$ graph $\bullet$--$\bullet$--$\cdots$--$\bullet$ with $n+1$ vertices,  the polygon $\Z_n$ and the square $\Z_2\times \Z_2$ (which has different 2-forms compared to $\Z_4$). The $A_{n+1}$ case is simply connected (i.e. a tree graph) so that only the gravity-like $\R_{>0}$ modes on edges remain, and we obtain results not unlike those for baby quantum gravity models~\cite{Ma:squ,ArgMa} with infinities handled by a cut-off $L$ in field strength and nonzero relative uncertainties. For $A_3$ we solve the Euclidean theory exactly, for higher $n$ we illustrate the theory with $A_6,\Z_5$ using numerical Monte-Carlo methods, and for general $n$ we obtain the asymptotic behaviour for both large and weak coupling limits. One of the features visible in Figure~\ref{fig:A6Z5} for $A_6$ is that the expectation values behave differently as we approach the ends of the chain compared to the bulk (this is in line with such effects in the quantum Riemannian geometry in \cite{ArgMa2}) and that this is more pronounced in the Lorentzian case. We also compare with parallel results with $U(1)$ lattice gauge theory with broadly similar behaviour for the Euclidean square and vastly better behaviour for the Lorentzian square (where it is known that $U(1)$ lattice gauge theory is rather wild), see Figure~\ref{fig:Z2Z2Lorexpval}.

Section~\ref{secnonu} briefly considers what happens if we take a different, non-universal, calculus on the structure group $G$ defined by an $\Ad$-stable generating set $\CC\subset G\setminus
\{e\}$. There is is still a finite group Lie algebra, now built on $\C \CC$,  but a long-standing issue here is how to then define gauge transformations that are differentiable and to ensure that these are closed under composition. We propose one solution to this based on a condition $\CC^2\subseteq \CC\cup\{e\}$, which we show can hold but is quite restrictive.  Section~\ref{secrem} ends with some concluding remarks about directions for further work.

\section{Preliminaries from noncommutative differential geometry}\label{secpre}

Gauge theory in noncommutative geometry makes sense on any possibly noncommutative algebra $B$, a quantum group  $H$ as `gauge group' and potentially nontrivial `quantum principal bundle' $P$ over it \cite{BrzMa}, albeit we will be concerned only with trivial bundles $P=B\tens H^*$ and $H=\C G$, where the latter is  the group Hopf algebra of a finite group obtained by taking $G$ as a basis and linearly extending the group product. The simplest case $H=\C \Z_2$ will turn out to be equivalent to the noncommutative $U(1)$-gauge theory in \cite{MaRai,Ma:cli} and a small new result, in Section~\ref{secUd}, is to extend this to noncommutative $U(d)$-gauge theory on any $*$-exterior algebra $(B,\Omega,\extd)$. We then specialise to the case where the base algebra is functions on a graph, recalling first the basics of NCG relevant to graphs such as in \cite{Ma:gra} and used in subsequent works \cite{Ma:squ,ArgMa} for quantum gravity on graphs.  An introduction to the entire formalism can be found in \cite{BegMa}.

\subsection{Noncommutative differentials and $U(d)$-gauge theory}\label{secUd}

Let $B$ be a possibly noncommutative `coordinate algebra'. This is required to be a unital $*$-algebra. A first order calculus on $B$ means a $B$-bimodule $\Omega^1$ and a map $\extd\colon B\to \Omega^1$ obeying the Leibniz rule such that the map $B\tens B\to \Omega^1$ given by $b\tens c\mapsto b\extd c$ is surjective. This implies that every such $\Omega^1$ is a quotient of a certain universal first order calculus. 

We also need $\Omega^1$ to extend to a differential graded algebra $(\Omega,\extd)$, which we will need only to degree 2. Here $\Omega = \bigoplus_{n\geq0} \Omega^n$ with $\Omega^0 = B$, wedge product $\wedge\colon \Omega^n \tens_B \Omega^m \to \Omega^{n+m}$, differential $\extd\colon \Omega^n \to \Omega^{n+1}$ and $*$ now obeying 
\begin{align*}
    &\extd^2 = 0,&
    &\extd(\omega \wedge \eta) = \extd \omega \wedge \eta + (-1)^n \omega \wedge \extd \eta,&
    &\extd \omega^* = (\extd \omega)^*,&
    &(\omega \wedge \eta)^* = (-1)^{nm} \eta^* \wedge \omega^*,
\end{align*}
for all $\omega \in \Omega^n$, $\eta \in \Omega^m$. A $*$-calculus is {\em inner} if there is an anti-Hermitian $\theta\in\Omega^1$ such that $\extd=[\theta, ]$ in degree 0  (graded commutator if we want to include higher degree).

At this level, one already has a `noncommutative $U(1)$-gauge theory', which is defined by taking  connections to be anti-Hermitian one forms $\alpha\in \Omega^1$ with curvature $F=\extd \alpha+\alpha\wedge\alpha$ and gauge transformations 
\[ \alpha \mapsto \alpha^\gamma=\gamma\alpha\gamma^{-1}+ \gamma\extd\gamma^{-1},\quad F(\alpha^\gamma)=\gamma F(\alpha)\gamma^{-1}\]
by unitary elements $\gamma\in B$. We will need the following generalisation, which we call noncommutative $U(d)$-gauge theory. We let $\alpha\in M_d(\C)\tens \Omega^1$ and $\gamma\in M_d(\C)\tens B$ with $\alpha^*=-\alpha$ and $\gamma^*\gamma=\gamma\gamma^*=\id\tens 1$, where we now use the Hermitian conjugate $\dagger$ on $M_d(\C)$. Explicitly the data is now $\alpha_{ij}\in \Omega^1$ and $\gamma_{ij}\in B$ such that
\[ (\alpha_{ji})^*=-\alpha_{ij},\quad \gamma_{ik}(\gamma_{jk})^*=( \gamma_{ki})^*\gamma_{kj}=\delta_{ij}\]
(sum over $k$ understood). This is similar in spirit to $K$-theory of $B$ where one considers Hermitian projections on $M_d(\C)\tens B$. The gauge transformation and curvature (which is now in $M_d(\C)\tens \Omega^2$) take the form
\begin{equation*}   \alpha^\gamma_{ij}:=\gamma_{ik}\alpha_{kl}(\gamma_{jl})^*+ \gamma_{ik}\extd(\gamma_{jk})^*,\quad F(\alpha)_{ij}:=\extd \alpha_{ij}+ \alpha_{ik}\wedge\alpha_{kj}.\end{equation*}
Here, and below,  $(\gamma_{ji})^*$ can be replaced by $(\gamma^{-1})_{ij}$ where $\gamma^{-1}$ is the inverse in $M_d\tens B$ if working without the benefit of a $*$ operation, for example over another field. The focus on the unitary case is for physical applications where this is important for unitarity of the theory.

\begin{lemma} The gauge transform preserves the anti-Hermitian condition, the curvature is anti-Hermitian and transforms as expected:
\[ (\alpha^\gamma_{ji})^*=-\alpha^\gamma_{ij},\quad (F(\alpha)_{ji})^*=-F(\alpha)_{ij},\quad  F(\alpha^\gamma)_{ij}=\gamma_{ik}F(\alpha)_{kl}(\gamma_{jl})^*.\]
\end{lemma}
\begin{proof} This is elementary but serves as an introduction to working with $*$-differential algebras, so we include the proof, as follows:
\begin{align*}(\alpha^\gamma_{ji})^*&=(\gamma_{jk}\alpha_{kl}(\gamma_{il})^*)^*+(\gamma_{ik}\extd(\gamma_{jk})^*)^*=\gamma_{il}(\alpha_{kl})^*(\gamma_{jk})^*+(\extd \gamma_{ik})(\gamma_{jk})^*\\
&=-\gamma_{il}\alpha_{lk}(\gamma_{jk})^*-\gamma_{ik}\extd( \gamma_{jk})^*+\extd(\gamma_{ik}(\gamma_{jk})^*)=-\alpha^\gamma_{ij}
\end{align*}
where we used the Leibniz rule and $\gamma$ unitary along with $\extd 1\delta_{ij}=\delta_{ij}\extd 1=0$. Similarly
\[ (F(\alpha)_{ij})^*=-\extd\alpha_{ij}-\alpha_{ik}\wedge\alpha_{kj}=-F(\alpha)_{ij}\]
where we used that $*$ on a product of 1-forms reverses the order but in the process gets an extra minus sign. Finally,
\begin{align*} F(\alpha^\gamma)&=\extd(\gamma_{ik}\alpha_{kl}(\gamma_{jl})^*) + \gamma_{ik}\extd(\gamma_{jk})^*)+\gamma_{ik}\alpha_{kl}(\gamma_{pl})^*\wedge\gamma_{pm}\alpha_{mn}(\gamma_{jn})^*+\gamma_{ik}\alpha_{kl}(\gamma_{pl})^*\wedge\gamma_{pm}\extd(\gamma_{jm})^*\\
&\quad+ \gamma_{ik}(\extd(\gamma_{pk})^*)\wedge \gamma_{pm}\alpha_{mn}(\gamma_{jn})^*+\gamma_{ik}\extd (\gamma_{mk})^*\wedge\gamma_{mn}\extd(\gamma_{jn})^*\\
&=\extd(\gamma_{ik}\alpha_{kn}(\gamma_{jn})^*)+ \extd\gamma_{ik}\wedge \extd(\gamma_{jk})^*+\gamma_{ik}(\alpha_{kl}\wedge\alpha_{ln})(\gamma_{jn})^*+ \gamma_{ik}\alpha_{kn}\wedge\extd(\gamma_{jn})^*- (\extd\gamma_{ik})\wedge \alpha_{kl}(\gamma_{jl})^*-\extd\gamma_{ik}\wedge\extd(\gamma_{jk})^*
\end{align*}
using the Leibniz rule for $\extd$ to expand as well as (as above) transfer with a minus sign to the other $\gamma$ in a product $\gamma\gamma^*$. We then used the unitarity to simplify. All that remains is to  expand the first term using the graded-Leibniz rule, where  $\extd$ picks up a minus sign when moving past $\alpha_{kl}$ to act on $(\gamma_{jl})^*$. We  then see the two terms of $\gamma_{ik}F(\alpha)_{kn}(\gamma_{jn})^*$ and that all the other terms cancel. \end{proof}

To define the Yang-Mills action, we need some kind of quantum metric or Hodge operation. In quantum Riemannian geometry \cite{BegMa} a quantum metric is $\cg\in \Omega^1 \tens_B \Omega^1$ obeying the reality condition $\mathrm{flip} \circ (* \tens *) \cg = \cg$ (the composition on the left is well-defined on $\Omega^1\tens_B\Omega^1$), some form of symmetry condition, and which is invertible in the sense of a bimodule map (the `inverse metric')  $(\ ,\ )\colon \Omega^1 \tens_B \Omega^1 \to B$ satisfying $((\omega, \cdot )\tens_B \id) \cg = \omega = (\id \tens_B (\cdot ,\omega)) \cg$. In terms of $(\ ,\ )$,  the reality condition is equivalent to $(\omega,\eta)^* = (\eta^*,\omega^*)$ for all $\omega, \eta\in \Omega^1$. At a general level the most natural symmetry notion is $\wedge(g)=0$, but there are others in specific contexts. In a Riemannian context we could also ask for positive definiteness $(\omega^*,\omega)$ a positive element of $B$ (hence its integral via  a positive linear functional $\int: B\to \C$ is positive), but we do not assume this here.  

Hodge theory does not, however, necessarily exist in the context of NCG (for example, there may not even be a top form in the exterior algebra). Hence, for our limited purposes we adopt a different approach: we assume directly that we are given as further input data a degree 2 metric  $\cg_2\in \Omega^2\tens_B\Omega^2$ with bimodule map `inverse metric'  $(\ ,\ )_2\colon \Omega^2\tens_B \Omega^2\to B$ and analogous reality properties 
\[ {\rm flip}\circ(*\tens *)(\cg_2)=\cg_2,\quad (\omega,\eta)_2^* = (\eta^*,\omega^*)_2\]
now for all $\omega,\eta\in \Omega^2$. Classically, this would be built from $\cg$ and the antisymmetric form of the exterior algebra, but in the quantum case we would need to specify $\Omega^2$ further to achieve this. We could also ask for positivity of $\cg_2$ but we do not do this. 

Given this data, we also take an inner product on $M_d(\C)$ given by the trace, so that for any $M_d(\C)$-valued 2-forms $F,G$,
\[ {\rm Tr}(F,G)_2:= (F_{ji},G_{ij})_2.\]
We can now define the Yang-Mills Lagrangian for noncommutative $U(d)$-gauge theory with coupling $c$ as 
\begin{equation*}  
    \CL(\alpha)\coloneqq \frac{1}{c}{\rm Tr}((F(\alpha)^*, F(\alpha))_2) = \frac{1}{c}((F(\alpha)^*)_{ji},  F(\alpha)_{ij})_2= \frac{1}{c}((F(\alpha)_{ij})^*,  F(\alpha)_{ij})_2= - \frac{1}{c} ((F(\alpha)_{ji},  F(\alpha)_{ij})_2,\end{equation*}
where we sum over $i,j$. Under a gauge transformation, we have
\[ (\gamma_{jk}F_{kl}(\gamma_{il})^*, \gamma_{im}F_{mn}(\gamma_{jn})^*)_2=(\gamma_{jk}F_{kl}, (\gamma_{il})^*\gamma_{im}F_{mn}(\gamma_{jn})^*)_2=(\gamma_{jk}F_{kl}, (\gamma_{il})^*\gamma_{im}F_{mn}(\gamma_{jn})^*)_2=\gamma_{jk}(F_{kl}, F_{ln})_2(\gamma_{jn})^*\]
which is invariant if $(\ ,\ )_2$ has its image in the centre of $B$ using $\gamma$ unitary. In general, however, this need not be invariant  but it would be typical to assume a linear map $\int: B\to \C$ which is a {\em trace} in the sense $\int ab=\int ba$ for all $a,b\in B$. More precisely, all we really need is that $\int a b c=\int c a b$ for $b$ in the image of $(\ , \ )_2$ and all $a,c\in B$. Then we see that
\begin{equation*}
    S[\alpha]=\frac{1}{c} \int {\rm Tr}((F(\alpha)^*, F(\alpha)))_2= - \frac{1}{c}\int (F_{ji}, F_{ij})_2\end{equation*}
is gauge-invariant.

\subsection{Differentials on graphs} 

If $B=\C(X)=C(X,\C)$ is the commutative algebra of complex functions on a discrete set $X$, the choices of $\Omega^1$ are precisely digraphs with vertex set $X$, i.e. the specification of at most one arrow between some distinct vertices. We do not allow self-arrows. For a $*$-calculus we need the digraph to be {\em bidirected} in the sense that if $x\to y$ is an arrow then so is $y\to x$. This is the same as an undirected graph with every edge given  2-way arrows. We let $E$ denote the set of arrows. Then 
\[
    \Omega^1 = \vspan_\C\langle e_{x\to y} \vert x\to y \in E \rangle,
    \quad 
    \extd f = \sum_{x\to y} (f(y)-f(x))e_{x\to y},
\]
\[
    f.e_{x\to y} = f(x)e_{x\to y},
    \quad 
    e_{x\to y}.f = f(y)e_{x\to y},
    \quad
    e_{x\to y}^*=-e_{y\to x},
\]
for the corresponding calculus. Here $\Omega^1=\C E$ has $E$ as a vector space basis as indicated, and $f\in \C(X)$. We will often work with the $\delta$-function basis $\{\delta_x\ |\ x\in X\}$ of $\C(X)$, where $\delta_x(y) = \delta_{x,y}$ are the delta functions. If the graph has a finite number of arrows then the calculus is inner with $\theta = \sum_{x\to y}e_{x\to y}$. For infinite graphs one needs additional techniques which we explain in the group case below (e.g. of a square lattice).

 In the case of digraphs the maximal prolongation exterior algebra has 3 natural quotients, giving four of interest. They are all  given by the tensor algebra $T_{\C(X)} \Omega^1$ (this is the path algebra on the graph) with relations
\begin{equation}
    \label{eq:RelationsExtAlgGraph}
    \sum_{y\colon x\to y\to z} e_{x\to y} \wedge e_{y\to z} = 0,
\end{equation}
for fixed $x$, $z$, obeying one of four conditions \cite{BegMa}
\begin{align}
    \label{eq:CasesExteriorAlgebraGraph}
	&\Omega_{max}\colon \,\, x\neq z \text{ and }x \not \to z,&
	&\Omega_{med'}\colon \,\,x \not \to z,&
	&\Omega_{med}\colon \,\, x\neq z,&
    &\Omega_{min}\colon \text{ all } x,z.
\end{align}
The differential in the case of $\Omega_{max}$ and $\Omega_{med'}$ is
\begin{align*}
	\extd e_{x\to y} &= \sum_{s: s\to x\to y} e_{s \to x} \wedge e_{x\to y} + \sum_{ t: x\to y \to t} e_{x\to y} \wedge e_{y\to t} - \sum_{i:x\to i \to y}  e_{x \to i} \wedge e_{i\to y}\\
	&= \{\theta,e_{x\to y}\} - \sum_{i:x\to i \to y}  e_{x \to i} \wedge e_{i\to y}
\end{align*}
using the anticommutator, and we drop the last term for $\Omega_{med}$ and $\Omega_{min}$, which makes these two exterior algebras inner. 

 For a digraph, a generalised quantum metric (i.e. before we impose quantum symmetry of some kind) has the form \cite{Ma:gra,BegMa}
\begin{align*}
    &(e_{x\to y}, e_{y'\to x'}) = \lambda_{x\to y} \delta_{x,x'}\delta_{y,y'} \delta_x,&
    &\cg = \sum_{x\to y} \cg_{x\to y} e_{x\to y} \tens_B e_{y\to x},&
    &\lambda_{x\to y} = \frac{1}{\cg_{y\to x}}
\end{align*}
where the $\cg_{x\to y} \in \R^\times$ are interpreted as the square arrow lengths associated with $x\to y$. Note the arrow reversal in the last formula. In a graph context the most natural notion of quantum symmetry is {\em edge symmetric} meaning $\cg_{x\to y}=\cg_{y\to x}$. In this case a quantum metric is just a real `square length' assigned to every (undirected) edge exactly as one would expect for a lattice embedded in a (pseudo)-Riemannian manifold. For some lattices, however, there is no quantum-Levi-Civita connection unless we depart from edge symmetry in a certain way \cite{ArgMa2}.

\begin{lemma}
\label{lem:InnerProductOmega2X}
A bimodule non-degenerate inner product $(\ , \ )_2\colon \Omega^2 \tens_B \Omega^2 \to B$ satisfying $(\omega,\eta)^* = (\eta^*,\omega^*)$ for all $\omega, \eta\in \Omega^2$ is of the form
\begin{align*}
    (e_{x\to y} \wedge e_{y\to z}, e_{z'\to y'} \wedge e_{y'\to x'})_2
    = \lambda_{x \to y \to z \to y'} \delta_{x,x'}\delta_{z,z'} \delta_x
\end{align*}
where the weights $\lambda_{x \to y \to z \to y'} \in \C^\times$ are required to obey
\[ 
    \sum_{y\colon x\to y\to z \to y' \to x} \lambda_{x \to y \to z \to y' } = 0,\quad \lambda_{x \to y \to z \to y'}^* = \lambda_{x \to y' \to z \to y}\]
for all $x,z$ constrained as in \eqref{eq:CasesExteriorAlgebraGraph} according the choice of exterior algebra, and $y, y'$ any vertices between $x$ and $z$.
\end{lemma}
\begin{proof}
The general form of $(\ ,\ )_2$ can be computed as in \cite[Proposition 1.28]{BegMa} in order to descent to the tensor product over $B$. The further conditions on $\lambda_{x\to y \to z \to y' }$ follow directly from the relations \eqref{eq:RelationsExtAlgGraph} applied on the first and second factors. Furthermore we need $\lambda_{x \to y \to z \to y' }$ to be non-zero for $(\ ,\ )_2$ to be non-degenerate. The reality condition  is equivalent to $ \lambda_{x \to y \to z \to y'}^* = \lambda_{x \to y' \to z \to y}$. \end{proof}

In a geometric context the coefficients $\lambda_{x\to y\to z\to y'}$ would be defined quadratically in terms of the inverse metric coefficients $\lambda_{x\to y}$. Although it is not necessarily a unique way to do this, a class of solutions inspired by the case of $X$ a discrete group discussed in the next section is
\begin{proposition}\label{eq:SolLambda(,)2}
The coefficients $\lambda_{x\to y\to z\to y'}$ given by
\begin{equation*}
\lambda_{x\to y \to z \to y'} = A_{x,z}  \lambda_{x\to y} \left( \lambda_{x\to y'} -  \delta_{y,y'} \sum_{w:x\to w} \lambda_{x\to w}\right)
\end{equation*}
with real dimensionless weights $A_{x,z}$ chosen arbitrarily for every pair of vertices $x$, $z$ as in \eqref{eq:CasesExteriorAlgebraGraph} satisfy the conditions of Lemma \ref{lem:InnerProductOmega2X}.
\end{proposition}
\begin{proof}
This is a straight forward computation to verify that this obeys the conditions of Lemma \ref{lem:InnerProductOmega2X}. \end{proof}

One can start with a more general ansatz with weights $A_{x,y,z,y'}$, but applying the conditions then leads to $A_{x,y,z,y'} = A_{x,\tilde y,z,\tilde y'}$ for any $y,y',\tilde y, \tilde y'$ between $x$ and $z$, and we find that the weights only depend on $x, z$. Note that the only exterior algebra where this solution specifies all coefficients $\lambda_{x\to y \to z \to y'}$ is the minimal one $\Omega_{min}$. The content of this solution is that the coeffiient $\lambda_{x\to y \to z \to y' }$ with $y\neq y'$ associated with the quadrilateral
\[
\begin{tikzcd}
{}	& y'  \arrow[dl]					& {} \\
x \arrow[rd] 	& {}	& z    \arrow[ul] \\
{}	& y	\arrow[ur]			& \\
\end{tikzcd}
\]
is the product of the inverse metric coefficients along the two arrows $x\to y$ and $x\to y'$ outgoing from the starting point $x$. For $y=y'$, the weights are then determined via the conditions of Lemma \ref{lem:InnerProductOmega2X}.

\subsection{Differentials on discrete groups} \label{secdiffG}

Just as Lie groups provide nice examples of differential manifolds, discrete groups $G$ provide nice examples of graphs and hence of differential calculi on the $G$ as sets of vertices. This will not only be relevant for the fibres but also when $X$ is an $M$-dimensional lattice $X=\Z^M$ or is another discrete group. Note that the group in this context, while denoted $G$, is not necessarily the gauge group but can also be the base space with $B=\C(G)$.  Here a {\em Cayley graph} has vertices the group $G$ and arrows of the form $x\to xa$ where $a\in \CC\subseteq G\setminus\{e\}$ a set of generators. In the finite group case case 
\[ e_a=\sum_{x\in G}e_{x\to xa}\in \Omega^1 \]
are the associated left-invariant differential forms and these form an basis {\em over the algebra $\C(G)$} (so all $\Omega^1$ is obtained by acting on the left, say, by $\C(G)$). The reason these are called left-invariant is that on a Cayley graph the left translation $G$-action $\la$ on $\C(G)$ given by $x \la \delta_{y} = \delta_{xy}$ extends to  arrows of the allowed form $e_{y \to ya}$ in a way that commutes with the differential
\[
x\la e_{y \to y a} = x \la (\delta_{y}\extd \delta_{ya}) = (x\la \delta_{y}) \extd (x\la \delta_{ya}) = e_{x y\to x y a},\]
so that  $e_a$, $a\in \CC$ are indeed invariant under this left action. If $\CC$ is stable under conjugation then the right translation action $\delta_{y}\ra x = \delta_{y x}$ also extends to one forms. This is because we need
\[ (\delta_x\extd \delta_{xa})\ra y=\delta_{xy}\extd\delta_{xay}=\delta_{xy}\extd \delta_{xy (y^{-1}ay)}\]
so the arrows remain given by right translation by elements of $\CC$ precisely of $\CC$ is preserved under conjugation. We say in this case that $\Omega^1$ is {\em bicovariant}, and in this case it follows that $e_{a} \ra x = e_{x^{-1} a x}$ for the right action. The bimodule relations for product by functions from the right, $\extd$, $*$ and inner structure  are 
\[
e_a.f = (R_a f) e_a,\quad \extd f = \sum_{a\in \CC} (R_a f - f) e_a, \quad   e^*_a = -e_{a^{-1}},\quad \theta = \sum_{a\in \CC} e_a,
\]
for $f\in \C(G)$ and $(R_a f)(g) = f(ga)$ the right translation by $a\in \CC$. In this form, the calculus makes sense even if $G$ is infinite, as long as $\CC$ is finite (i.e. $\Omega^1$ finite-dimensional over the algebra $B=\C(G)$). 

The vector space spanned by  $\{e_a\}$ is denoted $\Lambda^1$, the left-invariant differential forms with respect to multiplication in the group. As in Lie theory, $\Lambda^1{}^*$ is a `discrete group Lie algebra' of some kind (in this case, a rack) with the dual basis $f^a$ labelled again by $a\in \CC$ and acting by the `left invariant vector fields' $\del_a=R_a-\id$ as in $\extd$. One can also think of $\Lambda^1$ as a quotient of $\C(G)^+$ (the latter being the functions that vanish at $e$) and hence $\Lambda^1{}^*$ a subspace of $\C G$, namely with basis $f^a=a-e$. 

In the bicovariant case there is a canonical choice of exterior algebra $\Omega_{wor}$ which is a quotient of $\Omega_{min}$, see \cite{BegMa} (or could equal it in some cases) defined by a braiding map
\begin{equation*}
    \Psi\colon \Lambda^1 \tens \Lambda^1 \to \Lambda^1\tens \Lambda^1,\quad
    \Psi(e_a \tens e_b) = e_{aba^{-1}} \tens e_{a}.\end{equation*}
We skew symmetrize the tensor algebra on $\Lambda^1$ in a certain way, which in degree 2 amounts to setting $\ker(\Psi-\id)$ to zero, giving an algebra $\Lambda$ of left-invariant forms and $\Omega=\C(G)\Lambda$ is spanned by these over the function algebra. For $G$ an Abelian group with a Cayley graph $\CC$, the adjoint action is trivial and we have $\Psi(e_a \tens e_b) = e_b \tens e_a$, so that the kernel of $\Psi-\id$ is spanned by $e_a\tens e_a$ and $e_a\tens e_b+e_b\tens e_a$. We set these to be zero after the wedge product, making the $\Lambda$ generated by the left-invariant 1-forms a Grassmannian algebra as in classical differential geometry. Note that general 1-forms, however, need not anticommute.
 
Also in the case a Cayley graph a generalised quantum metric takes the form \cite[Proposition 1.58]{BegMa}
\[
\cg = \sum_{a \in \CC} c_a e_a \tens_{\C(G)} e_{a^{-1}}, \quad (e_a,e_b) = \frac{\delta_{a,b^{-1}}}{R_{a}(c_{a^{-1}})}
\]
with real nonvanishing coefficients $c_a \in \C(G)$. For a quantum metric one should impose quantum symmetry $\wedge \cg = 0$, and here the metric is edge-symmetric if and only if $R_a(c_{a^{-1}})=c_a$.

A natural Hodge operator has been proposed  in \cite{Ma:hod}  on any quantum group with bicovariant calculus and the additional data of an bi-invariant quantum metric and unique central bi-invariant top form, but is quite involved. On the other hand, there is also an earlier ad-hoc framework \cite{MaRai, BegMa} which merely mimics classical formulae under similar assumptions. For this, one assumes that $\Omega^n = \C(G)\Vol$, i.e. is one dimensional over the algebra with basis a central volume form $\Vol$. From this, one define the `anti-symmetric tensor' $\epsilon_{a_1 \dots a_n}$ via
\[
e_{a_1 }\wedge \dots \wedge e_{a_n} = \epsilon_{a_1 \dots a_n} \Vol
\]
which vanishes whenever $a_1 \cdots a_n \neq e$ as $\Vol$ is central. The Hodge star is then defined for the Euclidean metric $\cg^{a,b} = \delta_{a,b^{-1}} = (e_a,e_b)$ as the bimodule map $\star\colon \Omega^{k} \to \Omega^{n-k}$ which acts as
\begin{align*}
\star(e_{a_1}\wedge \dots \wedge e_{a_k} ) 
&= \sum_{bc} d_k^{-1} \epsilon_{a_1\cdots a_k b_{k+1} \cdots b_n}  \cg^{b_{k+1}c_{k+1}} \cdots \cg^{b_n c_n} e_{c_n} \wedge \cdots \wedge e_{c_{k+1}}
= \sum_{\{b_i\}} d_k^{-1} \epsilon_{a_1\cdots a_k b_{k+1} \cdots b_n} e_{b^{-1}_n} \wedge \cdots \wedge e_{b^{-1}_{k+1}}
\end{align*}
for some normalisation constants $d_k$ such that $\star^2 = \pm 1$. We then have $(e_a \wedge e_b , e_c \wedge e_d)^\star_2$ defined by
\[
	\star(e_a\wedge e_b)\wedge e_c\wedge e_d
	= \sum_{\{b_i\}} d_2^{-1} \epsilon_{a b b_{3} \cdots b_n}  e_{b^{-1}_n} \wedge \cdots \wedge e_{b^{-1}_{3}}\wedge e_c\wedge e_d
	= \sum_{\{b_i\}} d_2^{-1}\epsilon_{a b b_{3} \cdots b_n} \epsilon_{b^{-1}_n\cdots b^{-1}_{3} cd}\Vol
\]
which we equate to $(e_a\wedge e_b,e_c\wedge e_d)^\star_2\Vol$, giving
\begin{equation*}
(e_a\wedge e_b,e_c\wedge e_d)^\star_2=\sum_{\{b_i\}} d_2^{-1} \epsilon_{a b b_{3} \cdots b_n} \epsilon_{b^{-1}_n\cdots b^{-1}_{3} cd}.
\end{equation*}

An alternative way of defining $(\ , \ )_2$ on a discrete group $G$ with Euclidean metric is via the braiding $\Psi$ and its inverse $\Psi^{-1}(e_a\tens e_b)= e_{b} \tens e_{b^{-1}ab}$
\begin{align}
\label{eq:(,)2G}
(e_a \wedge e_b, e_c\wedge e_d)^\psi_2 &= ((\Psi + \Psi^{-1} - 2\id) e_a \tens e_b, e_c\tens e_d)_\tens = (e_a \tens e_b, (\Psi + \Psi^{-1} - 2\id)  e_c\tens e_d)_\tens \nonumber\\
&= \cg^{a,c} \cg^{b,cdc^{-1}} + \cg^{a,d^{-1}cd} \cg^{b,d} - 2\cg^{a,d} \cg^{b,c}.
\end{align}
Here $(\ ,\ )_\tens$ denotes the extension of the metric to $\Lambda^1\tens\Lambda^1$ in a nested manner, giving the result as shown, and we have used $((\Psi+\Psi^{-1})\omega, \eta)_\tens = (\omega, (\Psi+\Psi^{-1}) \eta)_\tens $ which can be checked directly. This is  manifestly well-defined as rewriting $(\Psi + \Psi^{-1} - 2\id) = (\id - \Psi^{-1})(\Psi-\id)$ makes it explicit that $(\ ,\ )_2$ vanishes on $\ker \Psi-\id$. Furthermore, it satisfies the reality condition, which it would not if we excluded the $\Psi^{-1}-\id$ term.

\begin{proposition}
\label{prop:(,)psi=(,)star}
For $G$ a discrete Abelian group with $\CC = \{c_1,\dots,c_n\}$ or $G= S_3$ with $\CC = \{u=(12),v=(23),w=(13)\}$, both with the Euclidean metric and a central volume form, we have that $(\ , \ )^\psi_2$ and $(\ , \ )^\star_2$ coincide up to an overall normalisation.
\end{proposition}
\begin{proof}
As noted before the exterior algebra $\Omega_{wor}$ of $G$ a discrete Abelian group is given by the Grassmann algebra generated by $e_{c_i}$. We set $\Vol = e_{c_1} \wedge  \dots \wedge e_{c_n}$ and $\epsilon_{c_1,\dots,c_n} = 1$, with even permutations of the indices giving $+1$ and odd permutations $-1$. In this case the Hodge operator is similar to the classical case, with normalisations $d_k = (n-k)!$. These can be ignored though, as we can use the explicit formula for the antisymmetric tensor to write 
\[
\star (e_{c_{i_1}} \wedge \dots \wedge e_{c_{i_k}}) 
= \epsilon_{c_{i_1},  \dots, c_{i_k},c_{1}, \dots,\widehat{c}_{i_j},\dots,c_n} e_{c^{-1}_n} \wedge \dots \wedge \widehat{e_{c^{-1}_{i_j}}} \wedge \dots \wedge e_{c^{-1}_1},
\]
where hat notation means that we are omitting this index or term. The simplification happens as the permutation of $b_{k+1},\cdots,b_{n}$ in $\epsilon$ to get the antisymmetric to the form above will be of the same parity as the permutation used to bring $e_{b^{-1}_n}\wedge \dots \wedge e_{b^{-1}_1}$ to what is shown above. The sum has a total of $d_k$ terms, which cancels the $d^{-1}_k$. This simplifies $(\ , \ )^\star_2$ to
\[
\star(e_{c_i} \wedge e_{c_j}) \wedge e_{c_k} \wedge e_{c_l}
=  \epsilon_{c_i, c_j, c_{1},\dots,\widehat{c}_i, \widehat{c}_j,\dots,c_n} \epsilon_{c^{-1}_{n},\dots,\widehat{c^{-1}_i}, \widehat{c^{-1}_j},\dots,c^{-1}_1,c_k, c_l}\Vol.
\]
As the only non-zero components of $\epsilon$ are the ones that have all indices different from each other, we see that this expression is only non-zero for either $c_k = c^{-1}_i $ and $c_l = c^{-1}_j$ or $c_k = c^{-1}_j$ and $c_l = c^{-1}_i$. One finds 
\[
\epsilon_{c_i, c_j, c_1, \dots,\widehat{c_i}, \widehat{c_j},\dots,c_n} \epsilon_{c^{-1}_{n},\dots,\widehat{c^{-1}_i}, \widehat{c^{-1}_j},\dots,c^{-1}_1,c_k, c_l}  
= \delta_{c_i,c^{-1}_l}\delta_{c_j,c^{-1}_k} -\delta_{c_i,c^{-1}_k}\delta_{c_j,c^{-1}_l} = (\cg^{c_i,c_l}\cg^{c_j,c_k} - \cg^{c_i,c_k}\cg^{c_j,c_l})
\]
As such $(e_{c_i} \wedge e_{c_j}, \wedge e_{c_k} \wedge e_{c_l})^\star_2
=  (\cg^{c_i,c_l}\cg^{c_j,c_k} - \cg^{c_i,c_k}\cg^{c_j,c_l})$ which is the same as \eqref{eq:(,)2G} in the Abelian case up to an overall normalisation of $-2$.

For $G=S_3$ the permutation group on three elements  and $\CC = \{u,v,w\}$, the `antisymmetric' tensor is given by $\epsilon_{uvuw} = \epsilon_{vuvw} = -\epsilon_{wuvu} = -\epsilon_{wvuv} = 1$ and their cyclic rotations under $u\to v \to w \to u$, and $d_2 = \sqrt{3}$, see \cite{MaRai}. We see that $\epsilon_{abcd}$ is $1$ when $a=c$, $d = abc = aba$ and $-1$ when $b=d$, $a = bcd = bcb$, which can be concisely written as $\epsilon_{abcd} = \delta_{a,c} \delta_{d,aba} - \delta_{b,d} \delta_{a,bcb}$.
As $\Omega$ has volume dimension 4 and we have $a=a^{-1}$ for $a\in\CC$, then 
\[
	\star (e_a\wedge e_b) \wedge e_c\wedge e_d = \sum_{f_3,f_4\in \CC} d_2^{-1} \epsilon_{ab f_{3} f_4} \epsilon_{f_4 f_3 c d} \Vol.
\]
Using the above formula for the antisymmetric tensor in terms of kronecker deltas, expanding and summing over $f_3$, $f_4$ leads to 
\[
	\star (e_a\wedge e_b) \wedge e_c\wedge e_d = d_2^{-1} (\delta_{bab,d} \delta_{b,cdc} + \delta_{aba,c} \delta_{d,cac}  - \delta_{a,d} \delta_{aba,cdc} - \delta_{b,c} \delta_{d,a} 
	)
\]
Focusing on the first term, we have $\delta_{bab,d} \delta_{b,cdc} = \delta_{a,bdb} \delta_{b,cdc}  = \delta_{a,cdcdcdc} \delta_{b,cdc} =  \delta_{a,c} \delta_{b,cdc} = \cg^{a,c} \cg^{b,c d c}$
where the relations $a^2=e$, $aba = bab$ for any $a,b \in \CC$ were used. This corresponds to the first term of \eqref{eq:(,)2G}. The rest of the terms can be simplified in the same manner to recover the full expression, up to an overall normalisation of $d_2^{-1}$. \end{proof}

Motivated by Proposition \ref{prop:(,)psi=(,)star}, we will from now use $(\ , \ )_2=(\ , \ )^\psi_2$ as the `inverse metric' on 2-forms for finite groups, as it makes sense for any metric, not only the Euclidean one, and does not rely on a central volume form. We now see how this approach via equation  \eqref{eq:(,)2G} looks in terms of the general description of $(\ , \ )_2$ for any graph as in Lemma \ref{lem:InnerProductOmega2X}. Decomposing the left-invariant 1-forms and evaluating the LHS of equation \eqref{eq:(,)2G} on $x\in G$ gives
\begin{align*}
(e_a \wedge e_b,e_c \wedge e_d)_2(x) = \sum_{\tilde x, \tilde x, y, y'} (e_{\tilde x \to \tilde x a} \wedge e_{\tilde x a \to \tilde x ab }, e_{y\to yc} \wedge e_{yc \to ycd})_2(x)
= \lambda_{x\to xa \to xab \to xabc} \delta_{e,abcd}.
\end{align*}

To compute the RHS of equation \eqref{eq:(,)2G} note that $\cg^{a,b}(x) = \frac{\delta_{a,b^{-1}}}{c_{a^{-1}(xa)}} = \lambda_{x\to xa} \delta_{a,b^{-1}}$. Equating both sides and setting $d = (abc)^{-1}$ (otherwise equation \eqref{eq:(,)2G} is trivial) and simplifying the expressions for the Kronecker deltas we have
\begin{equation}
\label{eq:weightsGroupCase}
\lambda_{x\to xa \to xab \to xabc} = \lambda_{x\to xa} \lambda_{x\to xb} (\delta_{a,c^{-1}} + \delta_{c,b^{-1}a^{-1} b} - 2 \delta_{b,c^{-1}}).
\end{equation}
Consider now the diagram
\[
\begin{tikzcd}
{} & {xb}\arrow[dd,leftrightarrow] & {} &{xab} \arrow[ll,leftrightarrow] \arrow[ld,leftrightarrow]  & {} \\
{} & {} & {xaba^{-1}} \arrow[ld,leftrightarrow]  &{} & {} \\
{x a^{-1}}  \arrow[r,leftrightarrow] & {x}  \arrow[rr,leftrightarrow] & {} & {x a} \arrow[r,leftrightarrow] \arrow[uu,,leftrightarrow]   & {xa^2}
\end{tikzcd}
\]
and assume an Ad-invariant metric with $\lambda_{x\to xa} = \lambda_{x\to xgag^{-1}} $. Going through all the possibilities for $a$, $b$, $c$ we have the following cases\\
\\
$\underline{ab \neq ba}$:
\begin{itemize}
	\item[1)] $c=a^{-1}$, $c\neq b^{-1}a^{-1}b$, $c\neq b^{-1},\quad \lambda_{x\to xa \to xab \to xaba^{-1}} = \lambda_{x \to xa} \lambda_{x \to xaba^{-1}}, \quad \begin{tikzcd}
{} & {} & {xab}\arrow[ld] \\
{} &  {xaba^{-1}}\arrow[ld] & {} \\
{x}  \arrow[rr] & {} & {xa} \arrow[uu]
\end{tikzcd}$
	\item[2)] $c\neq a^{-1}$, $c= b^{-1}a^{-1}b$, $c\neq b^{-1}, \quad \lambda_{x\to xa \to xab \to xb} = \lambda_{x \to xa}\lambda_{x \to xb}, \quad \quad \quad  \quad
	\begin{tikzcd}
{xb} \arrow[dd]  & {} & {xab} \arrow[ll] \\
{} & {} & {} \\
{x} \arrow[rr] & {} & {xa}  \arrow[uu]
\end{tikzcd}
$
	\item[3)] $c\neq a^{-1}$, $c\neq b^{-1}a^{-1}b$, $c = b^{-1}, \quad \lambda_{x\to xa \to xab \to xa} = -2 \lambda_{x \to xa}\lambda_{x \to xb}, \quad \quad \quad  
	\begin{tikzcd}
{}  & {} & {xab} \arrow[dd,shift left] \\
{} & {} & {} \\
{x} \arrow[rr,shift left] & {} & {xa}  \arrow[uu,shift left] \arrow[ll,shift left]
\end{tikzcd}
$
\end{itemize}
$\underline{ab = ba}$:
\begin{itemize}
	\item[4)] $c=a^{-1}$, $c=b^{-1}a^{-1}b$, $c\neq b^{-1},\quad \lambda_{x\to xa \to xab \to xaba^{-1}} = 2 \lambda_{x \to xa} \lambda_{x \to xaba^{-1}},  
\begin{tikzcd}
{xb = xaba^{-1}} \arrow[dd]  & {} & {xab} \arrow[ll] \\
{} & {} & {} \\
{x} \arrow[rr] & {} & {xa}  \arrow[uu]
\end{tikzcd}
$
\end{itemize}
$\underline{a = b}$:
\begin{itemize}
	\item[5)] $c=a^{-1}, \quad \quad \quad \quad \quad \quad \quad \quad \quad \quad  \lambda_{x\to xa \to xa^2 \to xa} = 0, \quad \quad \quad \quad\quad    \quad \quad \quad \quad  
\begin{tikzcd}
{x} \arrow[r,shift left] & {xa} \arrow[r,shift left]  \arrow[l,shift left] & {x a^2}  \arrow[l, shift left]
\end{tikzcd}
$
\end{itemize}
$\underline{a = b^{-1}}$:
\begin{itemize}
	\item[6)] $c = a, \quad \quad\quad \quad \quad \quad \quad \quad \quad \quad \quad 
	\lambda_{x\to xa \to x \to xa} = -2 \lambda_{x \to xa}\lambda_{x \to xa^{-1}}, 
	\quad \quad \quad \quad  
\begin{tikzcd}
{x} \arrow[r,shift right]\arrow[r,shift left=3] & {xa} \arrow[l,shift right]\arrow[l,shift left=3]  
\end{tikzcd}
$
\item[7)] $c=a^{-1}, \quad \quad \quad \quad \quad \quad \quad \quad \quad \quad 
	\lambda_{x\to xa \to x \to xa} = 2 \lambda_{x \to xa}\lambda_{x \to xa^{-1}}, 
	\quad \quad \quad \quad  \quad
\begin{tikzcd}
{xa^{-1}} \arrow[r,shift right] & {x} \arrow[r,shift right]\arrow[l,shift right] & {xa} \arrow[l,shift right]
\end{tikzcd}
$
\end{itemize}

 Cases 1) and 2) are what we expect from the solution in Proposition~\ref{eq:SolLambda(,)2} due to $\lambda_{x\to xb} = \lambda_{x\to xaba^{-1}}$ with $A_{x,xab} = 1$, and 3) has factor -2 such that the conditions on $\lambda_{x\to y\to z\to y' }$ are satisfied. Case 4) has a 2 as this accounts for the loops 1) and 2) in the Abelian case. Cases 5) can be seen as the collapse of both 3) and 4) when $a=b$, hence why it vanishes.  Cases 6) and 7) can be seen as the collapse of 3) and 4) respectively, hence the factor -2 and 2. All in all, we have
\begin{equation*}
  A_{x,xab} =
    \begin{cases}
      1 \quad  \text{ for } ab \neq ba,\\
      2 \quad  \text{ for } ab = ba, a\neq b,\\
      0 \quad   \text{ for } a=b.
    \end{cases}       
\end{equation*} 

\section{Graph gauge theory with universal calculus on the gauge group}
\label{sec:GraphUniCalcGaugeGroup}

Gauge theory on tensor product bundles with any differential algebra $B$ as base and any quantum group fibre with universal calculus looks similar in structure to the noncommutative $U(1)$-gauge theory case and is covered, for example, in \cite[Chap 5]{BegMa}. We specialise to the case of a finite group Hopf algebra $\C G$ as `gauge quantum group', defined a having basis $G$ with their group product and the coalgebra structure $\Delta g=g\tens g, \epsilon g=1, S g=g^{-1}$ for all basis elements $g\in G$, all then extended linearly. We also have $g^*=g^{-1}$ extended anti-linearly. This passage from elements of $G$ to linear combinations is parallel to the passage from classical to quantum computing. 
We define $\C G^+=\ker\epsilon$ the subalgebra of counit 0. This is $|G|-1$-dimensional with basis $\{g-e\}$ for $g\in G\setminus\{e\}$. In these terms, the general set up reduces to 
\begin{align*}
    &\CA = \{\alpha \in \C G^+ \tens \Omega^1\ | \  \alpha^* = -\alpha \},&
	&\CG = \{\gamma\in  \C G \tens B\ |\ (\epsilon \tens \id) \gamma = 1,\  \gamma \gamma^* = \gamma^* \gamma = 1\}.
\end{align*}
The latter forms a group and acts on the set of the former by transforming a connection $\alpha$ to
\[ \alpha^\gamma=\gamma\alpha\gamma^* + \gamma (\id\tens \extd) \gamma^*.\]
The curvature of a connection $\alpha$ is likewise a `Lie-algebra'-valued two form defined as 
\[ F(\alpha) \in \C G^+ \tens \Omega^2,\quad F(\alpha) = (\id \tens \extd) \alpha + \alpha \wedge \alpha\]
and as such depends on the given exterior algebra $\Omega$ of $B$ to degree 2.  Here, the product $\alpha \wedge \alpha$ is inherited from the usual product on $\C G^+$ and the given wedge product on $\Omega$ and $F(\alpha)$ inherits the property $F(\alpha)^*=-F(\alpha)$ and transforms to $F(\alpha^\gamma)=\gamma F(\alpha)\gamma^*$. 

As a remark, there are several definitions for gauge transformations that one could take. For example in \cite{HanLan} the authors suggest that one should impose the extra condition that $\gamma \in \Hom(\C(G),B) \simeq \C G \tens B$ is an algebra map. We do not consider this restriction in the case of the universal calculus on $G$, but will need to restrict $\gamma$ to be a differential algebra map to make sense of gauge theory with a non-universal calculus on $G$ in Section \ref{secnonu}.  Finally, we will  be interested in a block decomposition of the above $\C G$-gauge theory as follows. 

\begin{proposition}\label{propdec} $\C G$-valued gauge theory for the universal calculus as above can be decomposed as a product of $U(d_\rho)$-gauge theories over base $B$ in the sense of Section~\ref{secUd} for the nontrivial unitary irreps $\rho\ne 1$ of $G$ with dimension $d_\rho$. 
\end{proposition} 
\proof We use the Peter-Weyl decomposition whereby 
\[  \C G\simeq\oplus_{\rho\in \hat G}M_{d_\rho}(\C)\] 
provided by an orthogonal set of central projections that add up to 1, 
\[ \pi_\rho={d_\rho\over |G|}\sum_{g\in G} {\rm Tr}(\rho(g^{-1}))g\in \C G;\quad \sum_{\rho\in\hat G}\pi_\rho=1,\quad \pi_\rho\pi_{\rho'}=\delta_{\rho,\rho'}\pi_\rho.\]
 The linearly extended $\rho:\C G\to M_{d_\rho}(\C)$ then restricts to an isomorphism $(\C G)\pi_\rho\simeq M_{d_\rho}(\C)$ and to zero on other components $(\C G)\pi_{\rho'}$. Going the other way, a matrix $a_{ij}$ maps back to the element 
\[ {d_\rho\over|G|}\sum_{g\in G}\rho(g^{-1})_{ji} a_{ij}\, g\]
 due to the orthogonality of matrix elements of irreps. Next, from the orthogonality of characters and that our irreps are unitary, one has respectively that
 \[ \epsilon(\pi_\rho)=\delta_{\rho, 1},\quad \pi_\rho^*=\pi_\rho\]
 with respect to the counit and $*$ operation on $\C G$. We now expand a
 general $\alpha\in \CA$ and $\gamma\in \CG$ as
\[ \alpha=\sum_{\rho\ne 1}\alpha_\rho, \quad \gamma=\pi_1+\sum_{\rho\ne 1} \gamma_\rho ,\]
where $\alpha_\rho=\alpha(\pi_\rho\tens 1)$ (i.e., inserting $\pi_\rho$ in the left factor) and $\gamma_\rho= \gamma(\pi_\rho \tens 1)$. There is no $\pi_1$ term for $\alpha$ so that $(\epsilon\tens\id)\alpha=0$, and the coefficient of $\pi_1$ in $\gamma$ is fixed by $(\epsilon\tens\id)\gamma=1$.  By the above properties of the $\pi_\rho$ under $*$, and using that they are a complete central set of orthogonal projections (so that the above decompositions are unique) it is easy to see that $\alpha^*=-\alpha$ and $\gamma\gamma^*=\gamma^*\gamma=1$ then translate to  $(\alpha_\rho)^*=-\alpha_\rho$ and $\gamma_\rho(\gamma_\rho)^*=(\gamma_\rho)^*\gamma_\rho=1$. Finally, each component maps over under the block decomposition to elements of $M_{d_\rho}(\C)\tens\Omega^1$ and $M_{d_\rho}(\C)\tens B$ respectively with corresponding properties under $*$, i.e. to the $U(d_\rho)$-gauge theory as at the start of Section~\ref{secUd}. \endproof 

This analysis works for any differential algebra $B$, but in the case where $B$ is functions on a set, the   $\gamma_\rho$ appear as $U(d_\rho)$-valued functions since we can analyse the unitarity requirement at each point. These can then be viewed together as valued in the product 
\[  \C G^u\simeq\prod_{\rho\ne 1} U(d_\rho).\]
In the continuum case, the $\alpha_\rho$ appear as usual 1-forms valued in the anti-hermitian elements of $M_{d_\rho}(\C)$, i.e. usual gauge fields for the relevant $U(d_\rho)$ factor.   In what follows, we will consider both descriptions of our fields for $\C G$-valued gauge theory, as labelled by $G\setminus\{e\}$ or by $\hat G\setminus \{1\}$, where $\hat G$ denotes the set of unitary irreps and 1 denotes the trivial irrep.

\subsection{Specialisation to the graph case}\label{secgraphgauge}

In our case $B=\C(X)$ is the space of functions on the vertices of a graph with basis $\delta_x$ and $\Omega^1=\C E$ spanned  by the arrows. Taking the coefficients in this basis, we can think of a 1-form as a function on the set $E$ of arrows.  So then
\[ \CA=\{\alpha\in C(E, \C G)\ |\  \epsilon(\alpha_{x\to y})=0,\quad \alpha_{x\to y}^*=\alpha_{y\to x}\}=\{\alpha\in C(E, \C G^+)\ |\  \alpha_{x\to y}^*=\alpha_{y\to x}\},\]
\[\CG=\{\gamma\in C(X,\C G)\ |\ \epsilon(\gamma_x)=1,\quad \gamma_x^*\gamma_x=\gamma_x\gamma_x^*=1\}\]
for all $x\in X$ and $x\to y\in E$. This is equivalent to our previous description if we write $\alpha=\sum_{x\to y}\alpha_{x\to y}\tens e_{x\to y}$ and $\gamma=\sum_x \gamma_x\tens \delta_x$. We see that in our function description,   
$\alpha$  at each arrow is valued in Hermitian elements of $\C G^+$, and from the point of view of NCG the latter is a basis of the finite-Lie algebra dual to the universal calculus on $G$. Similarly $\gamma_x$ at each vertex $x$ is valued in the group of unitary elements of $\C G$, which is the group
\[\C G^u \coloneqq \{h\in \C G\ |\ \epsilon(h)=1,\ h^* h=h h^*=e\}.\]
Then we can say more compactly that
\[ \CA=\{\alpha\in C(E, \C G^+)\ |\  \alpha_{x\to y}^*=\alpha_{y\to x}\},\quad \CG=C(X, \C G^u).\]
We can also expand the values of our fields as, in each case, $|G|-1$ independent coefficients
\[ \alpha_{x\to y}=\sum_{g\in G\setminus\{e\}} \alpha_{x\to y}^g (g-e),\quad \gamma_x=e+\sum_{g\in G\setminus\{e\}} \gamma_x^g (g-e); \quad \gamma_x^e:=1-\sum_{g\in G\setminus\{e\}}\gamma_x^g\]
subject to the Hermitian and unitarity requirements
\[ (\alpha_{x\to y}^g)^*=\alpha^{g^{-1}}_{y\to x},\quad \sum_{h\in G} \gamma^h_x (\gamma^{g^{-1}h}_x)^* = \sum_{h\in G} (\gamma^h_x)^* \gamma^{hg}_x = \delta_{e,g}, \quad \gamma_x^*=e+\sum_{g\in G\setminus\{e\}}(\gamma_x^{g^{-1}})^*(g-e). \]
For the curvature, what constitutes a basis of $\Omega^2$ depends on which exterior algebra we take. However, we can still write it as
\[ F(\alpha)=\sum_{x\to y\to z}F(\alpha)_{x\to y\to z}\tens e_{x\to y} \wedge e_{y\to z}\]
for a certain expression for $F(\alpha)_{x\to y\to z}$ for every 2-step path, even if this is not uniquely determined due to the set of 2-steps not being a basis.
 
\begin{proposition}\label{alphaguage} The  gauge transformation and curvature in component form can be written as
\[ \alpha^\gamma_{x\to y}=\gamma_{x} \alpha_{x\to y}  \gamma^*_{y} + \gamma_x  \gamma^*_y - e,\]
\[ F(\alpha)_{x\to y\to z}=\alpha_{x\to y}+\alpha_{y\to z}-\mu \alpha_{x\to z} + \alpha_{x\to y}\alpha_{y\to z}+(1-\mu)\delta_{x\ne z}e; \quad F(\alpha^\gamma)_{x\to y\to z}=\gamma_x F(\alpha)_{x\to y\to z}\gamma_z^*\]
where  $\mu = 1$ for $\Omega_{max}$, $\Omega_{med'}$ and zero for $\Omega_{med}$  and $\Omega_{min}$ and $\delta_{x\ne z}$ means 1 if $x\ne z$ and zero otherwise.   \end{proposition}
\begin{proof} This is a special case of quantum group gauge theory. To compute the coefficients $\alpha^\gamma_{x\to y}$ recall that $\delta_x \extd \delta_y = e_{x\to y}$ for $x\to y$ and 
\[
    \delta_x \extd \delta_x 
    = \delta_x \sum_{w\to z}\left( \delta_x(z) - \delta_x(w)\right)e_{w\to z}
    = -\sum_{y:x\to y} e_{x\to y}.
\]
Writing the action of $\gamma$ on $\alpha$ in components then gives
\[
   \alpha^\gamma
    = \gamma \alpha \gamma^* + \gamma (\id\tens \extd) \gamma^*
    = \sum_{x\to y} \gamma_{x} \alpha_{x\to y} \gamma^*_{y} \tens  e_{x\to y}
    + \sum_{x\to y} \gamma_x \gamma^*_y \tens  \delta_x \extd \delta_y = \sum_{x\to y} (\gamma_{x} \alpha_{x\to y}  \gamma^*_{y} + \gamma_x  \gamma^*_y - e) \tens e_{x\to y}.
\]
which translates to the stated formula for the coeffcients of $e_{x\to y}$. We separated the $\gamma(\id\tens\extd)\gamma^*$ term into a sum over arrows $x\to y$ or  $x=y$ or pairs of vertices ${x,y}$ with no arrows between them. The latter vanishe as $\delta_x \extd \delta_y = 0$ for $x\neq y$, $x \not \to y$. We then use the formulas for $\delta_x\extd \delta_y$, $\gamma_x \gamma^*_x = e$,  keeping in mind that $\sum_x \sum_{y:x\to y} = \sum_{x\to y}$. 

For the curvature, there is  no problem to write down expressions from the previous description concretely via
\begin{align*}
    \alpha \wedge \alpha &= \sum_{x\to y \to z} \alpha_{x\to y}\alpha_{y\to z} \tens  e_{x\to y}\wedge e_{y\to z},\\
     (\id\tens \extd)\alpha&=\sum_{x\to y\to z} (\alpha_{x\to y}+\alpha_{y\to z})\tens e_{x\to y}\wedge\ e_{y\to z}-\mu \sum_{\substack{x\to y\to z: \\ x\to z }} \alpha_{x\to z} \tens e_{x\to y} \wedge e_{y\to z},
\end{align*}
where $\sum_{\substack{x\to y\to z: \\ x\to z }} $ is a sum over all paths $x\to y\to z$ with $x\to z$.  This leads us to define 
\begin{equation}\label{Fbare} F_{bare}(\alpha)_{x\to y\to z}=\alpha_{x\to y}+\alpha_{y\to z}-\mu \alpha_{x\to z} + \alpha_{x\to y}\alpha_{y\to z}\end{equation}
but remember that these coefficients are not uniquely determined and we can add something that does not change $F(\alpha)$ for any of the four calculi.  We now compute this using the formula for $\alpha^\gamma_{x\to y}$ and find that $F_{bare}(\alpha^\gamma)_{x\to y\to z}=\gamma_xF_{bare}(\alpha)_{x\to y\to z}\gamma_z^*+ (1-\mu)(\gamma_x\gamma_z^*-e)$. (This extra term does not affect the 2-form and we still have $F(\alpha^\gamma)=\gamma F(\alpha)\gamma^*$ as we must do by general theory). In light of this, we add a compensating `covariantising term' as in the statement of the proposition, which does not change $F(\alpha)$ for any of the calculi but ensures that $F(\alpha)_{x\to y\to z}$ transforms as expected under a gauge transformation. \end{proof}

It is easy to see that 
\[ F(\alpha)_{x\to y\to z}^*= F(\alpha)_{z\to y\to x},\quad \epsilon(F_{bare}(\alpha)_{x\to y\to z})=0,\quad \epsilon(F(\alpha)_{x\to y\to z})=(1-\mu)\delta_{x\ne z}, \]
where the first implies $F(\alpha)^*=-F(\alpha)$ as it must (remembering the extra minus sign for $*$ on a product of 1-forms). We also see that the bare coefficents are better behaved from the point of $\epsilon$ for $\Omega_{med},\Omega_{min}$ and that either way $F(\alpha)\in \C G^+\tens\Omega^2$ as it must (as a `Lie algebra' valued 2-form for the universal calculus on $G$). 

For our other description of the fields,  Proposition~\ref{propdec} tells us that  connections and gauge transformations have an expansion by $\rho$ according to 
\begin{equation}\label{rhodec} \alpha_{x\to y}=\sum_{\rho\ne 1}\alpha_{x\to y}\pi_\rho,\quad \gamma_x=\pi_1+\sum_{\rho\ne 1} \gamma_x\pi_\rho \end{equation}
using the projections $\pi_\rho$.  We view the $\rho$ components  $\alpha_{x\to y}\pi_\rho$ and $\gamma_x\pi_\rho$ for each $\rho\ne 1$ as matrices $\alpha_{x
 \to y}^\rho$ and $\gamma_x^\rho$ respectively by the block decomposition. The properties with respect to $*$ the become
 \[ (\alpha_{x\to y}^\rho)^* = \alpha^\rho_{y\to x},\quad \gamma_x^\rho(\gamma_x^\rho)^* =(\gamma_x^\rho)^*\gamma_x^\rho=1\]
 as a noncommutative $U(d_\rho)$-gauge theory for each $\rho\ne 1$ (where the calculus is noncommutative).  The case $G=\Z_2$ is simple enough for us to see how this works more explicitly. 

\begin{example} For $G=\Z_2=\{e,g\}$, where $g^2=e$, we have only one nontrivial irrep $\rho(g)=-1$. This is of dimension 1, so $\C \Z_2^u=U(1)$ and $\C\Z_2$-valued gauge theory can be identified with noncommutative $U(1)$-gauge theory specialised to a graph. To see more explicitly how the identification  of the two gauge theories works, we start with the $\C \Z_2$ point of view where 
\[ \gamma_x=e+ \gamma_x^g (g-e),\quad \alpha_{x\to y}=\alpha_{x\to y}^g(g-e)\]
since $|G|=2$ so $\C G^+$ is 1-dimensional and spanned by $g-e$. Then using $g^*=g^{-1}=g$, the requirement   $\gamma_x^*\gamma_x=\gamma_x\gamma_x^*=e$  corresponds to the condition  
\[ |\gamma_x^g|^2= {\rm Re}(\gamma_x^g)\]
One can solve this in terms of a phase at each $x$, because in the case of an Abelian group as here, we have $\C \Z_2\simeq \C(\Z_2)$ as Hopf $*$-algebras by Fourier transform and, working in $\C(\Z_2)$, unitary elements are a pair of phases; the counit   condition fixes one of these to be 1, leaving the other phase as a free angle. For $\alpha$,  the requirement is  $(\alpha^g_{x\to y})^*=\alpha^g_{y\to x}$. By contrast, these same elements can be expanded differently with one non-trivial representation $\rho$. This and the relation between the two sets of coefficients are 
\[ \alpha_{x\to y}=\alpha^\rho_{x\to y}\pi_{-1},\quad \gamma_x=\pi_1+\gamma_x^\rho\pi_{-1};\quad \alpha_{x\to y}^\rho=-2\alpha_{x\to y}^g,\quad \gamma_x^\rho=1-2\gamma^g_x.\]
One can check that the condition on the real part of $\gamma^g_x$ translates to $|\gamma_x^\rho|=1$, i.e. $\gamma_x^\rho=e^{\imath\theta_x}$ for some phase angle at each point. We also have $(\alpha_{x\to y}^\rho)^*=\alpha_{y\to x}^\rho$ of the same form as for $\alpha_{x\to y}^g$. If we let $\tilde\alpha=\sum_{x\to y}\alpha^\rho_{x\to y}e_{x\to y}\in \Omega^1$ be the corresponding 1-form from the point of view of noncommutative $U(1)$-gauge theory, then this corresponds to $\tilde\alpha^*=-\tilde\alpha$ as expected there. Finally, under a gauge transformation of $\alpha$ in the $\C\Z_2$-gauge theory,  
\[ \alpha_{x\to y}^\gamma= (\pi_1+e^{\imath\theta_x} \pi_{-1})\alpha^\rho_{x\to y}\pi_{-1}(\pi_1+e^{-\imath\theta_y} \pi_{-1})  + (\pi_1+e^{\imath\theta_x} \pi_{-1})(\pi_1+e^{-\imath\theta_y} \pi_{-1})-e=  (  e^{\imath(\theta_x-\theta_y)}\alpha^\rho_{x\to y}+ e^{\imath(\theta_x-\theta_y)}-1)\pi_{-1},\]
which tells us how the coefficient $\alpha_{x\to y}^\rho$ transforms.  This is exactly what we would have from noncommutative $U(1)$-gauge theory for the function $\tilde\gamma=e^{\imath\theta}$ as the corresponding gauge transform, where $\tilde\alpha^{\tilde\gamma}=\tilde\gamma\tilde\alpha{\tilde\gamma}^*+\tilde\gamma\extd {\tilde\gamma}^*$ implies that the coefficient of $e_{x\to y}$ transforms to  
\[ \tilde\alpha^{\tilde\gamma}_{x\to y}= e^{\imath(\theta_x-\theta_y)}\tilde\alpha_{x\to y}+ e^{\imath(\theta_x-\theta_y)}-1.\]
\end{example}

\subsection{Extended holonomy reformulation}
\label{sechol}

Motivated by LGT, we now adapt our description of $\CA$ in a manner related to the holonomy along each arrow as equivalent data. Thus, the \emph{extended holonomy} equivalent to a connection $\alpha$ is defined from the tensor and the `function' point of view respectively as
\begin{align*}
    u \coloneqq e\tens \theta + \alpha,\quad u_{x\to y}:=e+ \alpha_{x\to y}
\end{align*}
with $\theta = \sum_{x\to y\in E} e_{x\to y}$ is the inner element of $\Omega^1$. In classical geometry the holonomy along a path is given by the path ordered integral of the exponential of the integral of the connection along that path. Here, we see that the components $u_{x\to y}$ are that expression up to first order along the arrow $x\to y$. Since $\alpha_{x\to y}\tens  e_{x\to y}$ is already a 1-form, we do not need a metric or length of an arrow to parallel transport along it. The space of holonomies is then
\[     \CU = \{u\in \C G \tens \Omega^1\  \vert\  u^* = -u,\  (\epsilon\tens \id) u = \theta\}=\{u\in C(E,\C G)\ |\ \epsilon(u_{x\to y})=1,\  u_{x\to y}^*=u_{y\to x}\}\] 
for all $x\to y$, and is bijective to $\CA$. Here, the counit and reality conditions are derived from those of $\alpha$. 

The gauge transformation of the $u$ is derived from that of $\alpha$ and comes out as
\[ u^\gamma =  e\tens \theta + \alpha^\gamma =  \gamma u \gamma^* , \quad u^\gamma_{x\to y}=  \gamma_x u_{x\to y}\gamma^*_y\] 
 which is as one would expect for holonomy, and simpler (but equivalent) to working with $\alpha$ directly. 
 In the first version we use that $\extd \gamma^*=[\theta, \gamma^*]$ as the calculus is inner to write $\alpha^\gamma = \gamma (e\tens \theta + \alpha)\gamma^* - e\tens \theta$ from which the above follows, while for the coordinate form one can also use Proposition~\ref{alphaguage}.  We can also expand
  \[ u_{x\to y}=e+\sum_{g\in G\setminus\{e\}} u^g_{x\to y}(g-e),\quad u^e:=1-\sum_{g\in G\setminus\{e\}}u_{x
 \to y}^g,\quad (u^g_{x\to y})^*=u^{g^{-1}}_{y\to x}\]
 parallel to our description of $\gamma$, but now with the Hermitian requirement as stated. A special case is $u_{x\to y}\in G$, which corresponds to just one of the $u_{x\to y}^g=1$ (including the case $u_{x\to y}^e=1$) and all the rest zero. This special case corresponds to LGT and the difference more generally is that we allow linear combinations. We can also expand with respect to the Peter-Weyl decomposition as 
 \[ u_{x\to y}=\pi_1+ \sum_{\rho\ne 1}u^\rho_{x\to y}\pi_\rho,\quad u^\rho_{x\to y}\in M_{d_\rho}(\C),\quad (u^\rho_{x\to y})^\dagger=u^\rho_{y\to x}.\]
 
\begin{proposition}\label{prop:uFu}
The curvature of a connection $\alpha \in \CA$ with extended holonomy $u \in \CU$ is
\[F(\alpha)= u\wedge u - e\tens \theta \wedge \theta - \mu \sum_{\substack{x\to y\to z: \\ x\to z}} (u_{x\to y}-e)\tens  e_{x\to y}\wedge e_{y\to z}
\]
with the further simplification that $\theta\wedge\theta=0$ in $\Omega_{min}$. In component form $F(\alpha)=\sum_{x\to y\to z} F^u_{x\to y\to z}\tens e_{x\to y}\wedge e_{y\to z}$ with
\[ F^u_{x\to y\to z}:= u_{x\to y}u_{y\to z} -\mu u_{x\to z}-\delta_{x,z}e;\quad F^{u^\gamma}_{x\to y\to z}=\gamma_x F^u_{x\to y\to z}\gamma_z^*,\]
where we omit $u_{x\to z}$ if there is no arrow $x\to z$.  
\end{proposition}
\begin{proof} The tensor level proof uses an anticommutator $\extd\omega=\{\theta,\omega\}$ for any 1-form $\theta$ and the relation between $u$ and $\alpha$. For the component level formula we the bare formula (\ref{Fbare}) in the  Proposition~\ref{alphaguage} and the change from $\alpha_{x\to y}$ to $u_{x\to y}$ to obtain
\begin{equation*} F_{bare}^u{}_{x\to y\to z}=u_{x\to y}u_{y\to z}-e-\mu(u_{x\to z}-\delta_{x\to z}e),    \end{equation*}
where $\delta_{x\to z}$ is 1 if $x\to z$ and zero otherwise and we likewise omit $u_{x\to z}$ if there is no arrow $x\to z$.  This immediately agrees with the tensor calculation but is not covariant. Rather, we find $F^{u^\gamma}_{bare}{}_{x\to y\to z}=\gamma_x F^u_{bare}{}_{x\to y\to z}\gamma_z^*+ \gamma_x\gamma_z^*-e-\mu ( \gamma_x\gamma_z^*-e)\delta_{x\to z}$ and we compensate for this by adding a term to give
\begin{align*} F^u_{x\to y\to z}&\coloneqq F^u_{bare}{}_{x\to y\to z}+ \delta_{x\ne z} (1-\mu\delta_{x\to z})e\\
&= u_{x\to y}u_{y\to z}-e-\mu u_{x\to z} + \mu \delta_{x\to z}e+\delta_{x\ne z} (1-\mu\delta_{x\to z})e\\
&=u_{x\to y}u_{y\to z}-\mu u_{x\to z} -(1-\delta_{x\ne z})(1-\mu\delta_{x\to z}) e\\
&=u_{x\to y}u_{y\to z}-\mu u_{x\to z} -\delta_{x,z} (1-\delta_{x\to z}) e\end{align*}
which simplifies as stated since if $x=z$, there are no arrows $x\to z$. This then transforms as expected and gives the same $F(\alpha)$ as our bare formula. 
\end{proof}

It is also easy to see that
\[ F^u_{x\to y\to z}{}^*=F^u_{z\to y\to x},\quad \epsilon(F^u_{bare}{}_{x\to y\to z})=0,\quad \epsilon(F^u_{x\to y\to z})=\delta_{x\ne z}-\mu\delta_{x\to z}\]
so that the bare coefficients are again better from the point of view of $\epsilon$. The last expression is not the same as we got for $\epsilon(F(\alpha)_{x\to y\to z})$ because the additional covariantising terms are different in the two cases. Remember that these functions are not uniquely determined as 2-steps are not a basis of $\Omega^2$. 

Let us now better understand what the curvature measures in the different choices of exterior algebra. Starting with the simplest case, the component expression for $F_{min}$ is\begin{align*}
    F_{min} = \sum_{x\to y \to z} u_{x\to y}u_{y\to z} \tens e_{x\to y}\wedge e_{y\to z},
\end{align*}
but the 2 forms $e_{x\to y}\wedge e_{y\to z}$ are not independent due to the relations \eqref{eq:RelationsExtAlgGraph}, which need to be taken into account. For any two vertices $x$, $z$, suppose there are $n$ possible vertices $y_i$ making a two-step path $x\to y_i \to z$, with $n=0$ if there are none. The relations \eqref{eq:RelationsExtAlgGraph} for $\Omega_{min}$ are valid for any $x$, $z$, and we can use them to write
\begin{align*}
    e_{x\to y_n} \wedge e_{y_n\to z} = - \sum^{n-1}_{i=1} e_{x\to y_i} \wedge e_{y_i\to z},
\end{align*}
leading to
\begin{align}
\label{eq:FminComponents}
	F_{min} = \sum_{x, z} \sum^{n-1}_{i=1} (u_{x\to y_i}u_{y_i\to z} - u_{x\to y_n} u_{y_n\to z}) \tens e_{x\to y_i}\wedge e_{y_i\to z}.
\end{align}
For $n=1$ we simply have $e_{x\to y_1} \wedge e_{y_1\to z} = 0$. Given the following diagram
\[
\begin{tikzcd}
 & y_1 \arrow[rd,shift left]\arrow[ld,shift left] \arrow[d,phantom,"\vdots"] & \\
x \arrow[rd,shift left]\arrow[ru,shift left ]\arrow[r,shift left] & y_i \arrow[l,shift left] \arrow[r,shift left]  \arrow[d,phantom,"\vdots"]  & z \arrow[l,shift left]\arrow[ld,shift left] \arrow[lu,shift left]\\
  & y_n \arrow[ru,shift left] \arrow[lu,shift left] & \\
\end{tikzcd}
\]
we see that the components of the curvature $u_{x\to y_i}u_{y_i\to z} - u_{x\to y_n} u_{y_n\to z}$ measure the difference of the holonomy along the path  $x\to y_i \to z$ with the reference path $x\to y_n \to z$. Of course, we can also compare the holonomies along any two different $x\to y_i \to z$, $x\to y_j \to z$ via 
\[
(u_{x\to y_i}u_{y_i\to z} - u_{x\to y_n} u_{y_n\to z}) - (u_{x\to y_j}u_{y_j\to z} - u_{x\to y_n} u_{y_n\to z}) = u_{x\to y_i}u_{y_i\to z} - u_{x\to y_j} u_{y_j\to z}.
\]
As the relations in $\Omega_{min}$ also include the possibility $x=z$, the curvature in this case will also measure how the holonomies along $x\to y_i \to x$ between $x$ and all the points $y_i$ it is connected to differ, i.e. of the star shaped sub-graph
\[
\begin{tikzcd}
y_1 \arrow[rd,phantom,"\ddots"] \arrow[d,shift left] & \\
x \arrow[r,shift left] \arrow[d,shift left] \arrow[u,shift left]& y_i\arrow[l,shift left]\\
y_n \arrow[u,shift left]\arrow[ru,phantom,"\reflectbox{$\ddots$}"] & \\
\end{tikzcd}
\]
$F_{min}$ therefore measures the difference of holonomies along the sides of squares $x\to y_i \to z \to y_j \to x$, and the difference of holonomies along the edges of star shaped sub-graphs.

The case of $\Omega_{med}$ is similar. The argument goes as above, but with the difference that the relations \eqref{eq:RelationsExtAlgGraph} are only valid for $x\neq z$, and we have the extra term $-e\tens \theta \wedge \theta$ in the curvature. In components, we can write
\begin{align*}
	F_{med} 
    &= \sum_{x\to y\to z} (u_{x\to y}u_{y\to z} - e) \tens e_{x\to y}\wedge e_{y\to z}\\
    &= \sum_{\substack{x, z \\x\neq z}} \sum^{n-1}_{i=1} (u_{x\to y_i}u_{y_i\to z} - u_{x\to y_n} u_{y_n\to z}) \tens e_{x\to i}\wedge e_{i\to z}
    + \sum_{x\to y} (u_{x\to y}u_{y\to x} - e) \tens e_{x\to y}\wedge e_{y\to x}.
\end{align*}
As with $\Omega_{min}$ this expression measures the difference of the holonomy taken along the two different sides of a square, but in the case of a star shaped sub-graph it simply measures how the holonomy along $x\to y\to x$ differ from the identity $e$, instead of comparing the holonomies back and forth along different legs of the star shape subgraph.

For $\Omega_{med'}$ we can write
\begin{align*}
    F_{med'} = \sum_{\substack{x\to y \to z: \\ x\not \to z}} (u_{x\to y}u_{y\to z}-e)\tens e_{x\to y}\wedge e_{y\to z}
    + \sum_{\substack{x\to y \to z: \\ x \to z}} (u_{x\to y}u_{y\to z}-u_{x\to z})\tens e_{x\to y}\wedge e_{y\to z}
\end{align*}
and the relations \eqref{eq:RelationsExtAlgGraph} are taken over $x\not \to z$. Following the previous reasoning, the first term will compare the holonomies along the sides of a square between $x$, $z$, as long as the square in question does not have a diagonal arrow $x\to z$. Furthermore, it will compare the holonomies along the legs of a star shaped subgraph to one another, as the relations include the case $x=z$. The last term of $F_{med'}$ will clearly compare the holonomies along two different sides of a triangle $x\to y\to z$ and $x\to z$
\[
\begin{tikzcd}
 & y \arrow[rd,shift left] \arrow[ld,shift left] & \\
x \arrow[ru,shift left] \arrow[rr,shift left] & & z \arrow[lu,shift left] \arrow[ll,shift left]
\end{tikzcd}
\]

The case of $F_{max}$ is similar, but now as the condition for the relations are $x\not\to z$, $x\neq z$, it will again measure the difference of holonomies along squares and triangles, but will compare the holonomies along legs of a star shaped subgraph to the identity, instead of to each other. A summary of what the curvature measures in the different choices of exterior algebra can be found in the following table:
\begin{table}[h]
    \centering
    \begin{tabular}{c|c|c|c}
         &  Starshaped subgraphs & Squares & Triangles\\
          \hline
         $F_{min}$& compares legs to each other & yes & no\\
          \hline
         $F_{med}$& compares legs to $e$ & yes & no\\
          \hline
         $F_{med'}$& compares legs to each other & yes & yes\\
          \hline
         $F_{max}$& compares legs to $e$ & yes & yes\\
    \end{tabular}
\end{table}

In \cite{Gho}, the calculus $\Omega_{med'}$ is used to define the Hopf algebroid of differential operators on a graph and to show how this can be used to recover a certain `fundamental group' of a graph as defined via equivalence under square and triangle moves in \cite{GJM}.

\subsection{Yang-Mills action}  \label{secwilson}

Given a connection with extended holonomy $u$, a path $x_0 \to x_1 \to \dots \to x_{n-1} \to x_n$ on the graph $X$ and a representation $\rho$ of the group $G$, we define the \emph{Wilson line} of $u$ along $x_0 \to x_1 \to \dots \to x_{n-1} \to x_n$ in the $\rho$-representation to be
\begin{align*}
    W_{\rho,x_0 \to x_1 \to \dots \to x_{n-1} \to x_n}(u) \coloneqq \chi_\rho(u_{x_0\to x_1}u _{x_1\to x_2}\cdots u_{x_{n-1} \to x_n}),
\end{align*}
where $\chi_\rho:=d_\rho^{-1}{\rm Tr}\circ\rho$  is the normalised character of $\rho$ extended linearly to $\C G$, where our products are now valued. Of particular interest are the \emph{Wilson loops}, where we take the path to be closed with $x_0 = x_n$, as these objects are gauge invariant. Note that $\C G$ has a canonical trace in its left regular representation which can be written as $\sum_\rho d_\rho ^2\chi_\rho$ and $W_\rho$ just picks off the part of this relevant to the noncommutative $U(d_\rho)$-gauge theory in the $\rho$-decomposition (\ref{rhodec}). Only $\rho\ne 1$ are relevant as $\rho_1=\epsilon$ has value 1 on the $u_{x\to y}$. 

To construct the Yang-Mills action functional $S[u]$ we also need an `integration' and here we take
\[ \int_X f=\sum_{x\in X}\mu_x f(x)\]
for $f\in \C(X)$ with respect to a positive measure $\mu_x$, which ideally should come from the graph QRG (so that it is compatible with the noncommutative geometric divergence as used in the geometric Laplacian $\Delta = \delta \extd$  in \cite{BegMa}, where $\delta = (\, ,\, ) \circ \nabla$ for $\nabla$ a connection on $\Omega^1$). However, we will mostly be taking a quantum metric $(\ ,\ )_2$ directly on $\Omega^2$ and such a measure can be absorbed into that, so without loss of generality for our purposes here, we  just take counting measure $\mu_x=1$. We then define
\[ S[u]\coloneqq \int_X \sum_{\rho\ne 1} \frac{1}{c_\rho} \chi_\rho(  (F(\alpha)^*,F(\alpha))_2)=-\sum_{\rho\ne 1}\frac{1}{c_\rho} \sum_{\substack{ x\to y\to z \\ z'\to y'\to x'}}\chi_\rho(F^u_{x\to y\to z} F^u_{z'\to y'\to x'})\int_X (e_{x\to y}\wedge e_{y\to z},e_{z'\to y'}\wedge e_{y'\to x'})_2\]
where $(\ ,\ )_2$ is the assumed inner product on 2-forms and we then apply $\chi_\rho$ to the product of the $\C G^+$ values. 
The $\rho=1$ case means to apply $\epsilon$ and is excluded since the result does not depend on $u$ for the same reason as for $W_\rho$. We also used that $F(\alpha)$ is anti-Hermitian.  We have put weights $c_\rho$ as this is customary in LGT, but in our point of view these are just independent coupling constants that we are free to choose for each component noncommutative $U(d_\rho)$-gauge theory.

\begin{theorem}\label{thm:Wact}
The Yang-Mills action functional is gauge invariant $S[u^\gamma] = S[u]$ and has the following forms in terms of Wilson loops for the different exterior algebras on $X$: 
\begin{align*}
    S^{min}[u] &= - \sum_{\substack{ \rho\ne 1, \\ x\to y \to z \to y' \to x}} \frac{1}{c_\rho}  \lambda_{x\to y \to z \to y'}W_{\rho, x\to y \to z \to y' \to x}(u),\\
  S^{med}[u] &=S^{min}[u]+A[u],\quad  S^{med'}[u]= S^{min}[u]+ B[u],\quad S^{max}[u] = S^{min}[u]+A[u]+ B[u],
  \end{align*}
  where
  \begin{align*}
    A[u]&= \sum_{\substack{ \rho\ne 1, \\ x\to y \to x \to y' \to x}} \frac{1}{c_\rho}  \lambda_{x\to y \to x \to y'}(W_{\rho,x\to y \to x}(u) +W_{\rho,x \to y' \to x} (u)-1),\\
    B[u]&= - \sum_{\substack{ \rho\ne 1, \\ x\to y \to z \to y' \to x: x\to z}} \frac{1}{c_\rho}  \lambda_{x\to y \to z \to y' }(  W_{\rho,x\to y \to z\to x} (u)+W_{\rho,x \to z \to y' \to x}(u) - W_{\rho,x \to z \to x}(u)).\end{align*}
\end{theorem}
\begin{proof} We take a generic form of $(\  ,\ )_2$ defined as in Lemma~\ref{lem:InnerProductOmega2X} so that
\begin{align*}
    S[u]&=-\sum_{\rho\ne 1}\frac{1}{c_\rho} \sum_{x,z}\sum_{y,y': {x\to y\to z\atop z\to y'\to x}}\chi_\rho(F^u_{x\to y\to z} F^u_{z\to y'\to x})\lambda_{x\to y\to z\to y'}\int\delta_x \\ 
    &=-\sum_{\substack{ \rho\ne 1, \\ x\to y \to z \to y' \to x}} \frac{1}{c_\rho} \chi_\rho(F^u_{x\to y\to z} F^u_{z\to y'\to x})\lambda_{x\to y\to z\to y'}
\end{align*}
given our counting measure (otherwise there is a factor $\mu_x$ which we can absorb in $\lambda_{x\to y\to z\to y'}$) and using that our graphs are bidirectional. 

Next, using the covariant form of $F^u$ in Proposition~\ref{prop:uFu}, it is apparent that $S[u^\gamma]=S[u]$ since  
\[ \chi_\rho(\gamma_xF^u_{x\to y\to z} \gamma_z^*\gamma_z F^u_{z\to y'\to x}\gamma_x^*)=\chi_\rho(F^u_{x\to y\to z} F^u_{z\to y'\to x})\]
 by the unitarity of $\gamma$ and cyclicity of the trace in $\chi_\rho$. Moreover, we compute
\begin{align*} \chi_\rho&(F^u_{x\to y\to z} F^u_{z\to y'\to x})=\chi_\rho((u_{x\to y}u_{y\to z} -\mu u_{x\to z}-\delta_{x,z}e)(u_{z\to y'}u_{y'\to x} -\mu u_{z\to x}-\delta_{x,z}e))\\
&= W_{\rho, x\to y\to z\to y'\to x}-\mu(W_{\rho, x\to y\to z\to x}+W_{\rho, x\to z\to y'\to x}- \mu W_{\rho, x\to z\to x})+\delta_{x,z}(1-W_{\rho, x\to y\to x}-W_{\rho, x\to y'\to x})  \end{align*}
For $\Omega_{min},\Omega_{med}$ we have $\mu=0$, and for $\Omega_{min}$ the $\delta_{x,y}$ part does not contribute due to each term not depending on $y$ or not depending on $y'$, and the relations in Lemma~\ref{lem:InnerProductOmega2X} when we do the sums of the $\lambda_{x\to y\to z\to y'}$, giving the result stated for this case. For $\Omega_{med}$ we have to add this contribution with $z=x$ in the sum over paths, which gives the result stated.

Finally, for $\Omega_{med'},\Omega_{max}$ we have $\mu=1$, and for $\Omega_{med'}$ the $\delta_{x,z}$ part does not contribute due to each term not depending on $y$ or not depending on $y'$. Here $x=z$ implies $x\not \to x$ so the relations of this calculus apply. This gives the stated result for the action for this case where we just add the extra terms for $\mu=1$. For $\Omega_{max}$ we have both types of extra terms since the relations do not allow to cancel the $\delta_{x,z}$ part either.  \end{proof}

We see that  as the relations of $\Omega^2$ become weaker, we have to add extra terms to the main action of the minimal calculus. In the latter case, we sum over the Wilson loops along length 4 loops $x\to y \to z \to y' \to x$, weighted by the factors $\lambda_{x\to y\to z\to y'}$. Here, there are several cases, as we can either have $x = z$ or $x\neq z$, and similarly for $y$, $y'$. This results in the following loops appearing in the action $S^{min}[u]$:
\[
\begin{tikzcd}
	x  \arrow[shift left=3]{r}\arrow[shift right]{r} 
		& y  \arrow[shift left=3]{l}\arrow[shift right]{l}
\end{tikzcd},
\quad\quad \quad \quad \quad
\begin{tikzcd}
x \arrow[shift left,r] \arrow[shift left,d] & y\arrow[shift left,l]\\
y' \arrow[shift left,u] &
\end{tikzcd},
\quad \quad \quad \quad
\begin{tikzcd}
x \arrow[r,shift left] & y\arrow[d,shift left] \arrow[l,shift left]\\
& z \arrow [u,shift left]
\end{tikzcd},
\quad \quad \quad \quad
\begin{tikzcd}
x \arrow[r] & y\arrow[d]\\
y' \arrow[u] & \arrow[l] z
\end{tikzcd},
\]
the first three being Wilson loops along `trivial loops', which contribute non-trivially as we do not have $u_{y\to x} = u^{-1}_{x\to y}$, and the last is the Wilson loop around a square. Note that these loops are not summed over in $S^{min}[u]$ in an independent manner, as the coefficients $\lambda_{x\to y\to z\to y'}$ are related to one another via Proposition \ref{lem:InnerProductOmega2X}. Turning now to $S^{med}[u]$, the action has the extra $A[u]$ contribution where $x = z$, with the following terms
\[
\quad \quad \quad \quad
 2
\quad \quad
\begin{tikzcd}
	x  \arrow[r,shift left] 
		& y  \arrow[l,shift left]
\end{tikzcd}
\quad \quad 
-
\quad \quad
1
\]
for $y = y'$ and 
\[
\quad \quad 
\begin{tikzcd}
x \arrow[shift left,r]& y\arrow[shift left,l]
\end{tikzcd}
\quad \quad 
+
\quad \quad
\begin{tikzcd}
x \arrow[shift left,d] \\
y' \arrow[shift left,u] 
\end{tikzcd}
\quad \quad
-
\quad \quad
1
\]
for $y\neq y'$. In the case of $S^{med'}[u]$ we have the additional $B[u]$ contribution, with terms of the form
\[
\quad\quad
-
\begin{tikzcd}
x \arrow[r] & y\arrow[d]\\
 & \arrow[lu] z
\end{tikzcd}
\quad\quad
-
\quad \quad
\begin{tikzcd}
x \arrow[rd] & y \arrow[l]\\
& z \arrow[u]
\end{tikzcd}
\quad\quad
+
\quad \quad
\begin{tikzcd}
x \arrow[rd,shift left] & y\\
& \arrow[lu,shift left] z
\end{tikzcd}
\]
for $y = y'$ and
\[
\quad\quad
-
\begin{tikzcd}
x \arrow[r] & y\arrow[d]\\
y' & \arrow[lu] z
\end{tikzcd}
\quad\quad
-
\quad\quad
\begin{tikzcd}
x \arrow[rd] & y\\
y' \arrow[u] & \arrow[l] z
\end{tikzcd}
\quad\quad
+
\quad \quad
\begin{tikzcd}
x \arrow[rd,shift left] & y\\
y' & \arrow[lu,shift left] z
\end{tikzcd}
\]
for $y \neq y'$. The case of $S^{max}[u]$ has both types of extra terms. 
\begin{example}
\label{ex:ActionAn1}
As an example consider the finite line $A_{n+1}$. The simplest way to compute the minimal action is via \eqref{eq:FminComponents} as this already incorporates the $\Omega^2$ relations. We have several cases for $x,z$: $x=i$, $z = i\pm 2$ for $i=3,\dots,n-1$ only have one vertex between them and hence the curvature vanishes; similarly for $x=z=1,n$. The non-trivial contributions come from $x=z=i$ , $i=2,\dots n$, with $y_1 = i+1, y_2 = i-1$, resulting in the curvature
\[
F_{min} = \sum^{n}_{i=2} (u_{i\to i+1}u_{i+1\to i}-u_{i\to i-1}u_{i-1\to i} ) \tens e_{i\to i+1} \wedge e_{i+1 \to i} 
\] 
and the action 
\begin{align*}
S^{min}_{A_{n+1}}[u] &= - \sum_{\rho\neq 1} \frac{1}{c_\rho} \sum^{n}_{i=2} \lambda_{i\to i+1 \to i \to i+1} \chi_\rho((u_{i\to i+1}u_{i+1\to i} - u_{i\to i-1}u_{i-1\to i} )^2 ) \nonumber\\
&= -\sum_{\rho\neq 1}\frac{1}{c_\rho} \sum^{n}_{i=2} \lambda_{i\to i+1 \to i \to i+1} (W_{\rho,i\to i+1\to i\to i+1}(u) + W_{\rho,i\to i-1\to i\to i-1}(u) - 2W_{\rho,i\to i+1\to i\to i-1}(u)),
\end{align*}
with $\lambda_{i\to i+1 \to i \to i+1} = - A_{i,i} \lambda_{i\to i+1} \lambda_{i\to i-1}$ using the solution from Proposition \ref{eq:SolLambda(,)2}. This is a stark contrast to conventional LGT and continuum Yang-Mills theory, which are both trivial in one dimension. 
\end{example}

\subsection{Yang-Mills theory on $X$ a discrete group}\label{secXG}

Let us now restrict to $X$ itself a discrete group with calculus given by $\CC$, a  bicovariant exterior algebra $\Omega$, a general metric $\cg^{a,b} = \frac{\delta_{a,b^{-1}}}{R_a c_{a^{-1}}}$ and $(\ , \ )_2$ as in \eqref{eq:(,)2G}, with fibre $G$ a generic finite group. We can write $\C G \tens \Omega^1 \simeq \C G \tens \C(X) \tens \Lambda^1$ with $\Lambda^1$ the space of left invariant one forms on the base graph and therefore decompose the holonomies as $u = \sum_{a\in \CC} u_a \tens e_a$ with $u_a = \sum_{x\in X} u^x_a \tens \delta_x \in \C G \tens \C(X)$. To see how this relates to the presentation in terms of the arrows, write $e_a = \sum_{x\in X} e_{x\to xa}$ which leads to
\begin{align*}
    u = \sum_{a\in \CC} u_a \tens e_a = \sum_{a\in \CC, x \in X} u^x_a \tens e_{x\to xa}
\end{align*}
showing that $u^x_a = u_a(x) = u_{x\to xa}$, thus recovering the value of the holonomy associated to an arrow $x\to xa$ by evaluating $u_a$ at $x$. Furthermore, right translating $u_b$ by $R_a$ results in $(R_au_b)(x) = u_b(xa) = u_{xa \to xab}$. We can then use these to represent holonomies along paths and loops, for example $u_a R_a u_b$ evaluated at $x$ will give the holonomy along $x\to xa \to xab$.

As we work with $\Omega$ bicovariant, the exterior algebra is inner with $\theta \wedge \theta = 0$, leading to
\[
F = u \wedge u = \sum_{a,b\in \CC} (u_a \tens e_a) \wedge (u_b \tens e_b)
= \sum_{a,b\in \CC} u_a (R_a u_b) \tens e_a \wedge e_b
\]
and the Yang-Mills term 
\begin{align*}
	&(F, F )_2
	=  \sum_{a,b,c,d \in \CC} u_a (R_a u_b) (R_{ab} u_c) (R_{abc} u_d) (e_a \wedge e_b, e_c \wedge e_d)_2 \nonumber \\
	&=\sum_{x\in X} \sum_{a,b,c,d \in \CC} 
	u_a (R_a u_b) (R_{ab} u_c) (R_{abc} u_d) \lambda_{x\to xa}\lambda_{x\to xb} \delta_{e,abcd} (\delta_{a,c^{-1}} + \delta_{c, b^{-1}a^{-1} b } - 2 \delta_{b,c^{-1}} )\delta_x  \nonumber\\
	&= \sum_{x\in X} \sum_{a\neq b \in \CC} 
	 \lambda_{x\to xa}\lambda_{x\to xb} u_a (R_a u_b)[ (R_{ab} u_{a^{-1}}) (R_{aba^{-1}} u_{ab^{-1}a^{-1}})
	+  (R_{ab} u_{b^{-1}a^{-1}b} ) (R_{b} u_{b^{-1}})
	- 2 (R_{ab} u_{b^{-1}}) (R_{a} u_{a^{-1}})] \delta_x\nonumber \\
	&= \sum_{x\in X} \sum_{a\neq b \in \CC} 
	 \lambda_{x\to xa}\lambda_{x\to xb} u_{x\to xa} u_{xa\to xab} [ u_{xab \to xaba^{-1}}u_{xaba^{-1} \to x}
	+  u_{xab \to x b} u_{xb \to x}
	- 2 u_{x ab \to xa } u_{xa \to x}] \delta_x.
\end{align*}
We find the loops according to the discussion below Equation \eqref{eq:weightsGroupCase}, and omit the $a=b$ case since $\lambda_{x\to xa\to xa^2 \to xa} = 0$. In the case that $X$ is Abelian, the Yang-Mills term further reduces to
\begin{equation}\label{SYMab}
(F,F)_2 = 2 \sum_{x\in X} \sum_{a\neq b \in \CC}  \lambda_{x\to xa}\lambda_{x\to xb}
	u_{x\to xa} u_{xa\to xab}( u_{xab \to xb}u_{xb \to x} 
	- u_{xab \to xa} u_{xa \to x}) \delta_x, 
\end{equation}
which diagrammatically can be represented as
\[
\begin{tikzcd}
x \arrow[r] & xa \arrow[d]\\
xb \arrow[u] & xab \arrow[l]
\end{tikzcd}
\quad - \quad
\begin{tikzcd}
x \arrow[r,shift left] & xa \arrow[d,shift left] \arrow[l,shift left]\\
xb  & xab \arrow[u,shift left]
\end{tikzcd}.
\]
This is relevant when considering for example the lattice $\Z^M$ or $\Z^M_N$ with calculus $\CC = \{\pm e_i\vert i = 1,\dots, M\}$, with $e_i$ the $M$-tuple with $1$ in the $i$'th place and zero elsewhere. In the case of $\Z_2^M$, we do not have the $\pm$ signs. One can also insert a measure $\mu(x)$ in the integration for the action, but for simplicity we do not do so here.

\begin{example}
\label{ex:ActionZ}
In the case of the integer line $X=\Z$, we have $e_1=1\in \Z$ and 
\begin{align*}
\label{eq:ActionZ}
S_{\Z}[u] = - 4 \sum_{\rho\neq 1} \frac{1}{c_\rho} \sum_{i\in \Z} \lambda_{i \to i+1}\lambda_{i \to i-1}  \chi_\rho(u_{i \to i+1} u_{i+1\to i} (u_{i\to i-1} u_{i-1\to i} - u_{i\to i+1} u_{i+1\to i})),
\end{align*}
\begin{equation*}
    = - 4 \sum_{\rho\neq 1} \frac{1}{c_\rho}\sum_{i\in \Z} \lambda_{i \to i+1}\lambda_{i \to i-1}\chi_\rho
    \left(
        \begin{tikzcd}
            i-1  \arrow[shift left]{r} 
            & i \arrow[shift left]{l} \arrow[shift left]{r}
            & i+1 \arrow[shift left]{l} 
        \end{tikzcd}
        -
        \begin{tikzcd}
        i \arrow[shift left=3]{r}\arrow[shift right]{r} 
        & i+1  \arrow[shift left=3]{l}\arrow[shift right]{l}
        \end{tikzcd}
    \right),
\end{equation*}
where we also give the shorthand diagrammatic representation. The factor 2 appears since the $a=+1,b=-1$ and  $a=-1,b=+1$ terms are the same due to the cyclic properties of $\chi_\rho$. Similarly for $\Z_N$. 
\end{example}

\begin{example}\label{Z2Z2act} In the case of $X = \Z_2\times \Z_2$, we use an abbreviated notation $ij$ for the vertex $(i,j)$ unless required for clarity, and have $e_1 = (1,0), e_2 = (0,1)\in \CC$. This makes $X$ into a square graph
\[    
\begin{tikzcd}
    01 \arrow{r}\arrow{d} & 11 \arrow{d}\arrow{l}\\
    00 \arrow{u}\arrow{r} & 10 \arrow{l}\arrow{u} \\
    \end{tikzcd},
\]
on which we take the Euclidean metric where $\lambda_{x\to y}=1$. There are four arrows and hence 8 holonomies $u_{(i,j)\to (i+1,j)}$ and $u_{(i,j)\to (i,j+1)}$, where addition is mod 2, with half of these adjoint to the other half under $*$. Adding up the values of $(F,F)_2$ at the four vertices, the action is then
\begin{align*} S_{\Z_2\times\Z_2}[u]=-\sum_{\rho\ne 1}{2\over c_\rho}\big(& 4 W_{\rho,00\to 01\to 11\to 10\to 00}(u) + 4 W_{\rho,00\to 10\to  11\to 01\to 00}(u) - 2(W_{\rho,00\to 10\to 11\to 10\to 00}(u)\\ &-2 W_{\rho,00\to 01\to 11\to 01\to 00}(u) -
        2 W_{\rho,00\to 01\to 00\to 10\to 00}(u) -2W_{\rho,10\to 11\to 01\to 11\to 10}(u) \big).
        \end{align*}
\[ 
    = - \sum_{\rho \neq 1}\frac{2}{c_\rho}
   { \chi_\rho\left[ 4 \left(
    \begin{tikzcd}
    \bullet \arrow{r} & \bullet \arrow{d}\\
    \bullet \arrow{u} & \bullet \arrow{l}
    \end{tikzcd}
    +
    \begin{tikzcd}
    \bullet \arrow{d} & \bullet \arrow{l}\\
    \bullet \arrow{r} & \bullet \arrow{u}
    \end{tikzcd}
    \right)
-2 \left(
\begin{tikzcd}
\bullet \arrow[shift left]{d}\arrow[shift left]{r} & \bullet \arrow[shift left]{l}\\
\bullet \arrow[shift left]{u} & \bullet 
\end{tikzcd}
+
\begin{tikzcd}
\bullet \arrow[shift left]{r} & \bullet \arrow[shift left]{l} \arrow[shift left]{d}\\
\bullet  & \bullet \arrow[shift left]{u}
\end{tikzcd}
+
\begin{tikzcd}
\bullet  \arrow[shift left]{d} & \bullet\\
\bullet \arrow[shift left]{r} \arrow[shift left]{u} & \bullet \arrow[shift left]{l}
\end{tikzcd}
+
\begin{tikzcd}
\bullet  & \bullet \arrow[shift left]{d}\\
\bullet \arrow[shift left]{r} & \bullet \arrow[shift left]{u}\arrow[shift left]{l}
\end{tikzcd}
\right)
\right]}
\]
using cyclicity of $\chi_\rho$. 
\end{example}

\begin{example} 
In the case of the 2D lattice $X=\Z\times\Z$, the calculus is 4-dimensional with $\CC=\{\pm e_1,\pm e_2\}$ for independent signs. We find for the Euclidean metric ($\lambda_{x\to y} = 1$),
\begin{equation*}
    S_{\Z\times \Z}[u] 
    = -\sum_{* \in \Z\times \Z} \sum_{\rho\ne 1} \frac{2}{c_\rho} \chi_\rho
    \left[ 4 \left(
    \begin{tikzcd}
    \bullet \arrow{r} & \bullet \arrow{d}\\
    * \arrow{u} & \bullet \arrow{l}
    \end{tikzcd}
    +
    \begin{tikzcd}
    \bullet \arrow{d} & \bullet \arrow{l}\\
    * \arrow{r} & \bullet \arrow{u}
    \end{tikzcd}
    \right)\right.
\end{equation*}
\begin{equation*}
-2 \left(
\begin{tikzcd}
\bullet \arrow[shift left]{d}\arrow[shift left]{r} & \bullet \arrow[shift left]{l}\\
* \arrow[shift left]{u} & \bullet 
\end{tikzcd}
+
\begin{tikzcd}
\bullet \arrow[shift left]{r} & \bullet \arrow[shift left]{l} \arrow[shift left]{d}\\
\bullet  & * \arrow[shift left]{u}
\end{tikzcd}
+
\begin{tikzcd}
\bullet  \arrow[shift left]{d} & \bullet\\
\bullet \arrow[shift left]{r} \arrow[shift left]{u} & * \arrow[shift left]{l}
\end{tikzcd}
+
\begin{tikzcd}
\bullet  & \bullet \arrow[shift left]{d}\\
* \arrow[shift left]{r} & \bullet \arrow[shift left]{u}\arrow[shift left]{l}
\end{tikzcd}
\right)
\end{equation*}
\begin{equation*}
-2 \left(\left.
\begin{tikzcd}
* \arrow[shift right]{d} \arrow[shift left=3]{d} \\
\bullet \arrow[shift right]{u} \arrow[shift left=3]{u}
\end{tikzcd}
-
\begin{tikzcd}
\bullet \arrow[shift left]{d}\\
* \arrow[shift left]{u} \arrow[shift left]{d} \\
\bullet \arrow[shift left]{u}
\end{tikzcd}
+
\begin{tikzcd}
*  \arrow[shift left=3]{r}\arrow[shift right]{r} 
& \bullet  \arrow[shift left=3]{l}\arrow[shift right]{l}
\end{tikzcd}
-
\begin{tikzcd}
    \bullet  \arrow[shift left]{r} 
& * \arrow[shift left]{l} \arrow[shift left]{r}
& \bullet \arrow[shift left]{l} 
\end{tikzcd}
\right)\right].
\end{equation*}
The Wilson action \cite{Wil:con} for $\Z^2$ is recovered but now with additional terms. The latter correspond to trivial `back and forth' loops and become constant if we restrict the extended holonomies $u$ to be in the group $G$ as one would do in LGT. For then, as $u_{x\to y}= u^*_{y\to x} = u^{-1}_{y\to x}$, the holonomy along such a loop is the identity and these terms will be inconsequential to the dynamics of the theory, as they are independent of $u$.
\end{example}

A further interesting case is that of a 2D triangular lattice $\Z^2$ with calculus $\CC = \{(\pm 1,0),(0,\pm 1), (\pm 1, \pm 1)\}$. Here if we work with $\Omega_{min}$, we would again find trivial loop terms and \emph{square} terms in the Yang-Mills action, not loops around the triangular plaquettes as one would expect. The squares that appear in the action include additional square loops terms as follows
\[
\begin{tikzcd}
    \bullet \arrow[dash]{r} \arrow[dash]{d} & \bullet \arrow[dash]{d} \arrow{r} & \bullet \arrow[dash]{d} \arrow{ld}\\
    \bullet \arrow{ur} & \bullet \arrow{l} & \bullet\arrow[dash]{l}
\end{tikzcd}
\]
and similar loops. If one considers $\Omega_{med'}$ or $\Omega_{max}$ instead, then triangular loops appear directly in the action. Note that for a hexagonal lattice, the Yang-Mills action would only consist of trivial loops, as the loops appearing in it are at most of length 4.

\begin{example}
A nonAbelian example is  to take the base to be $S_3$ with $\CC = \{u,v,w\}$, which results in
 \[
	S_{S_3}
	= - \sum_{\rho\neq 1}\frac{1}{c_\rho}\int_{S_3} \sum_{a\neq b \in \CC} 
	 \lambda_{a}\lambda_{b} \chi_\rho \left(u_a (R_a u_b)[ (R_{ab} u_{a}) (R_{aba} u_{aba})
	+  (R_{ab} u_{bab} ) (R_{b} u_{b})
	- 2 (R_{ab} u_{b}) (R_{a} u_{a})]\right), 
\]
where we use $\lambda_a(x) = \lambda_{x\to xa}$ and $a = a^{-1}$. Integration means to sum the functions shown over $S_3$, possibly with a measure. The first and second terms of this action correspond to loops around different plaquettes, for example starting at $u\in S_3$ we have $u \to ua \to uab \to uaba \to uabaaba = u$ for the first term and $u\to ua \to uab \to uababa = ub \to u$, while the last term corresponds to trivial loops $u \to ua \to uab \to uabb = ua \to u$. The Cayley graph describing the base geometry here is
\[
\begin{tikzcd}
{} & {u \arrow[r,leftrightarrow]} & {uv \arrow[rd,leftrightarrow] } & {} \\
{e \arrow[ur,leftrightarrow]\arrow[rrr,leftrightarrow] } & {} & {} & {w\arrow[ld,leftrightarrow] } \\
{} & {v \arrow[lu,leftrightarrow] \arrow[ruu,leftrightarrow]  } & {vu \arrow[l,leftrightarrow]\arrow[luu,leftrightarrow]  } & {} \\
\end{tikzcd}.
\]
\end{example}

\subsection{Circle limit of $\Z_N$}

To take the circle limit of $\Z_N$, i.e. $N\to \infty$, we follow the approach in \cite[Section 3.2]{ArgMa} by Fourier transforming $\C(\Z_N)$ to the variable $s\in \C(\Z_N)$ with $s(i) = q^i, q= e^{\frac{2\pi \imath }{N}}$. This is such that the base algebra can be reinterpreted as $A = \C(\Z_N) \simeq \C \Z_N = \C[s,s^{-1}]/\langle s^N-1 \rangle$, with relations $s^N=1$. In the circle limit, we can then think of the variable as $s = e^{\imath \varphi}$ for $\varphi$ the angle around the circle. In this form, we work with the basis 1-forms $f_\pm$ of $\Omega^1$, with module structure, $*$-structure, inner element and differential
\begin{align*}
    &f_+ = s^{-1}\extd s = (q-1)e_+ + (q^{-1}-1)e_-,&
    &f_+ s = s(f_- + (q+q^{-1}) f_+),& 
    &{f_\pm}^* = - f_\pm,\\
    &f_- = s\extd s^{-1} = (q^{-1}-1)e_+ + (q-1)e_-,&
    &f_- s = -s f_+,&
    &\theta = \frac{q}{(q-1)^2} \theta_0,\\
    &\extd s^m = \frac{q[m]_q s^m}{q+1}(q[-1-m]_q f_+ + [m-1]_q f_-)
\end{align*}
with $[m]_q = \frac{1-q^m}{1-q}$ the $q$-integers and $\theta_0 = f_++f_-$. Since $e_{\pm}$ are Grassmanian, the $f_{\pm}$ generators will also be. In this form, one can easily take the $N\to \infty$ limit, resulting in the stated $q$-deformed 2-dimensional calculus on classical $S^1$. We then recover the usual classical geometry of $S^1$ as a quotient by taking $q\to 1$ {\em and} setting $\theta_0 = 0$. 

It is simpler to consider the above calculus in the $f_+$, $\theta_0 = f_+ + f_-$ basis, as $\theta_0$ is central. In particular we have
\[
\extd s^m = \frac{q}{1+q} \left((1+q^{-m})s\del_q s^m f_+ + q^{1-m} s^2 \del^2_q s^m \theta_0 \right),
\quad \quad
f_+ s^m = \frac{1}{1+q}\left((1+q^{-1})q^{-m} s^m f_+ + q(1+q^{-m})s\del_q s^m \theta_0\right)
\]
where $\del_q s^m = [m]_q s^{m-1}$ is the $q$-derivative wrt $s$. 

This algebra is quite complicated, and therefore we focus on its $q\to 1$ limit where results are easier to interpret. In this setting one can easily compute $\extd f$ and $[f_+,f]$ for any function $f\in \C[s,s^{-1}]$ as 
\[
\extd f = -\imath (\del_\varphi f) f_+ + \frac{1}{2} (\imath \del_\varphi f - \del_\varphi^2 f) \theta_0,
\quad \quad
[f_+,f] = -\imath (\del_\varphi f) \theta_0,
\]
where we work with an angle $\varphi = -\imath \ln(s)$. In this form we can now allow $f$ to be any smooth function on $S^1$. A downside of this limit is that the calculus is no longer inner, meaning that the extended holonomy formulation is not applicable, and we have to work directly with the NCG point of view where connections are 1-forms $\alpha = \alpha_+ \tens f_+ + \alpha_0 \tens \theta_0 \in (\C G^+ \tens \C[s,s^{-1}]) \tens \Lambda^1$, with $\alpha_{+,0} \in \C G \tens \C[s,s^{-1}]$ which can also be taken to be smooth functions on $S^1$ with values in $\C G^+$. The reality condition  $\alpha^* = -\alpha$ on connection $1$-forms then translates to $\text{Im}(\alpha_+) = 0$ and $\text{Im}(\alpha_0) = - \frac{1}{2} \del_\varphi \alpha_+$, meaning that the connection consists of two independent real fields $\alpha_+$ and $\mathrm{Re}(\alpha_0)$. The curvature $F = \extd \alpha +\alpha \wedge \alpha$ is
\[
F =  \left(\imath (\alpha_+ - 1) \del_\varphi \mathrm{Re}(\alpha_0) - \frac{\imath}{2} (\alpha_+ + 1) \del_\varphi \alpha_+ +[\alpha_+,\mathrm{Re}(\alpha_0)-\frac{\imath}{2} \del_\varphi \alpha_+] \right) f_+ \wedge \theta_0,
\]
where the last term vanishes when $G$ is Abelian. In this setting, the most general central metric and its inverse obeying the reality condition are
\[
\cg = \imath a (f_+ \tens \theta_0 - \theta_0 \tens f_+) + b \theta_0 \tens \theta_0,
\quad 
(f_+,f_+) = - \frac{b}{a^2}, \quad
(\theta_0,f_+) =-  \frac{\imath}{a}, \quad
(f_+,\theta_0) = \frac{\imath}{a}, \quad
(\theta_0,\theta_0) = 0,
\]
with $a, b$ real functions on the circle and $a\ne 0$. Note that the quantum symmetry condition $\wedge \cg = 0$ implies $a = 0$, making the metric degenerate (we do not impose this).  The 2-metric then reads $(f_+\wedge \theta_0,f_+\wedge \theta_0)_2 = - \frac{2}{a^2}$ and the action is of the form
\[
S[\alpha] = -\sum_{\rho} \frac{1}{c_\rho} \int^{2\pi}_{0} \frac{2}{a^2} \chi_\rho\left( (\alpha_+ - 1) \del_\varphi \mathrm{Re}(\alpha_0) - \frac{1}{2} (\alpha_+ + 1) \del_\varphi \alpha_+ -\imath [\alpha_+,\mathrm{Re}(\alpha_0)-\frac{\imath}{2} \del_\varphi \alpha_+] \right)^2 \extd \varphi.
\]
Expanding this expression in the Abelian $G$ case and dropping the $\mathrm{Re}$ for simplicity gives
\begin{align*}
S[\alpha] = - \sum_{\rho} \frac{1}{c_\rho} \int^{2\pi}_{0} \frac{2}{a^2} \chi_\rho 
&\left((\del_\varphi \alpha_0)^2 + 2 \del_\varphi \alpha_0 \del_\varphi \alpha_+ + \frac{1}{4}(\del_\varphi \alpha_+)^2 + \alpha_+ \left(\frac{1}{2}(\del_\varphi \alpha_+)^2-2 (\del_\varphi \alpha_0)^2 \right)\right. \\ 
& \left. + \alpha^2_+ \left((\del_\varphi \alpha_0)^2
- 2 \del_\varphi \alpha_0 \del_\varphi \alpha_+ 
+\frac{1}{4}(\del_\varphi \alpha_+)^2\right)
\right) \extd \varphi.
\end{align*}
We see that noncommutative Abelian gauge theory on $S^1$ results in a rich action,  with 2, 3 and 4 point (derivative) interactions, which one does not expect from Abelian gauge theory. We can also Fourier transform the action, 
\begin{align*}
    S[\alpha] =- 2 \sum_{\rho} \frac{1}{c_\rho} \chi_\rho 
    &\left(
    \sum_{k\in \Z}    
    k^2 (\tilde \alpha_{0,k} \tilde \alpha_{0,-k} + 2 \tilde \alpha_{0,k} \tilde \alpha_{+,-k} + \frac{1}{4}\tilde \alpha_{+,k}\tilde \alpha_{+,-k})\right. \\ 
    & + \sum_{k_1,k_2,k_3 \in \Z} \delta(k_1+k_2+k_3) k_2 k_3 \tilde \alpha_{+,k_1} \left(\frac{1}{2} \tilde \alpha_{+,k_2}\tilde \alpha_{+,k_3} - 2 \tilde \alpha_{0,k_2}\tilde \alpha_{0,k_3} \right)\\ 
    & + \left.\sum_{k_1,k_2,k_3,k_4 \in \Z} \delta(k_1+k_2+k_3+k_4)  k_3 k_4 \tilde \alpha_{+,k_1} \tilde\alpha_{+,k_2} \left(\tilde \alpha_{0,k_3}\tilde\alpha_{0,k_4} 
    - 2 \tilde\alpha_{0,k_3}\tilde\alpha_{+,k_4} 
    +\frac{1}{4}\tilde\alpha_{+,k_3}\tilde\alpha_{+,k_4}\right)\right),
\end{align*}
shown here for $a=1$, as a starting point for Feynman rules or for numerical calculations for bounded momenta $-N\le k\le N$. 

\section{Moduli space $\CU^\times / \CG$}
\label{sec:mod}

In our NCG approach, the  moduli spaces $\CA / \CG \simeq \CU / \CG$  for Abelian gauge theories can be computed explicitly, as we will demonstrate in this section. It will, however, be convenient to restrict the set of connections to `regular' ones in the sense of the extended holonomy invertible on each arrow,
\[     \CU^\times = \{u\in \C G^\times \tens \Omega^1\  \vert\  u^* = -u,\  (\epsilon\tens \id) u = \theta\}=\{u\in C(E,\C G^\times_1)\ |\  u_{x\to y}^*=u_{y\to x}\},\] 
where $\C G^\times$ denotes the group of invertible elements of the group algebra and $\C G^\times_1$ denotes the subgroup of counit 1. This is dense the set of all elements of $\C G$ of counit 1, hence $\CU^\times \subset \CU$ is dense as well at least if $E$ is finite. Consequentially, this restriction does not have an impact on the physics from the perspective of the path integral, as we are ignoring a subset of extended holonomies with vanishing measure. 

Secondly, the key to our analysis will be our Peter-Weyl decomposition into a product of noncommutative $U(d_\rho)$-gauge theories allowing analysis for each $\rho\in \hat G\setminus\{1\}$ at a time. Here $\hat G$ denotes the set of irreducible representations up to equivalence. These can be taken to be unitary up to unitary equivalence, which we do. 

Note that in conventional LGT, both holonomies and gauge transformations are $G$-valued, while in the NCG approach we have taken them both essentially $\C G$-valued. An intermediate case of  $\C G$-holonomies and $G$-valued gauge transformations is also of interest, see Section~\ref{secnonu}.

\subsection{Moduli spaces for $G$ Abelian}

We first analyse the case of $G$ Abelian as this has some significant simplifications. In this case $\pi_\rho={1\over |G|}\sum_{h\in G}\rho(h^{-1})h$ sends $g\in G$ to  \[ \pi_\rho g={1\over |G|}\sum_{h\in G}\rho(h^{-1})hg={1\over |G|}\sum_{hg\in G}\rho(g(hg)^{-1})hg=\rho(g)\pi_\rho\]
so the decomposition just associates to a group element its different values in the irreps. This therefore just becomes the tautological isomorphism
\[  \C G\simeq \C(\hat G)\]
as Hopf $*$-algebras that sends $g\in G$ to the function  $g(\rho)=\rho(g)$. This extends linearly to send $\sum_{g\in G}f_g g\in \C G$ to $\tilde f(\rho)=\sum_{g\in G} f_g \rho(g)$, i.e. the Fourier transform of the coefficients. Going the other way is therefore the inverse Fourier transform.  We are interested in unitary and/or invertible elements of  counit 1, 
\[ \C G^u\simeq C(\hat G\setminus \{1\}, U(1))=U(1)^{|G|-1},\quad \C G^\times_1\simeq C(\hat G\setminus \{1\}, \C^\times)=(\C^\times)^{|G|-1}\]
where $|\hat G|=|G|$ and counit 1 in the group algebra translates to a fixed value of $1$ at $\rho=1$, so this is excluded from the data (but should be remembered when inverse Fourier transforming back to the group algebra).  Consequently, we see that we can equivalently work with our function spaces in the form
\[ \CU^\times=C(E_0\times(\hat G\setminus\{1\}), \C^\times),\quad \CG=C(X\times (\hat G\setminus\{1\}), U(1))\]
where $E_0$ denotes the set of undirected edges. This depends on choosing an {\em orientation} of the graph (a choice of arrow on every edge) on which $u_{x\to y}$ is specified, with $u_{y\to x}=u_{x\to y}^*$. Meanwhile, 
 the functional dependence on $\rho\ne 1$ just gets carried along in all computations. 

\subsubsection{Tree graph $X$ case}\label{secsimply}

In this section we will use that every element of $\C^\times$ has a unique polar decomposition as $re^{\imath\theta}$ where $\theta\in [0, 2\pi)$ and $r>0$, and show that for a  connected simply connected (i.e., tree) graph the phase degree of freedom in $\CU^\times$ can thereby always be gauged away. This leaves a residual quantum-gravity like moduli space 
\[  \CU^\times/\CG=C(E_0\times(\hat G\setminus\{1\}), \R_{>0})= \R_{>0}^{|E_0|(|G|-1)}\]
where, for each irrep $\rho\ne 1$, we associate a residual real positive number on every (undirected) edge as in the specification of a quantum metric on a graph \cite{Ma:squ, BegMa}. We will see in later sections that in some cases the YM action also reduces to the action for quantum gravity on graphs. 

We start by computing $\CU^\times/\CG$ for the Dynkin graph $A_2$ with two vertices and two arrows
\[
\begin{tikzcd}
x \arrow[shift left,r]& y\arrow[shift left,l]
\end{tikzcd}.
\]
Clearly $\CU^\times=(\C^\times)^{|G|-1}$ as the value of $u_{x\to y}$ (say) at each $\rho\ne 1$, which we denote $u_{x\to y}^\rho\in \C^\times$. Similarly, the group $\CG=(U(1)\times U(1))^{|G|-1}$ as the values of $(\gamma_x,\gamma_y)$ at each $\rho\ne 1$ and we denote these as $\gamma_{x}^\rho$ and $\gamma_{y}^\rho$. Now, under a gauge transformation $u_{x\to y}$ maps to the product $\gamma_x u_{x\to y} \gamma_y^*\in \C G^\times_1$ which is the pointwise product as functions on $\hat G\setminus\{1\}$, i.e.
\[ (u^\gamma)_{x\to y}^\rho= \gamma_{x}^\rho u_{x\to y}^\rho\gamma_{y}^\rho{}^*.\]
All that remains is to make a polar decomposition of $u_{x\to y}^\rho=r_{x\to y}^\rho e^{\imath \theta_{x\to y}^\rho}$ and (for example) set $\gamma_x=e$ and $\gamma_{y}^\rho=e^{\imath \theta_{x\to y}^\rho}$ so that $(u^\gamma)_{x\to y}^\rho=r_{x\to y}^\rho$ a positive real number. This is the best we can do by gauge transformations as (in the Abelian case) only the ratio $\gamma_x\gamma_y^*$ effectively  acts and this is just a phase at each $\rho\ne 1$. Hence, for this graph,
\[ \CU^\times/\CG\simeq C(\hat G\setminus\{1\},\R_{>0})= \R_{>0}^{|G|-1}.\]
Given such a function $\tilde f(\rho)\in \R_{>0}$ the corresponding value in the group algebra is 
\[ u_{x\to y}=\sum_g f_g g,\quad f_g={1\over |G|}\sum_{\rho\in \hat G}\tilde f(\rho)\rho(g)^{-1};\quad \tilde f(1)=1.\]
This argument generalises to any tree graph $X$:

\begin{proposition}\label{propsimply}
Let $G$ be a finite Abelian group and $X$ a  bidirected tree graph. Then the moduli space $\CU^\times/\CG$ is isomorphic to $\R^{|E_0|(|G|-1)}_{>0}$.
\end{proposition}
\proof  This can be proven by induction, but it is more informative to describe the process. Also, we suppress a fixed $\rho\ne 1$ as an argument to all the fields. We start at any vertex $x_0$ and consider a path $x_0\to x_1\cdots\to x_n$ with no intersections along the part. We choose an initial $\gamma_{x_0}$ randomly and adjust $\gamma_{x_1}$ so that $\gamma_{x_0}\gamma_{x_1}^*$ kills the phase of $u_{x_0\to x_1}$ as we did for the $A_2$ case. We then take this value of $\gamma_{x_1}$ and adjust $\gamma_{x_2}$ so that $\gamma_{x_1}\gamma_{x_2}^*$ kills the phase of $u_{x_1\to x_2}$, etc, to the end path. We then repeat the process for all other such paths starting from $x_0$ and using the previously assigned initial $\gamma_{x_0}$. Likewise, for each path such as $x_0\to x_1\cdots \to x_n$ we repeat the process for all paths stating from $x_n$ using the previous assigned initial value $\gamma_{x_n}$. We repeat the process at all the branch ends of these paths, and so on. As the underlying undirected graph has no loops, we will not encounter in this process another previously assigned value of $\gamma$.  At the end of this process, we will have set all the $u_{x\to y}$ real and positive, hence equal to $u_{y\to x}=u_{x\to y}^*$  as the residual moduli space. \endproof

Note that for a tree graph $|X|=|E_0|+1$ and indeed, we chose one initial value $\gamma_{x_0}$ randomly then determined the remaining $\{\gamma_x\}$ from the phases of the equal number of $\{u_{x\to y}\}$ counted in one direction only for a given orientation.

\subsubsection{$N$-gon graph $X = \Z_N$}\label{secloop}

Take now the simplest non-simply connected graph, the triangle $\Z_3$ with calculus $\CC = \{\pm1 \}$,
\[
\begin{tikzcd}
 & 1\arrow[rd,shift left] \arrow[ld,shift left] & \\
2 \arrow[ru,shift left] \arrow[rr,shift left] & & 3 \arrow[lu,shift left] \arrow[ll,shift left]
\end{tikzcd}.
\]
Here $\CU^\times=(\C^\times)^{3(|G|-1)}$ as the value in $\C^\times$ of $u_{1\to 2}, u_{2\to 3}, u_{3\to 1}$ (say) at each $\rho\ne 1$. We do all the analysis for a fixed $\rho$ at a time, so we will suppress writing this explicitly. Similarly, the group of gauge transformations is $\CG=(U(1)\times U(1)\times U(1))^{|G|-1}=U(1)^{3(|G|-1)}$ as the values of $\gamma_1,\gamma_2,\gamma_3$ at each $\rho\ne 1$. Following the algorithm of the proof of Proposition~\ref{propsimply}, we choose $\gamma_1$ randomly then $\gamma_2$ such that $\gamma_1\gamma_2^*$ cancels the phase of $u_{1\to 2}$. Using this value of $\gamma_2$ we choose $\gamma_3$ so that $\gamma_2\gamma_3^*$ cancels the phase of $u_{2\to 3}$. We will not, however, be able to cancel the phase of $u_{3\to 1}$ as $\gamma_1$ was already specified. In fact the product $u_{1\to 2} u_{2\to 3}u_{3\to 1}$ is gauge invariant hence the product of the phases cannot be changed by a gauge transformation. We also cannot affect the real parts of the $u$, which are independent of the direction  by $u_{x\to y}^*=u_{y\to x}$. Hence
\[ \CU^\times/\CG\simeq (U(1)\times \R_{>0}^3)^{|G|-1}\]
when we put in the possible values of $\rho\ne 1$. We see that the real part of the moduli space is $C(E_0\times(\hat G\setminus\{1\}),\R_{>0})$ of an $\R_{>0}$ `gravity-like' assignment to every (undirected) edge, for each $\rho\ne 1$ as for a tree  graph, but there is an additional $U(1)$ for each $\rho\ne 1$ for the phase part of the holonomy around the triangle in our description. This immediately generalises.

\begin{proposition}\label{proploop}
Let $G$ be an Abelian finite group and $X = \Z_N$ with calculus $\CC = \{\pm1\}$. Then $\CU^\times/ \CG \simeq   (U(1)\times \R_{>0}^N)^{|G|-1} $.
\end{proposition}
\begin{proof} The argument is identical. We number the vertices $1,2,\cdots, N$, assign a random $\gamma_1$ and proceed to assign $\gamma_i$ up to and including $\gamma_{N-1}$ to cancel the phases of $u_{1\to 2},\cdots,u_{N-1\to N}$. We cannot choose $\gamma_N:=\gamma_1$ as this was assigned and indeed the product, $u_{1\to 2}\cdots u_{N-1\to N}u_{N\to 1}$ is gauge invariant leaving a residual $U(1)$ freedom for each $\rho$ and the real parts of the $u_{x\to y}$ associated to each undirected edge. \end{proof}

\subsubsection{General connected graph $X$ case} 

The above results are warm-up for the general case of any connected  graph $X$. First fix a minimal {\em spanning tree} of $X$ as an undirected graph, i.e. a subgraph with the same vertices and which is a tree. Then any two vertices have a unique path between them in $T$. A {\em fundamental loop} is given for any edge $x\leftrightarrow y$ that is {\em not in the tree} completed to an undirected loop by travelling back through $T$. The number of such fundamental loops is  $L=|E_0|-|X|+1$ and it is known that they form a basis of the cycle space \cite{Die}. 

\begin{proposition}
Let $G$ be an Abelian finite group and $X$ a bidirected  graph with $L$ fundametal loops. Every regular extended holonomy configuration can be gauge transformed to a unique standard form that for each $\rho\ne 1$ is real on the spanning tree $T$. The moduli space is
\begin{align*}
\CU^\times / \CG \simeq (U(1)^L\times \R_{>0}^{|E_0|})^{|G|-1},
\end{align*} 
where $L=|E_0|-|X|+1$. 
\end{proposition}
\begin{proof}
We proceed for each representation $\rho\ne 1$ in the same way.  Pick a vertex $x_0$ in $T$ and assign $\gamma_{x_0}=1$ there. Then move along  $x_0\to x_1$ along $T$ assigning $\gamma_{x_1}$ to cancel the phase there. Proceed until you reach a univalent vertex. Do the same at any branches where there was more than one direction in $T$ in which to proceed. As there are no loops in $T$ there is no obstruction. Up to an overall constant, there is a unique gauge transform that achieves this. Any edges that are not in the spanning tree cannot be changed further so these have both a real and phase degree of freedom. The number of these edges (which is also the number of fundamental loops) is $L=|E_0|-|X|+1$ because the tree $T$ has $|X|-1$ edges and there are $|E_0|$ edges in total, so $L$ many edges that are not in the tree. We also have parallel arguments in the nonAbelian case to be covered later. \end{proof}

\subsubsection{Yang-Mills action for Abelian $G$}

We have seen that, in the case $G$ Abelian and for each $\rho\ne 1$ we have a regular noncommutative $U(1)$-gauge theory. Moreover, we have that the polar decomposition of the $u^\rho_{x\to y}=r^\rho_{x\to y}e^{\imath\theta^\rho_{x\to y}}$ can be viewed as two so far independent theories, for each $\rho\ne 1$. Namely, a regular lattice $U(1)$ for $e^{\imath\theta^\rho_{x\to y}}$ with $U(1)$ gauge transforms at vertices, and a theory for real fields $r^\rho_{x\to y}=r^\rho_{y\to x}>0$ which are already gauge invariant. Here, the Wilson lines factorise as
\[ W_{\rho, x_0\to\cdots x_n}[u]=W^{>0}_{\rho,x_0\to\cdots x_n}[r] W^{U(1)}_{\rho, x_0\to\cdots x_n}[\theta];\quad W^{>0}_{\rho,x_0\to\cdots x_n}[r]=\prod_{i}r^\rho_{x_i\to x_{i+1}};\quad  W^{U(1)}_{\rho,x_0\to\cdots x_n}[\theta]=e^{\imath\sum_i\theta^\rho_{x_i\to x_{i+1}}}.\]
This is in contrast with LGT, as there only the phase degrees of freedom are considered. In Theorem~\ref{thm:Wact} we saw that different representations $\rho\neq 1$ are not coupled to each other. From now on we will restrict ourselves to one representation and will drop  $\rho \neq 1$ everywhere for clarity. We have the following contribution
\[ S^{min}[r,\theta]= - \frac{1}{c} \sum_{ x\to y\to z\to y'\to x}   \lambda_{x\to y\to z\to y'} r_{x\to y}r_{y\to z} r_{z\to y'} r_{y'\to x} e^{i (\theta_{x\to y}+\theta_{y\to z}+\theta_{z\to y'}+\theta_{y'\to x})}\]
which interestingly couples these two theories. Similarly for the other variants with different differential calculi. 

In particular, if we restrict ourselves to tree graphs, the phase degrees of freedom disappear and we find that the theory is equivalent to a massless positive scalar field on a graph $X_{E_0} = E_0$, where the vertices correspond to undirected edges $e = x-y$ of the original graph, with arrows $e \to f$ if $e$ and $f$ share a vertex in $X$, i.e.
\[
y \overset{e}{-}x \overset{f}{-} y' \quad \text{or} \quad y \overset{f}{-}x \overset{e}{-} y'.
\]
We define the inverse metric via the coefficients of $(\ , \ )_2$ as $\lambda_{e\to f} = \lambda_{x\to y \to x\to y'}$ for $e= x-y$, $f= x-y'$, which is well-defined, real and edge symmetric if we restrict the coefficients $\lambda_{x\to y \to z\to y'}$ to be real. This is the case of interest as we expect $(\ , \ )_2$ to be constructed via $(\ , \ )$ on $X$ which has real coefficients. We will also need the graph Laplacian $\Delta\colon \C(X_{E_0}) \to \C(X_{E_0})$ which is defined on $w\in \C(X_{E_0})$ as \cite[Proposition 1.28]{BegMa} 
\[
\Delta w(e) = 2 \sum_{f:e\to f} \lambda_{e\to f} (w(e) - w(f)).
\]

\begin{theorem}
\label{thm:MasslessScalar}
On a tree graph $X$ with set of arrows $E$ and $\lambda_{x\to y \to z\to y'}\in \R$, Abelian Yang Mills theory is equivalent to massless positive scalar field theory on $X_{E_0}$ with arrows and inverse metric as above.
\end{theorem}
\begin{proof} 
The Yang Mills action for a tree graph $X$ has the form
\begin{equation}\label{Swsimply}
S[w] = - \frac{1}{c} \sum_{x\to y \to x \to y'\to x} \lambda_{x\to y \to x \to y'} w_{x\to y}  w_{x\to y'}
\end{equation}
with the Wilson loop variables $w_{x\to y} := W_{x \to y\to x}(u) = r^2_{x\to y}$ satisfying $w_{x\to y} = w_{y\to x}$. There is no contribution from loops $x\to y \to z \to y \to x$ with $x\neq z$, since for such pairs vertices $x,z$ there is only one vertex $y$ between them, otherwise the undirected graph would not be a tree. Together with Lemma \ref{lem:InnerProductOmega2X}, this implies $\lambda_{x\to y\to z \to y} = 0$ for $x\neq z$. We can rewrite $S[w]$ as 
\[
S[w] = - \frac{1}{c}\sum_x \sum_{y,y':x\to y\to x\to y'\to x} \lambda_{x\to y\to x\to y'} w_{x\to y}w_{x\to y'} = 
 - \frac{1}{c}\sum_x\sum_{y:x\to y} \left( \lambda_{x\to y\to x\to y} w^2_{x\to y}  +  \sum_{y'\neq y: x\to y'} \lambda_{x\to y\to x\to y'} w_{x\to y}w_{x\to y'} \right).
\]
Using the relations from  Lemma \ref{lem:InnerProductOmega2X}  $\lambda_{x\to y\to x\to y} = -\sum_{y'\neq y: x\to y\to x\to y'\to x} \lambda_{x\to y \to x\to y'}$ (this follows from taking the complex conjugate of the first relation and then using the second) we get
\[S[w] = \frac{1}{c} \sum_x\sum_{y:x\to y}w_{x\to y} \sum_{y'\neq y:x\to y'} \lambda_{x\to y\to x\to y'}   \left(w_{x\to y}  - w_{x\to y'} \right)
= \frac{1}{2c} \sum_{e\in X_{E_0}} \sum_{f:e\to f} \lambda_{e\to f} w_e(w_e - w_f).
\]
where the $1/2$ comes from
\[
\sum_{e\in X_{E_0}} \sum_{f: e\to f} = \sum_x \sum_{y:x \to y} \sum_{y'\neq y:x\to y'} + \sum_y \sum_{y \to x} \sum_{x'\neq x: y\to x'} = 2\sum_x \sum_{y:x \to y} \sum_{y'\neq y:x\to y'}.
\]
We recognise the term in the action as a quarter of the graph Laplacian on the graph $X_{E_0} = E_0$ with arrows and inverse metric described above, leading to
\[
S[w] = \frac{1}{4c} \int_{X_{E_0}} w  \Delta  w 
\]
now with the positive scalar field $w\colon X_{E_0} \to \R_{>0}$. 
\end{proof}

Due to this equivalence, it is clear that the action has a translation symmetry which shifts $w \to w + \alpha$, $\alpha \in \R_{>0}$, and following the discussion in \cite{FaIva:fms}, that there is a zero mode corresponding to the sum $ \sum_{x\to y} w_{x\to y}$, which will need to be dealt with when quantising the theory. This  points to a kind of spontaneous symmetry breaking in these systems. Also note that for a non-simply connected graph we can just restrict the Yang-Mills action in Theorem~\ref{thm:Wact} to trivial loops, by separately restricting the loops to $z=x$ and $y=y'$, resulting in
\[
- \frac{1}{c} \sum_{x\to y \to x \to y'\to x} \lambda_{x\to y \to x \to y'} w_{x\to y}  w_{x\to y'} - \frac{1}{c} \sum_{x\to y \to z \to y\to x} \lambda_{x\to y \to z \to y} w_{x\to y}  w_{z\to y}
\]
The first term is \ref{Swsimply}, which appears as the kinetic term for a positive scalar field. Using $w_{x\to y} = w_{y\to x}$ and changing the labels, the second term can be rewritten to be
\[
- \frac{1}{c} \sum_{x\to y \to x \to y'\to x} \lambda_{y \to x \to y' \to x} w_{x\to y}  w_{x\to y'},
\]
which can be interpreted as another derivative coupling different from the Laplacian. In the case of $\Z_N$ we have $\lambda_{y \to x \to y' \to x} = 0$ and the additional terms from this vanish, therefore Theorem~\ref{thm:MasslessScalar} still applies as written. Some examples of $X_{E_0}$ including this case are:
\begin{center}    \begin{tabular}{c|c}
          $X$ & $ X_{E_0}$ \\
          \hline
          Starshaped graph with $N$ legs & $\Z_N$, $\CC = \Z_N \backslash\{e\}$ \\
          \hline
          $A_{n+1}$ & $A_{n}$\\
          \hline
          $\Z_N$, $\CC = \{\pm 1\}$ & $\Z_N$, $\CC = \{\pm 1\}$ \\
          \hline
          $\Z$, $\CC = \{\pm 1\}$ & $\Z$, $\CC = \{\pm 1\}$.
    \end{tabular}
    \end{center}

%Here we consider the graphs $\Z_N$ with $\Z_N$, $\CC = \{\pm 1\}$ since, even though they are not simply connected, there are no boundary loop terms in the action and Theorem \ref{thm:MasslessScalar} can be applied.

\subsection{NonAbelian case and example of $G=S_3$ for small graphs $X$}\label{secnona}

In the case of a nonAbelian group $G$, we still have, for each $\rho\ne 1$, a polar decomposition 
\[ u^\rho_{x\to y}= U_{x\to y} r_{x\to y},\quad  U_{x\to y}\in U(d_\rho),\quad r_{x\to y}\in M^{>0}_{d_\rho}(\C),\]
where $M^{>0}_{d_\rho}(\C)$ denotes positive Hermitian matrices. But now  
\[ u_{y\to x}= U_{y\to x}r_{y\to x}=r_{x\to y}U^{-1}_{x\to y}\]
tells us that
\[ r_{y\to x} = \sqrt{U_{x\to y} r^2_{x\to y} U^{-1}_{x\to y}},\quad U_{y\to x}=r_{x\to y}U^{-1}_{x\to y}r_{y\to x}^{-1},\]
which no longer decouples into independent reality constraints on $U,r$ but nevertheless tells us how they are determined under a change of arrow orientation. If $U_{x\to y}=1$ then $r_{y\to x}=r_{x\to y}$ and $U_{y\to x}=1$. 

Similarly, when we analyse $\CU^\times/\CG$, this is also more complicated as the unitary and Hermitian elements do not commute and the theory fails to decouple into positive and unitary parts. We still have the following. 

\begin{proposition}\label{propnona} For a bidirected tree graph, 
\[ \CU^\times/\CG\simeq\prod_{\rho\ne 1} M_{d_\rho}^{>0}(\C)^{E_0}/U(d_\rho)\]
as an assignment for each $\rho\ne 1$ of a positive matrix to every edge modulo an overall global gauge transformation by conjugation. 
\end{proposition}
\proof We consider each $\rho\ne 1$ at a time, of some dimension $d$. We follow the same steps as in the proof of  Proposition~\ref{propsimply}, but this time $u_{x_i\to x_{i+1}}=U_{x_i\to x_{i+1}}r_{x_i\to x_{i+1}}$ is a matrix polar decomposition. Starting a chain at $x_0$, we choose $\gamma_{x_0}$ randomly and choose $\gamma_{x_1} = \gamma_{x_0}U_{x_0\to x_1}$ so that 
\[ \gamma_{x_0}U_{x_0\to x_1}r_{x_0\to x_1}\gamma_{x_1}^{-1}=\gamma_{x_0}U_{x_0\to x_1}\gamma_{x_1}^{-1}r'_{x_0\to x_1}=r'_{x_0\to x_1}\]
where $r'_{x_0\to x_1}:=\gamma_{x_1}r_{x_0\to x_1}\gamma_{x_1}^{-1}$ is a conjugate of $r_{x_0\to x_1}$ and hence another positive matrix. This puts $\gamma_{x_1}$ at the start of $x_1\to x_2$ and we repeat the process by choosing $\gamma_{x_2}$ to render $r'_{x_1\to x_2}$ alone.  Then dropping the primes, we see that every connection can be gauge transformed to a `positive' one where there is only a matrix  $r_{x\to y}\in M_d^{>0}(\C)$ at each edge. The reverse direction necessarily gauge transforms correctly, but one can check for example at the first step that $u_{x_1\to x_0}$ gauge transforms to 
\[ \gamma_{x_1}U_{x_1\to x_0}r_{x_1\to x_0}\gamma^{-1}_{x_0}=\gamma_{x_1}r_{x_0\to x_1}\gamma_{x_1}^{-1}\gamma_{x_1}U^{-1}_{x_0\to x_1}\gamma^{-1}_{x_0}=r'_{x_0\to x_1}=r'_{x_1\to x_0}.\]

Next, a gauge transform $\{\gamma_{x}\}$ connecting two positive configurations $\{r_{x\to y}\}$ and $\{r'_{x\to y}\}$ requires 
\[ \gamma_x\gamma_y^{-1} (\gamma_y r_{x\to y}\gamma_y^{-1})= r'_{x\to y}.\]
But $\gamma_y r_{x\to y}\gamma_y^{-1}$ is positive and by uniqueness of factorisation of $r'_{x\to y}=1\cdot r'_{x\to y}$ into an unitary times a positive matrix, we see that $\gamma_x\gamma_y^{-1} =1$ and $\gamma_y r_{x\to y}\gamma_y^{-1}= r'_{x\to y}$. \endproof

 As an example, we can take  $G=S_3$, which has 2 nontrivial representations $\rho$, one of them 1-dimensional so a noncommutative $U(1)$ as before, and the other 2-dimensional, which we focus on.

\begin{lemma} $M_2^{>0}(\C)$ can be identified with the positive light cone in $\R^{1,3}$ Minkowski space. Conjugation by $U(2)$ is then a spatial rotation. 
\end{lemma}
\proof We use the standard picture of $\R^{1,3}$ as the Hermitian elements of $M_2(\C)$ via Pauli matrices according to $X=t {\rm id} + \sum_i x^i \sigma^i$. Here $\det(X)=t^2-\sum_i (x^i)^2$ and $t={1\over 2}{\rm Tr}(X)$. If $X$ is a positive Hermitian matrix then both of these are positive, while conversely if both of these are positive then the eigenvalues have the same sign and their sum is positive, hence $X$ is positive. Meanwhile, conjugation of $\sum_i x^i\sigma^i$ rotates the vector $\{x^i\}$.  \endproof

By similar arguments as in the proof for Proposition \ref{propnona}, we have a parallel result for a general graph $X$. As before, we let $T$ be a maximal spanning tree  and $L=|E_0|-|X|+1$ the number of  fundamental loops of the undirected graph.

\begin{proposition}\label{propplanar}
Let $G$ be a nonAbelian finite group and $X$ a bidirected planar graph with $L$ fundamental loops. Then $\CU^\times / \CG \simeq \prod_{\rho\neq 1}(M^{>0}_d(\C)^{|E_0|} \times U(d)^L)/U(d)$.
\end{proposition}

Rather than giving a formal proof, we illustrate how this works for $\Z_3$ and for a planar graph with two loops. For $\Z_3$, we choose $\gamma_1$ randomly, choose $\gamma_2$ to transform $u_{1\to 2}$ to $r'_{1\to 2}$. We similarly choose $\gamma_3$ to do the same for $u_{2\to 3}$. We took these two edges as the spanning tree. However, for the last edge (which provides the fundamental loop around the triangle), we are stuck with $\gamma_3 U_{3\to 1}r_{3\to 1}\gamma_1^{-1}=U'_{3\to 0}r'_{3\to 1}$ without being able to set $U'_{3\to 1}=1$. Hence every regular connection can be gauge transformed to one with one unitary and three positive matrices on the arrows (in one direction) like this. If we were to further gauge transform such a configuration then by arguments as above, we would need $\gamma_1=\gamma_2=\gamma_3$, so a global one. Hence, the moduli space for each $U(d)$-gauge theory component is $(M_d^{>0}(\C)^3\times U(d))/U(d)$, where the latter  $U(d)$ acts by conjugation. It is important to note that the representative of $u$ in $\CU^\times / \CG$ depends on the edge where one does not erase the polar part of the holonomy. For two loops,  we consider two triangles
 \[
\begin{tikzcd}
 & 1 \arrow[rd,shift left]\arrow[ld,shift left] \arrow[dd,shift left] & \\
2 \arrow[rd,shift left]\arrow[ru,shift left ] & {}  & 4 \arrow[ld,shift left] \arrow[lu,shift left]\\
  & 3\arrow[ru,shift left] \arrow[lu,shift left] \arrow[uu,shift left] & \\
\end{tikzcd}
\]
and take spanning tree starting at 1 and branching to 2,3,4. Then following the steps as in the proof of Proposition~\ref{propnona}, we  choose $\gamma_1 = e$ and then $\gamma_2$, $\gamma_3$, $\gamma_4$ to fix the holonomies $u_{1\to 2}$, $u_{1\to 3}$ and $u_{1\to 4}$ to be in $M^{>0}_d$. There are then no more $\gamma_i$ to fix remaining  $u_{2\to 3}$, $u_{3\to 4}$, so they will be stuck in $M^{>0}_d \times U(d)$. These correspond to each triangle as a fundamental loop. As before, a gauge transformation connecting two such configurations needs to be a global one. We see this by restricting the argument presented before to the spanning tree without the arrows $2 \leftrightarrow 3$, $3\leftrightarrow 4$.

\section{Abelian quantum gauge field theory via functional integrals}\label{secQFT}

In this section, we look at quantum gauge field theory (QGFT) for our Yang-Mills actions on graphs using a functional integral approach. Because a graph is a perfectly good (but noncommutative) geometry, we see the finite graph and discrete lattice cases as exact theories which are not necessarily approximations of a continuum theory, and where the finite case can be fully computed numerically. We focus on the case of Abelian structure group $G$. The choice $G=\Z_2$ is equivalent as we have seen above to noncommutative $U(1)$-gauge theory, but this already carries the flavour of the nonAbelian case.

\subsection{Measure on $\CU^\times$ and partition function}
\label{sec:measureAb}

To start, we need to define an integration measure on the space of connections $\CA$. As this is a finite dimensional vector space in our setting, one can simply take the Lesbegue measure. Note that this will also be the Lesbegue measure on $\CU$, as the map $\CA \to \CU$ is given by a shift. In the case of Abelian gauge theory we have the isomorphism $\CU^\times \simeq (\C^\times)^{|E_0|(|G|-1)}$ via the Fourier transform, and thus can define the integral on $\CU^\times$ by decomposing $\C^\times$ into polar coordinates
\begin{align*}
	\int_{\CU^\times} \CD u \coloneqq  \int_{(\R_{>0} \times U(1))^{|E_0|(|G|-1)}} \prod_{\rho\neq 1} \prod_{x \leftrightarrow y}  r^\rho_{x\to y} \extd r^\rho_{x\to y} \extd \theta^\rho_{x\to y} 
\end{align*}
where $\prod_{x \leftrightarrow y}$ means that an orientation on the arrows is chosen, i.e. either $x\to y$ or $y\to x$. Choosing differently only changes the integration by an overall sign, as $r^\rho_{x\to y} = r^\rho_{y\to x}$, $\theta^\rho_{x\to y} = -\theta^\rho_{y\to x}$, which is inconsequential when computing expectation values. Furthermore, gauge transformations change $\theta^\rho_{x\to y}$ by a shift. The measure is therefore gauge invariant and descends to the moduli space $\CU^\times / \CG$. Using the Wilson loop variables $w^\rho_{x\to y} := W^\rho_{x \to y\to x}(u) = (r^\rho_{x\to y})^2$ results in the partition function
\[
	\CZ = \int_{\CU^\times / \CG} \CD[u] e^{-S[u]} = \prod_{\rho \neq 1} \int_{\R^{|E_0|}_{>0}\times U(1)^L} \prod_{x\leftrightarrow y} \prod^L_{l=1} \extd w^\rho_{x\to y} \extd \theta^\rho_{l}  e^{-S[w,\theta]}
\]
for a planar graph $X$ where the index $l$ runs over the fundamental loops of $X$. We see that the quantised Abelian gauge theory on $X$ is just $|G|-1$ copies of $\C\Z_2$-gauge theory on $X$, we therefore will drop the $\rho$ index and will only consider the theory at a fixed representation from now on.

To be concrete, we fix the exterior algebra to $\Omega_{min}$, so the action in this section will be the $S=S_{min}$ case. This obeys 
$S[w,\theta] = S_{nt}[w,\theta] + S_{t}[w]$ where the first term is a sum over all non-trivial loops appearing in the action (i.e. fundamental loops) and the second over the trivial ones, the latter of which does not depend on the phases $U(1)$, as these cancel out over trivial loops due to $u_{x\to y} = u^*_{y\to x}$. It is interesting to note that the phase degrees of freedom can be integrated out to give an effective theory for the $w$ variables
\[
\CZ = \int \CD w \CD \theta e^{-S_{nt}[w,\theta] - S_{t}[w]}.
\]
Let $L_{nt}$ denote the set of non-trivial loops of the graph $X$, then the non-trivial loop part of the action can be written as
\begin{align*}
S_{nt}[r,\theta]  &= -\frac{1}{c} \sum_{x\to y\to z\to y'\to x\in L_{nt}} \mathrm{Re}[ \lambda_{x\to y\to z\to y'}W^{>0}_{x\to y\to z\to y'\to x}[r] W^{U(1)}_{x\to y\to z\to y'\to x}[\theta]]\\
&=  - \frac{1}{c} \sum_{x\to y\to z\to y'\to x\in L_{nt}}
W^{>0}_{x\to y\to z\to y'\to x}[r]( \mathrm{Re}[ \lambda_{x\to y\to z\to y'} ]\cos(\theta_{x\to y\to z\to y'\to x} ) - \mathrm{Im}[ \lambda_{x\to y\to z\to y'} ]\sin(\theta_{x\to y\to z\to y'\to x} ))\\
&=-\frac{1}{c} \sum_{x\to y\to z\to y'\to x\in L_{nt}}| \lambda_{x\to y\to z\to y'}| W^{>0}_{x\to y\to z\to y'\to x}[r]\cos(\varphi^\lambda_{x\to y\to z\to y'} + \theta_{x\to y\to z\to y'\to x} ) \end{align*}
where in the first equality we use that the terms for $x\to y\to z\to y' \to x$ and $x\to y'\to z\to y \to x$ are the complex conjugate to one another, and take $\varphi^\lambda_{x\to y\to z\to y'} $ to be the phase of $\lambda_{x\to y \to z\to y'}$. As the phase degrees of freedom over different fundamental loops on a planar graph are independent, we can explicitly compute the integrals to give the effective contribution
\[
S_{nt,eff}[r]  = - \ln \left( \int \CD \theta e^{-S_{nt}^{\rho}[r,\theta]} \right) = -\sum_{x\to y\to z\to y'\to x\in L_{nt}}\ln\left( 2\pi I_{0}(c_\rho|\lambda_{x\to y\to z\to y'}| W^{>0}_{x\to y\to z\to y'\to x}[r] )\right)
\] 
where $I_0$ is the modified Bessel function. This generalises the results of \cite{Ma:cli}, where the partition function of noncommutative $U(1)$-gauge theory was also found to be proportional to the Bessel function after integrating out the phase degrees of freedom.

\subsection{Abelian QGFT on Euclidean $A_{n+1}$ and $\Z_n$}

Here we give results for  $\C\Z_2$-gauge theory on $X$ a discrete open line $\bullet-\bullet-\cdots-\bullet$ and $X$ a discrete circle, both with $n$ edges and represented by the Dynkin graph $A_{n+1}$ and the $n$-gon $\Z_n$ respectively. We take the Euclidean metric $\lambda_{i\to i\pm1 }=1$ for both, which by Proposition \ref{eq:SolLambda(,)2} with $A_{i,i} = 2$ gives $\lambda_{i\to i+1\to i \to i+1} =-2$ for $A_{n+1}$ and $\Z_n$. In terms of the Wilson loop variables $w_i \coloneqq u_{i\to i+1} u_{i+1\to i} = r^2_{i\to i+1}$ the actions as in Examples \ref{ex:ActionAn1}, \ref{ex:ActionZ} are
\begin{equation}
\label{eq:ActionsAn+1ZnWilsonLoop}
S_{A_{n+1}}[w] = \frac{2}{c} \sum^{n-1}_{i=1} (w_i - w_{i+1})^2 ,
\quad 
S_{\Z_n}[w] = \frac{2}{c} \sum^{n}_{i=1} w_i (w_i - w_{i-1}),
\end{equation}
where in the latter case the index is mod $n$. As discussed after Theorem \ref{thm:MasslessScalar}, these are the actions for a massless positive scalar field $w$ on $A_n, \Z_n$ respectively, and have a zero mode corresponding to the sum $\sum^{n}_{i=1} w_i$ which makes the path integral (IR) divergent. We regulate this divergence by introducing a scale $L$ as an upper value for the zero mode
\[
\CZ_{A_{n+1},L} = \int_{\R^n_{>}} \extd^n w \,\,\, \chi_{[0,L]}\left(\sum^{n}_{i=1} w_i\right)\, e^{-S_{A_{n+1}}[w]},\quad \CZ_{\Z_{n},L} = 2\pi \int_{\R^n_{>}} \extd^n w \,\,\, \chi_{[0,L]}\left(\sum^{n}_{i=1} w_i\right)\, e^{-S_{\Z_n}[w]}
\]
where $\chi_{[0,L]}$ is the characteristic function for the interval $[0,L]$. Here $\CZ_{\Z_{n},L}$ has an extra $2\pi$ factor due to the extra phase in the moduli space as seen in Proposition \ref{proploop}. To make the divergence explicit, first make a coordinate change to increasing coordinates $s'_1 \leq s'_2 \leq \dots \leq s'_n$, $s'_i = \sum^i_{j=1} w_j$ and then to the coordinates
\begin{equation}
    \label{eq:VariablesS}
	s_i = \frac{s'_i}{s'_{i+1}} = \frac{ \sum^i_{j=1} w_j}{ \sum^{i+1}_{j=1} w_j},
	\quad \quad
	 s_n = \frac{s'_n}{L} = \frac{ \sum^n_{j=1} w_j}{L}
\end{equation}
where $i=1,\dots,N-1$. Changing to the new coordinates results in the following the partition function for the $A_{n+1}$ graph
\[
\CZ_{A_{n+1},L} = \int^L_0 \extd s'_n \int^{s'_{n}}_0 \extd s'_{n-1} \dots \int^{s'_3}_0 \extd s'_2 \int^{s'_{2}}_0 \extd s'_{1}e^{-S_{A_{n+1}}[s']}  = L^{n-1} \int_{[0,1]^n} \CD s e^{-S_{A_{n+1}}[s]} 
\]
with the measure $\CD s = s^{n-1}_n  s^{n-2}_{n-1} \cdots s^2_3 s_2 \extd^n s$, similarely for $\Z_n$. The Wilson loop variables can be recovered via
\[
w_1 = s_1 s_2 \cdots s_n L, \quad \quad w_i = (1-s_{i-1})  s_i \cdots s_n L
\]
which in turn results in the following actions in the new coordinates
\[
S_{A_{n+1}}[s] =  S_{A_{n+1},R}[s]/c_R;\quad c_R= \frac{c}{2L^2},\quad   S_{A_{n+1},R}[s]= \sum^{n-1}_{i=1}((2-s_{i-1}) s_i -1)^2 s^2_{i+1} \cdots s^2_{n} 
\]
where we have set $s_0 = 0$ and use the rescaled coupling $c_R$ and $S_{A_{n+1},R}[s]$ as stated. For $\Z_n$ we similarly have $S_{\Z_n}[s] =  S_{\Z_n,R}[s]/c_R$ with
\[
S_{\Z_n,R}[s]= (s^2_1\cdots s^2_{n-1} - s_1 \cdots s_{n-1}(1-s_{n-1})) s^2_n + (1-3s_1+2s^2_1 )s^2_2\cdots s^2_n + \sum^{n}_{i=3}(1-3s_{i-1} + 2 s^2_{i-1} + s_{i-1} s_{i-2} -s^2_{i-1}s_{i-2})s^2_{i} \cdots s^2_{n} 
\]

It is now clear that the only divergent part of $\CZ_L$ for both $A_{n+1}$ and $\Z_{n}$ is given by the factor $L^{n-1}$, and that expectation values are finite if the observables do not depend on $L$, i.e. if they are functions of the rescaled Wilson loop variables
\begin{equation}
\label{eq:WilsonLoopInSvariables}
w_{R,1} = \frac{w_1}{L} = s_{1} s_2 \cdots s_n,
\quad \quad
w_{R,i} = \frac{w_i}{L} = (1-s_{i-1}) s_i \cdots s_n.
\end{equation}
The expectation values of such observables as functions of $c_R$,
\[
\langle \CO \rangle_{A_{n+1}} = \frac{\int_{[0,1]^n} \CD s \CO e^{- S_{A_{n+1}}[s]/c_R} }{\int_{[0,1]^n} \CD s e^{-S_{A_{n+1}}[s]/c_R}},
\quad
\langle \CO \rangle_{\Z_{n}} = \frac{\int^{2\pi}_0\extd \theta \int_{[0,1]^n} \CD s \CO e^{- S_{\Z_{n}}[s]/c_R} }{2\pi \int_{[0,1]^n} \CD s e^{-S_{\Z_{n}}[s]/c_R}}
\]
are indeed finite and independent of $L$, only depending on the coupling $c_R$.

\subsubsection{Exact results for $A_3$}
\label{sec:ExactA3}

The simplest case $n=2$ has $S=0$ for $X=\Z_2$ according to (\ref{eq:ActionsAn+1ZnWilsonLoop}), while the finite line $X=A_3$ $\bullet - \bullet - \bullet$ is simple enough to be solved exactly. Here, the action is $S_{R}[s] = (2s_1-1)^2 s^2_2 = 4 \tilde s^2_1 s_2^2$ with $\tilde s^2_1 = s_1 - \frac{1}{2}$, and we can compute the partition function and expectation values of Wilson loops $w_1,w_2$ to be
\begin{equation}
    \label{eq:PartFunA3}
\CZ_{A_3} = \int^1_0 \extd s_2 s_2 \int^{\frac{1}{2}}_{-\frac{1}{2}} \extd \tilde s_1 e^{-\frac{4}{c_R}\tilde s^2_1 s_2^2 } 
= \frac{1}{2} c_R \left( \sqrt{\frac{\pi}{c_R}} \text{erf}\left(\sqrt{\frac{1}{c_R}}\right)+e^{-\frac{1}{c_R}}-1\right),
\end{equation}
\begin{equation}
    \label{eq:ExpValwA3}
\langle w_1 \rangle = \langle w_2 \rangle = \frac{\sqrt{\pi } (c_R-2) e^{\frac{1}{c_R}} \text{erf}\left(\frac{1}{\sqrt{c_R}}\right)-2 \sqrt{c_R}}{8
\sqrt{c_R} \left(e^{\frac{1}{c_R}}-1\right)-8 \sqrt{\pi } e^{\frac{1}{c_R}}
\text{erf}\left(\frac{1}{\sqrt{c_R}}\right)}
\end{equation}
with variances $(\Delta w_i)^2 = \langle w^2_i \rangle-\langle w_i \rangle^2$ and connected correlation $\langle w_1 w_2 \rangle^c = \langle w_1 w_2 \rangle - \langle w_1 \rangle \langle w_2 \rangle$
\[     (\Delta w_i)^2 = \frac{\left(\frac{\sqrt{\pi } (3 c_R+2) \text{erf}\left(\frac{1}{\sqrt{c_R}}\right)}{\sqrt{c_R}}-8 c_R+(8 c_R+2) e^{-1/c_R}\right) F(x)-\frac{3 e^{-2/c_R} \left(\sqrt{\pi } (c_R-2) e^{\frac{1}{c_R}} \text{erf}\left(\frac{1}{\sqrt{c_R}}\right)-2 \sqrt{c_R}\right)^2}{4 c_R}}{12\, F(x)^2}
\]
\[     
    \langle w_1 w_2 \rangle^c =\frac{e^{-1/c_R} \left(-\frac{\sqrt{\pi } (3 c_R-2) e^{\frac{1}{c_R}}
    \text{erf}\left(\frac{1}{\sqrt{c_R}}\right)}{\sqrt{c_R}}+4 c_R
    \left(e^{\frac{1}{c_R}}-1\right)+2\right) F(x)-\frac{3 e^{-2/c_R} \left(\sqrt{\pi } (c_R-2)
    e^{\frac{1}{c_R}} \text{erf}\left(\frac{1}{\sqrt{c_R}}\right)-2 \sqrt{c_R}\right)^2}{4
    c_R}}{12 \, F(x)^2}
\]
where $F(x) \coloneqq \left(E_{\frac{3}{2}}\left(\frac{1}{c_R}\right)+2
\sqrt{\frac{\pi }{c_R}}-2\right)$. Here $\text{erf}(x) = \int_{-\infty}^{x} e^{-t^2} \extd t$, $E_{\frac{3}{2}}\left(x\right) = \int_{1}^{\infty} e^{-xt}t^{-\frac{2}{3}} \extd t$ are the error function and the exponential integral function of integer order for $n=2/3$. We plot these in Figure~\ref{fig:ExactResultsA3}, with grey dotted and dashed lines representing the $c_R \to 0, \infty$ limits respectively as determined in the following sections.

\begin{figure}[h!]
    \centering
    \includegraphics[width=1\textwidth]{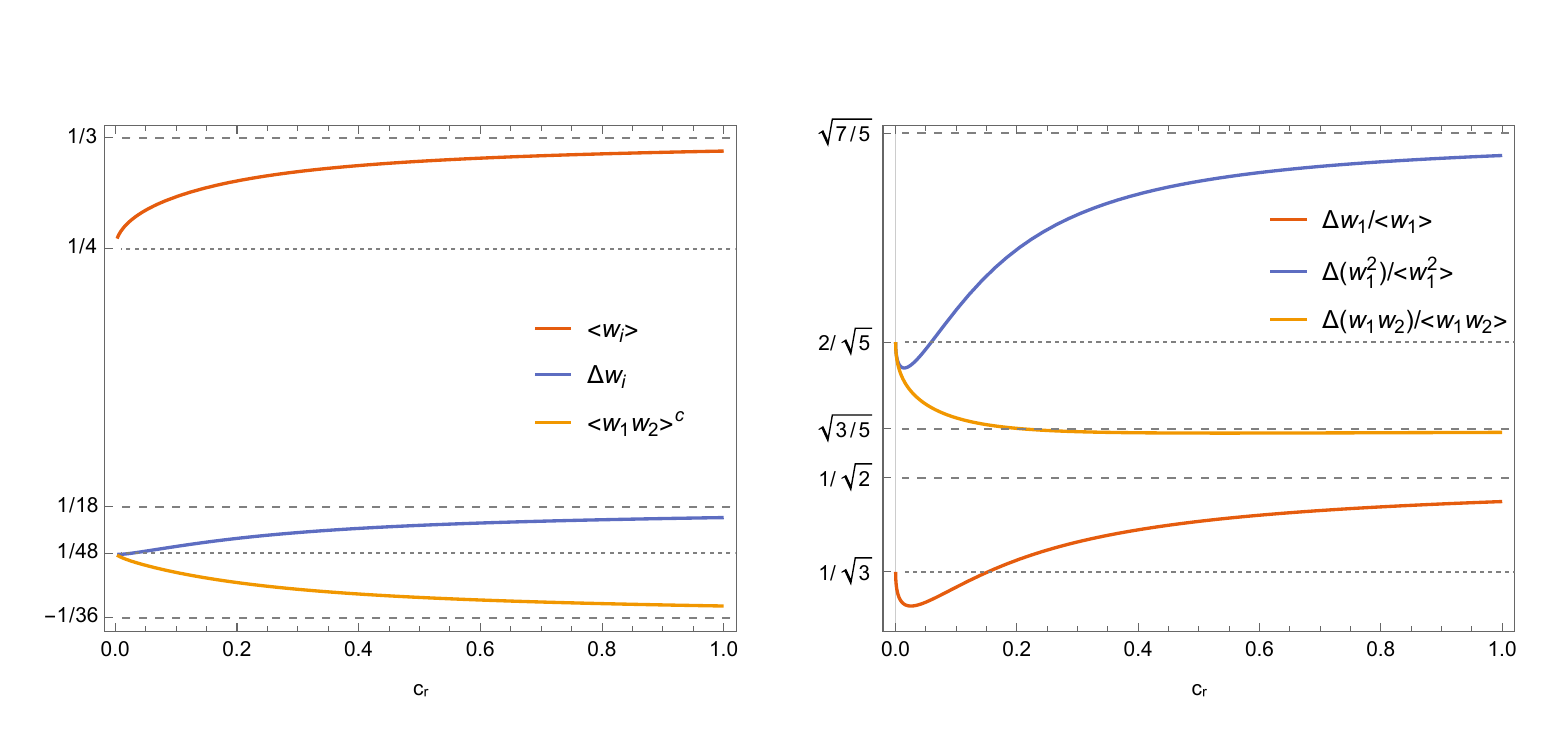}
    \caption{Wilson loop correlation functions for  $X=A_3=\bullet-\bullet-\bullet$ plotted against the coupling constant $c_R$. The dotted and dashed grey lines denote the $c_R \to 0,\infty$ limits, respectively.}
    \label{fig:ExactResultsA3}
\end{figure}

\subsubsection{Expansion at strong coupling}

We next look expansions for strong and weak coupling $c_R=c/(2 L^2)$, starting here with the strong case. We aim to compute the $c_R\to \infty$ limit of $\C\Z_2$-gauge theory on $A_{n+1}$ and $\Z_n$. Expanding the exponential to first order in $1/c_R$ and writing
\[
\langle \CO \rangle = \frac{\langle \CO \rangle _{\infty} -  \langle \CO S \rangle _{\infty}/c_R +O(c^{-2}_R)}{1 -  \langle  S \rangle _{\infty}/c_R +O(c^{-2}_R)} = \langle \CO \rangle _{\infty} +\frac{1}{c_R} \left(\langle \CO \rangle _{\infty}  \langle S \rangle _{\infty} - \langle \CO  S\rangle _{\infty} \right)+O(c^{-2}_R),
\]
where $\langle \ \rangle _{\infty}$ denotes the expectation value in the limit $c_R\to \infty$ with $e^{-S_R/c_R} = 1$. Note that at 0th order, the expectation value does not depend on the action functional, and therefore both the $A_{n+1}$ and $\Z_n$ theories result in the same expressions. For simplicity, we drop the index $R$ on the rescaled Wilson loop variables.

\begin{proposition}
\label{prop:crinftyLimitExpVal}
The expectation values of Wilson loops at first order in $1/c_R$ for $A_{n+1}$, $\Z_n$ are given by
\begin{align*}
	\langle w^{m_1}_1 w^{m_2}_2 \cdots w^{m_n}_n  \rangle_{A_{n+1}} &= \binom{n+M}{m_1,\dots,m_n,n}^{-1}\left(1+ \frac{1}{c_R} \left(2(n-1) \frac{n!}{(n+2)!} - C^{A_{n+1}}_1(w^{m_1}_1 w^{m_2}_2 \cdots w^{m_n}_n) \right)\right) +O(c^{-2}_R),\\
	\langle w^{m_1}_1 w^{m_2}_2 \cdots w^{m_n}_n  \rangle_{\Z_{n}} &= \binom{n+M}{m_1,\dots,m_n,n}^{-1}\left(1+ \frac{1}{c_R} \left(n \frac{n!}{(n+2)!} - C^{\Z_{n}}_1(w^{m_1}_1 w^{m_2}_2 \cdots w^{m_n}_n) \right)\right)+O(c^{-2}_R)
\end{align*}
for $m_i \in \N_0$, $M = \sum_{k} m_k$ the length of the Wilson loop, the prefactor the multinomial coefficients and 
\begin{align*}
	C^{A_{n+1}}_1(w^{m_1}_1 w^{m_2}_2 \cdots w^{m_n}_n) &= \frac{(n+M)!}{(n+M+2)!} \left(2(n+M) -m_1 - m_n + \sum^{n-1}_{i=1}(m_i - m_{i+1})^2\right),\\
	C^{\Z_{n}}_1(w^{m_1}_1 w^{m_2}_2 \cdots w^{m_n}_n) &= \frac{(n+M)!}{(n+M+2)!} \left(n+M + \sum^{n}_{i=1} m_i(m_i-m_{i-1}) \right)
\end{align*}
\end{proposition}

\begin{proof}
We show the computation for $\langle w^{m_1}_1 w^{m_2}_2 \cdots w^{m_n}_n  \rangle_{\infty}$, which does not depend on the model. After this, the computations of $\langle S  \rangle_{\infty},$ $\langle w^{m_1}_1 w^{m_2}_2 \cdots w^{m_n}_n  S \rangle_{\infty}$ follow in a similar manner. Start by computing the partition function in the limit $c_R\to \infty$
\[
\CZ_\infty = \lim_{c_R\to \infty}\CZ = \prod^{n}_{i=1}  \int^1_0 s^{i-1}_i\extd s_i = \frac{1}{n!}.
\]
For the expectation value of any Wilson loop we use \eqref{eq:WilsonLoopInSvariables} and introduce the sums $M_i = \sum^i_{k=1} m_k$ to give
\begin{align*}
\langle w^{m_1}_1 w^{m_2}_2 \cdots w^{m_n}_n  \rangle _{\infty} &= \langle \prod^n_{i=1} s^{M_i}_i (1-s_i)^{m_{i+1}} \rangle _{\infty}
= n! \prod^n_{i=1}  \int^1_0 s^{M_i + i -1}_i (1-s_i)^{m_{i+1}} \extd s_i = n! \prod^n_{i=1}  \frac{m_{i+1}!(M_i+i-1)!}{(M_{i+1} +i)!}\\
&= n!\frac{m_1! \cdots m_n!}{(n+M)!}
\end{align*}
with $m_{n+1}=0$. This results in the expressions $\langle w^{m_1}_1 w^{m_2}_2 \cdots w^{m_n}_n  \rangle_{\infty} = m_1!m_2!\cdots m_n! n! / (n+M)!$ for both $A_{n+1},\Z_n$ at $c_R\to \infty$, which is the inverse of the multinomial coefficient as shown above. To compute the first order, we can write the respective actions in terms of the rescaled Wilson loop variables similar to \eqref{eq:ActionsAn+1ZnWilsonLoop} and use the formula for  $\langle w^{m_1}_1 w^{m_2}_2 \cdots w^{m_n}_n  w^2_i \rangle_{\infty}$ etc to compute 
\[ \langle w^{m_1}_1 w^{m_2}_2 \cdots w^{m_n}_n  S\rangle_{\infty} = \langle w^{m_1}_1 w^{m_2}_2 \cdots w^{m_n}_n \rangle_{\infty} C^{A_{n+1},\Z_n}_1( w^{m_1}_1 w^{m_2}_2 \cdots w^{m_n}_n),\]
with $\langle S \rangle_\infty$ being  recovered by setting $m_i=0$ for all $i$. 
\end{proof}

At 0th order, the expectation values for the $A_{n+1}$ and $\Z_n$ coincide in the limit. Some examples of expectation values and standard deviations for small loops are in Table~\ref{table1}. The variances and connected correlations are
\[
    (\Delta w_i)^2 = \frac{n}{(n+1)^2 (n+2)},
\quad
\langle w_i w_j \rangle^c_{\infty} = -\frac{1}{(n+1)^2(n+2)},
\quad i\neq j
\]
with non-trivial values for $A_{n+1}, \Z_n$. In fact, from the connected correlation functions we see that long range order appears for any $n$, and do not depend on the distance $|i-j|$. We see different behaviours at $i=j$ and $i\neq j$, the negative correlation functions $\langle w_i w_j\rangle_\infty^c$ implying that increasing $w_i$ lowers $w_j$ and vice versa. 

\begin{table}
    \centering
    \begin{tabular}{|c|c|c|c|}
        \hline  
        $\CO$ & $\langle \CO \rangle_{\infty}$ & $\Delta(\CO)$  \\
          \hline
          $w_i$ & $\frac{1}{n+1}$ & $\frac{1}{n+1} \sqrt{\frac{n}{n+2}}$ \\
          \hline
          $w^2_i$ & $\frac{2}{(n+1)(n+2)}$ & $ 2 \sqrt{\frac{n (5 n+11)}{\left(n^2+3 n+2\right)^2 \left(n^2+7 n+12\right)}}$ \\
          \hline
          $w_iw_j$ & $\frac{1}{(n+1)(n+2)}$ & $\sqrt{\frac{3 n^2+5 n-4}{\left(n^2+3 n+2\right)^2 \left(n^2+7 n+12\right)}}$ \\
          \hline
    \end{tabular}
    \caption{Wilson loop correlation functions for $X=A_{n+1}, \Z_n$ in strong coupling  $c_R\to \infty$ limit}
    \label{table1}
\end{table}

Interestingly, in the $n\to \infty$ limit where one considers $\C\Z_2$-gauge theory on the integer line $\Z$, all expectation values, (connected) correlations and standard deviations computed above vanish. The only `quantum' effects that remain are seen through the relative uncertainties $\Delta \CO / \langle \CO \rangle$, which for the observables above converge to the following values
\[\lim_{n\to \infty }\frac{\Delta(w_i)}{\langle w_i \rangle_{\infty}} = 1,
\quad 
\lim_{n\to \infty }\frac{\Delta(w^2_i)}{\langle w^2_i \rangle_{\infty}} = \sqrt{5},
\quad
\lim_{n\to \infty }\frac{\Delta(w_i w_j)}{\langle w_i w_j \rangle_{\infty}} = \sqrt{3},\]
giving a measurable prediction in the theory of quantum behaviour. In general we can compute by using Proposition \ref{prop:crinftyLimitExpVal} that 
\[\lim_{n\to \infty} \frac{\Delta( w^{m_1}_1 w^{m_2}_2 \cdots w^{m_j}_j )}{\langle w^{m_1}_1 w^{m_2}_2 \cdots w^{m_j}_j  \rangle_{\infty}} = \sqrt{\frac{(2m_1)! \cdots (2 m_j)!}{(m_1!)^2 \cdots (m_j!)^2}-1}. \]
We see that the QGFT here is nontrivial,  in contrast usual LGT on the $\Z$ lattice where the theory would be trivial for any choice of gauge group.

For simplicity, we now focus on $\Z_n$ to understand the effect of first order terms (for the $A_{n+1}$ theory, one would have to compute many more cases for $C^{A_{n+1}}_1$ due to boundary terms.) We have
\[
C^{\Z_n}_1(w_i) = \frac{1}{n+3},
\quad
C^{\Z_n}_1(w^2_i) = \frac{n+6}{(n+3)(n+4)},
\quad
C^{\Z_n}_1(w_i w_{i-1}) = \frac{1}{n+4},
\quad
C^{\Z_n}_1(w_i w_j) = \frac{1}{n+3},
\]
where we take $|i-j|>1$. Here we see that the long range order starts to change as $C^{\Z_n}_1(w_i w_{i-1}) \neq C^{\Z_n}_1(w_i w_j)$. One can then follow the computations to the first order in $1/c_R$, and find 
\begin{align*}
    \langle w_i \rangle&= \frac{1}{n+1} -\frac{2}{(n+1)^2 (n+2) (n+3)}\frac{1}{c_R} +O(c_R^{-2})\overset{n\to \infty}{\longrightarrow} 0,\\
    \langle w^2_i \rangle&= \frac{2}{(n+1) (n+2)}-\frac{4 \left(n^2+4 n+6\right)}{(n+1)^2 (n+2)^2 (n+3) (n+4)} \frac{1}{c_R}+O(c_R^{-2})\overset{n\to \infty}{\longrightarrow} 0,\\
    \langle w_i w_{i-1} \rangle^c_{c_R} &=-\frac{1}{(n+1)^2 (n+2)}+\frac{n^3+6 n^2+19 n+26}{(n+1)^3 (n+2)^2 (n+3) (n+4)}\frac{1}{c_R}+O\left(c^{-2}_R\right)\overset{n\to \infty}{\longrightarrow} 0,\\
    \langle w_i w_{j} \rangle^c_{c_R} &= -\frac{1}{(n+1)^2 (n+2)}+\frac{2}{(n+1)^3 (n+2)^2}\frac{1}{c_R}+O\left(c^{-2}_R\right)\overset{n\to \infty}{\longrightarrow} 0\\
    \Delta(w_i) &= \frac{1}{n+1}\sqrt{\frac{n}{(n+2)}}-\frac{2 \left(n (n+2)^2-2\right)}{ \left(\sqrt{n} (n+1)^2 (n+2)^{3/2} (n+3) (n+4)\right)}\frac{1}{c_R}+O(c_R^{-2}) \overset{n\to \infty}{\longrightarrow} 0,\\
    \frac{\Delta(w_i)}{\langle w_i \rangle} &= \sqrt{\frac{n}{n+2}}+\frac{4-2 n (n+2)}{ \sqrt{n} (n+2)^{3/2} (n+3) (n+4)} \frac{1}{c_R} +O(c_R^{-2}) \overset{n\to \infty}{\longrightarrow} 1
\end{align*}
which shows that first order corrections make near neighbours be stronger correlated than further away ones $\langle w_i w_{i-1} \rangle^c_{c_R} > \langle w_i w_{j} \rangle^c_{c_R}$, and we also see that first order corrections do not change the $n\to \infty$ behaviour of the theory.

\subsubsection{Expansion at weak coupling}

We start by comparing the perturbative expansion in the weak limit for $X=A_3$ with the exact results from section \ref{sec:ExactA3}. The partition function \eqref{eq:PartFunA3} and expectation values for the Wilson loops \eqref{eq:ExpValwA3} can be approximated in the $c_R\to 0$ limit as 
\[\CZ_{A_3} \approx \frac{\sqrt{\pi c_R}}{2}-\frac{c_R}{2},\quad \quad \langle w_1 \rangle = \langle w_2 \rangle \approx \frac{1}{4}\left(1 + \sqrt{\frac{c_R}{\pi }} + \frac{c_R}{\pi }\right)-\frac{c_R}{8},\]
where we have ignored terms of the form $c^m_R e^{-1/c_R}$ as these are rapidly surpressed for $c_R\to 0$. Using Mathematica, the exact expression for $\langle w^{m_1}_1 w^{m_2}_2 \rangle$ for $m_1,m_2=1,\dots,10$ can be computed exactly and then approximated at low $c_R$, to give 
\[\langle w^{m_1}_1 w^{m_2}_2 \rangle = \frac{1}{2^{M}(M+1)} \left(1+\sqrt{\frac{c_R}{\pi}} + \frac{c_R}{\pi} \right) + \frac{M}{2^{M+2}}  c_R + O(c_R^{3/2}).\]

Let us try to compute these expressions by approximating the integrals. We start with the partition function
\[
\CZ_{A_3} = \int^1_0 \extd s_2 \int^{1/2}_{-1/2} \extd \tilde s_1 s_2 e^{-4 \tilde s_1^2 s_2^2/c_R} =\int^1_0 \extd s_2 \frac{1}{2} \sqrt{\pi c_R} \text{erf}\left(\frac{s_2}{\sqrt{c_R}}\right) \approx \int^1_0\extd s_2 \left(\frac{\sqrt{\pi c_R} }{2}-\frac{c_R e^{-\frac{s_2^2}{c_R}}}{2 s_2}\right)
\]
where we expanded the integrand $\frac{1}{2} \sqrt{\pi c_R} \text{erf}\left(\frac{s_2}{\sqrt{c_R}}\right)$ around $c_R=0$. The first term gives $\sqrt{\pi c_R}/2$ as expected. For the second, the integral over $s_2$ is divergent, and we therefore introduce a cutoff $\epsilon > 0$ such that
\[
    - \int^1_\epsilon \extd s_2 \frac{c_R e^{-\frac{s_2^2}{c_R}}}{2 s_2} =  \frac{1}{4} c_R \text{Ei}\left(-\frac{\epsilon ^2}{c_R}\right)-\frac{1}{4} c_R \text{Ei}\left(-\frac{1}{c_R}\right) \approx \frac{c_R}{2}\log (\epsilon )
\]
where $\text{Ei}(x) = -\int^\infty_{-x} e^{-t}/t\extd t$ and the last approximation is the $\epsilon \to 0$ behaviour of the function. All in all the expansion results in
\[
\CZ_{A_3} \approx  \frac{\sqrt{\pi c_R}}{2} + \frac{c_R}{2}\log (\epsilon).  
\]
This is $\log$ divergent as $\epsilon\to 0$, nonetheless we see that the factor $\frac{c_R}{2}$ agrees with the expansion of the exact form for $\CZ_{A_3}$, if we take into account that $\text{sgn}(\log(\epsilon))=-1$ as $\epsilon\to 0$.

Consider now $\CO = w^{M}_1 = (1/2 + \tilde s_1)^M s_2^M$. We have
\begin{align*}
    \langle w^{M}_1 \rangle 
    &\approx \frac{1}{\CZ_{A_3}} \int^1_0 \extd s_2 s^{1+M}_2 \int^{1/2}_{-1/2} \extd \tilde s_1 \left(\frac{1}{2} + \tilde s_1\right)^{M} e^{-\frac{4}{c_R} \tilde s^2_1 s^2_2}\\
    &\approx \frac{1}{\CZ_{A_3}} \int^1_0 \extd s_2 s^{1+M}_2 \int^{1/2}_{-1/2} \extd \tilde s_1 \left(\frac{1}{2^M} + \binom{M}{1} \frac{1}{2^{M-1}} \tilde s_1 + \binom{M}{2}\frac{1}{2^{M-2}} \tilde s_1^{2}\right) e^{-\frac{4}{c_R} \tilde s^2_1 s^2_2}.
\end{align*}
For the second approximation, we use the fact that $\int^{1/2}_{-1/2} x^{m}e^{-x^2/c_R}$ vanishes for $m$ odd and behaves as $c_R^{(1+m)/2}$ for $m$ even when $c_R\to 0$. As we are considering the $c_R\to 0$ limit, only the lower orders in $\tilde s_1$ will contribute to the integral. The integral over $\tilde s_1$ can then be solved exactly. We then expanded up to order $c^{3/2}_R$ before taking the integral over $s_2$. The result is
\begin{align*}
    \langle w^{M}_1 \rangle \approx \frac{1}{\CZ_{A_3}} \left(\frac{\sqrt{\pi c_R} (c_R M (M+1)+4)}{2^{M+3}(M+1)}\right)
    \approx 
    \frac{1}{2^{M}(M+1)} \left(1-\log(\epsilon)\sqrt{\frac{c_R}{\pi}} + \log(\epsilon)^2 \frac{c_R}{\pi} \right) + \frac{M}{2^{M+2}}c_R.
\end{align*}
As for $\CZ_{A_3}$, we see that the factors of the $\log$ divergences correspond to the expansion of the exact formulas, if we take into account the sign of $\log$ when $\epsilon\to 0$. Here when expanding the fraction
\[
  \frac{\alpha_1 x + \alpha_2 x^2 + \alpha_3 x^3 + \cdots}{\beta_1 x +\beta_2 x^2 + \beta_3 x^3 + \cdots}  
\]
up to $x^2$ around $x=0$, it is important to note that the coefficients $\alpha_i, \beta_i$ up to $i=3$ contribute to the second order coefficient of the expansion. Hence why we expanded the numerator of $\langle w^{M}_1 \rangle$ up to $c^{3/2}_R$. In the case of $\CZ_{A_3}$, there is no $c^{3/2}_R$ order term, hence why we did not consider it. We see that the first order of the expansion when $c_R\to 0$ shows some divergences, due to the fact the integrals and limits should not be swapped in these cases. It does not seem that there is a straight forward expansion similar to Feynman diagrams that we can perform in our case where one considers positive valued fields.

For general $n$, we can compute the 0th order of the weak coupling limit $c_R\to 0$ using Laplace's method for approximating Gaussian integrals over finite intervals. We start by approximating $S_R[s]$ up to second order around its minima. As all the summands of $S_{A_{n+1}}$ in \eqref{eq:ActionsAn+1ZnWilsonLoop} are non-negative, it is extremal when the value of $w$ is constant with $w_i=w_{j}$. In the case of $\Z_n$, extremising $S_{\Z_n}$ is equivalent to solving the Laplace equation
\[(\Delta_{\Z_n}w)(i) = 2w_i - w_{i-1} - w_{i+1} =0\]
which again has solutions $w_i=w_j$. In terms of the $s_i$ variables \eqref{eq:VariablesS}, the minima are at the values $s^0_i = \frac{i}{i+1}$ for $i=1,\dots, n-1$ and $s_n\in \R_{>0}$, which are degenerate as $s_n$ can take any value. This reflects the fact that $s_n$ is a zero mode and should be considered separately.
We therefore only expand $S_{R}$ for $A_{n+1}$ and $\Z_n$ in the first $n-1$ variables, to give
\[
S_{A_{n+1}}[s_1,\dots,s_n] \simeq \frac{1}{2} \sum^{n-1}_{i=1}\sum^{n-1}_{j=1} (s_i - s^0_i)(s_j - s^0_j) H_{ij}(s^0_1,\dots,s^0_{n-1}) s^2_n + O(s^3_i)
\]
where $H(s_1,\cdots,s_{n-1})s^2_n$ denotes the Hessian matrix of $S_{R}$ w.r.t the first $n-1$ variables, and where we have factored out $s^2_n$ for future convenience. Take now the variables  $\tilde s_i = (s_i - s^0_i)$ for $i=1,\dots,n-1$ so that the expectation of an observable $\CO$ is of the form
\[
\langle \CO \rangle = \frac{1}{\CZ} \int^1_0 \extd s_n s^{n-1}_n \int^{1-s^0_{n-1}}_{-s^0_{n-1}} \extd \tilde s_{n-1} (\tilde s_{n-1} + s^0_{n-1})^{n-2} \dots \int^{1-s^0_{1}}_{-s^0_{1}} \extd \tilde s_{1} \CO e^{-\frac{1}{2c_R}s_n^2 \<\tilde s, H \tilde s \>}.
\]
Here, Laplace's method states that in the $c_R\to \infty$ limit, the integrals over $\tilde s_i$ can be approximated by the evaluation of $(\tilde s_{n-1} + s^0_{n-1})^{n-2} \dots (\tilde s_{2} + s^0_{2}) \CO$ at the minima $\tilde s_i = 0$ times the value of the Gaussian integral $\int_{\R^{n-1}} \extd \tilde s^{n-1} e^{-\frac{1}{2}c_Rs_n^2 \<\tilde s, H \tilde s \>} $, where $H$ is the Hessian matrix as above. We have
\[
\langle \CO \rangle = \frac{1}{\CZ} \sqrt{\frac{(2\pi c_R)^{n-1}}{ \det(H)}} (s^0_{n-1})^{n-2} \cdots  (s^0_{3})^2 s^0_{2} \int^1_0 \extd s_n \CO(s^0_1,\cdots, s^0_{n-1},s_n) + O(c_R^{\frac{n+1}{2}})
\]
where we used $\det\left(\frac{1}{c_R} s^2_n H\right) = \frac{1}{c^{n-1}_R} s^{2(n-1)}_n \det (H)$. Similarly we can compute the partition function to be approximated by 
\[
\CZ = \sqrt{\frac{(2\pi c_R)^{n-1}}{ \det(H)}} (s^0_{n-1})^{n-2} \cdots  (s^0_{3})^2 s^0_{2} + O(c_R^{\frac{n+1}{2}})
\]
giving
\[
\langle \CO \rangle = \int^1_0 \extd s_n  \CO(s^0_1,\cdots, s^0_{n-1},s_n) + O(c_R).
\]
In terms of Wilson loop variables we have $w_i(s^0_1,\cdots, s^0_{n-1},s_n) = \frac{s_n}{n}$, and therefore
\[
\langle w^{m_1}_1 w^{m_2}_2 \cdots w^{m_n}_n  \rangle = \frac{1}{n^M(M+1)} + O(c_R)
\]
with $M = \sum_i m_i$. At 0th order, the value for both $A_{n+1}$ and $\Z_n$ again coincide. In the weak coupling limit the theory has the interesting feature that expectation values only depend on $M$, not $m_i$. For example, connected correlation where $M=2$ is
\[\langle w_i w_j \rangle^c =\<w_i^2\>^c= (\Delta w_i)^2=\frac{1}{12n^2}+O(c_R)\]
for any $i,j$. The standard deviation and relative uncertainties can similarly be computed and in the $\Z$ limit we have
\[
    \lim_{n\to \infty} \langle w^{m_1}_1 w^{m_2}_2 \cdots w^{m_j}_j  \rangle_{0} = 0, \quad 
    \lim_{n\to \infty} \Delta(w^{m_1}_1 w^{m_2}_2 \cdots w^{m_j}_j)=0,
    \quad
    \lim_{n\to \infty} \frac{\Delta(w^{m_1}_1 w^{m_2}_2 \cdots w^{m_j}_j)}{ \langle w^{m_1}_1 w^{m_2}_2 \cdots w^{m_j}_j  \rangle_{0}} = \frac{M}{\sqrt{2 M+1}},
\]
showing again that the `quantum' behaviour of the theory in the $n\to \infty$ case can be measured through relative uncertainties.

\subsubsection{Numerical results for $\mathbb{Z}_5$ and $A_6$}

 As well as the $c_R$ limits analysed above, we can also approach the higher $X=\Z_n, A_{n+1}$ using numerical methods to compute expectation values, standard deviations and relative uncertainties. Here, we illustrate  this for $A_6$ and $\Z_5$ using Monte-Carlo integration with C++, both in Euclidean and Lorentzian signature. For the Lorentzian case, we interpret $A_6$ as a spacelike point evolving 5 steps in time with timelike edges, and $\Z_5$ as a 5 step time loop.  Then the only changes are at the level of  the partition function, where we now take the weight $e^{\imath S}$ instead of $e^{-S}$. 

The results for the expectation values and relative uncertainties of $w_i$ are presented in Figure \ref{fig:A6Z5}. In the case of $A_6$ we only plot $w_1, w_2, w_3$ as the results for $w_4, w_5$ can be obtained by symmetry, similar in the case of $\mathbb{Z}_5$. We see clear boundary effects both with the expectation value and relative uncertainties of $w_1$ in the $A_6$ graph, while bulk values of the expectations and relative uncertainties of $w_2, w_3$ are closer to the $\Z_5$ theory. Boundary vs bulk effects have similarly been observed  in the quantum Riemannian geometry of  $A_{n+1}$ in \cite{ArgMa2}. Furthermore, we notice that quantum fluctuations become more relevant in the strong coupling regime where $\Delta(w_i)/\langle w_i \rangle$ is of order 1. The  Lorentzian theory is broadly similar to the Euclidean one but with an even more pronounced  boundary effect in $\<w_1\>$. We also see a bulk effect for $\<w_2\>$ and $\<w_3\>$ in the weak $c_R$ theory in the case of Lorentzian signature.

\begin{figure}[h!]
    \centering
    \includegraphics[width=1\textwidth]{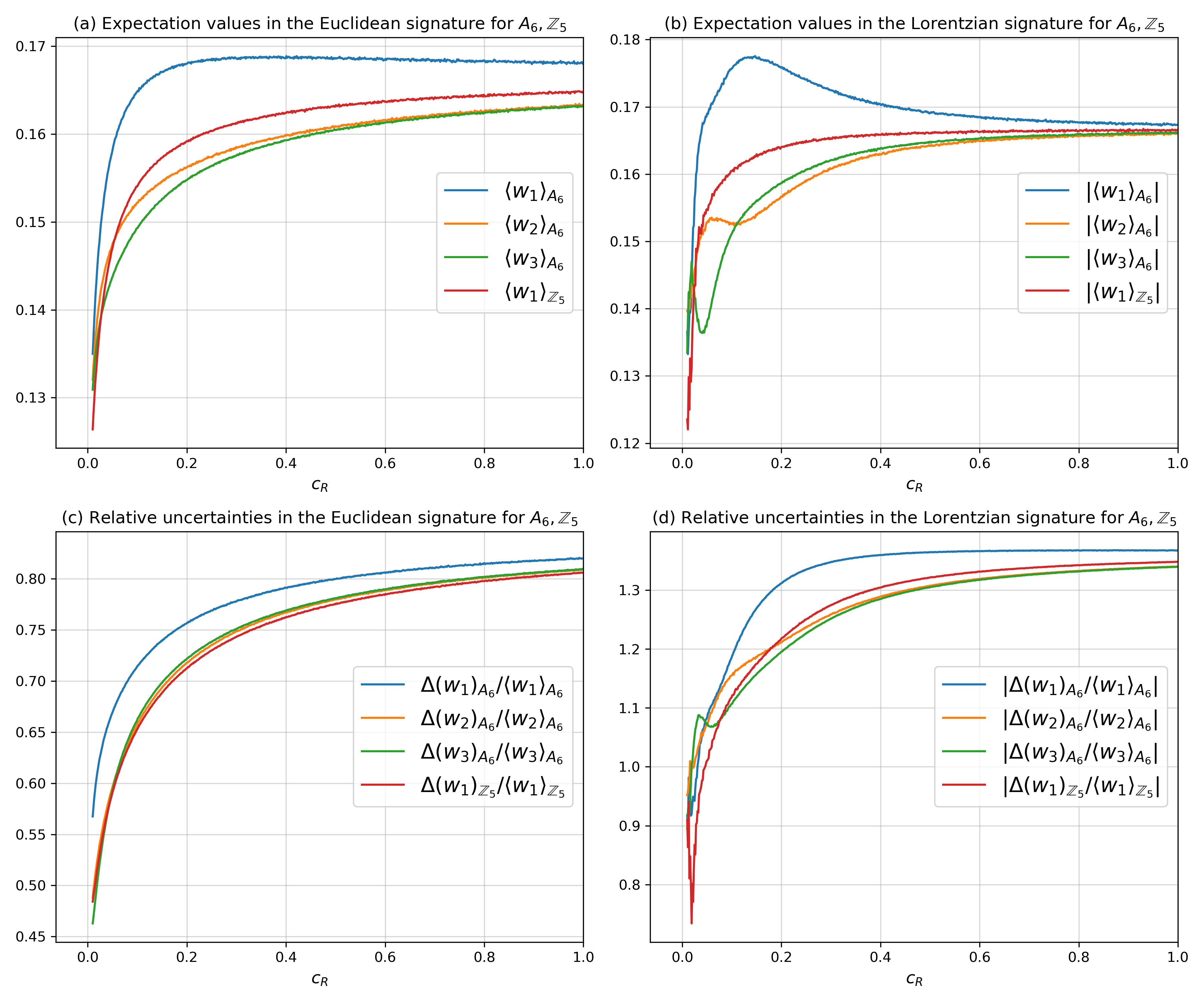}
    \caption{Parts (a) and (c) show the expectation values and relative uncertainties of the real positive scalar field $w$ for $X=A_6, \mathbb{Z}_5$ in the Euclidean signature plotted against the coupling $c_R$. Parts (b) and (d) show parallel results in the Lorentzian signature, where for simplicity we only plot  absolute values.}
    \label{fig:A6Z5}
\end{figure}

\subsubsection{Comparison with a massless real scalar field}\label{sec:scalar}

We have seen that the real modes in the QGFT on $\Z$ and $\Z_n$ look a lot like a positive $\R$-valued (or `tropical') scalar field theory. For comparison, we use the same methods for quantum field theory of a massless scalar field theory $\phi\colon \Z_n\to \R$ on $\Z_n$ in order to see what are the differences are.  We take the action $S[\phi] = \frac{1}{4 c} \int_{\Z_n} \phi \Delta_{\Z_n} \phi$ and solve it by using the Fourier transform $\tilde \phi_k = \frac{1}{\sqrt{n}}\sum^{n}_{j=1}\phi_j e^{\imath \frac{2\pi}{n} k j}$, resulting in 
\[S[\phi] = \frac{1}{2 n c} \sum^{n-1}_{k=1} \left(1-\cos\left(\frac{2\pi k}{n}\right)\right)\tilde \phi_k \tilde \phi_{-k}. \]
Due to the shift symmetry $\phi_i\to \phi_i+\alpha$ the $k=n$ mode $\tilde \phi_n = \frac{1}{\sqrt{n}}\sum^{n}_{j=1}\phi_j$ is a zero mode and does not appear in the action. As before we will need to consider it seperately.

In this case we can use standard QFT methods to compute expectation values. We consider a partition function with a source $J$, and regulate the zero mode by introducing the function $\chi_{[-L,L]}\left( \frac{1}{\sqrt{n}}\sum^{n}_{j=1}\phi_j\right)$
\begin{align*}
    \CZ[J] &= \int_{\R^n} \extd \phi_1\cdots \extd \phi_n \, \chi_{[-L,L]}\left( \frac{1}{\sqrt{n}}\sum^{n}_{j=1}\phi_j\right) e^{- \frac{1}{4 c} \sum^{n}_{j=1} \phi_j (\Delta_{\Z_n} \phi)_j + \sum^{n}_{j=1} J_j \phi_j} \\
    &= \int_\R \extd \tilde \phi_n e^{\tilde J_n \tilde \phi_n} \chi_{[-L,L]}\left( \tilde \phi_n\right) \int_{\R^{n-1}} \extd \tilde \phi_1\cdots \extd \tilde \phi_{n-1} e^{-\frac{1}{4 n c}\sum^{n-1}_{k=1} \left(1-\cos\left(\frac{2\pi k}{n}\right)\right)\tilde \phi_k \tilde \phi_{-k} + \sum^{n-1}_{k=1} \tilde J_{k} \tilde \phi_{k}} \\
    &= \CZ[0] \frac{\sinh (\tilde J_n L)}{ \tilde J_n L} e^{\frac{n c}{2} \sum^{n-1}_{k=1} \frac{\tilde J_k \tilde J_{-k}}{\left(1-\cos\left(\frac{2\pi k}{n}\right)\right)}}\\
    &= \CZ[0]\frac{ \sinh \left(\sum^{n}_{j=1}J_j L\right)}{\sum^{n}_{j=1}J_j L}
    e^{\frac{L^2}{2} \sum_{i,j} J_i G_{i-j}J_j}
\end{align*}
where $G_{i-j} =  c_R \sum^{n-1}_{k=1} \frac{e^{\frac{2\pi\imath}{n} k(i-j)}}{\left(1-\cos\left(\frac{2\pi k}{n}\right)\right)}$ is the propagator with rescaled coupling $c_R = n c / L^2$ and $\CZ[0] = 2 L \sqrt{\frac{(2\pi)^{n-1}}{\det(\tilde c G)}}$. The expectation values and correlation functions of the rescaled field can be computed to be
\[ \frac{1}{L}\langle \phi_i \rangle = 0, \quad \frac{1}{L^2}  \langle \phi_i \phi_j\rangle =\left(\frac{1}{\CZ[J]}\frac{\del^2\CZ[J]}{\del J_i\del J_j}\right)\vert_{J=0} = \frac{1}{3} + G_{i-j}.\]

We see that a massless real scalar field $\phi\colon \Z_n \to \R$ behaves differently from a massless positive scalar field $w\colon \Z_n\to \R_{>0}$. First, while the field $\phi$ have vanishing vacuum expectation value, the field $w$ does not. Secondly, while long range order appears in both theories, the behaviour of both theories in the weak and strong coupling regimes differ drastically. In the case of real scalar field theory these depend linearly on the coupling $c_R$, and therefore approach a constant value of $1/3$ in the weak regime and diverge in the strong coupling regime, except for the values of $i-j$ for which the sum $\sum^{n-1}_{k=1} \frac{e^{\frac{2\pi\imath}{n} k(i-j)}}{\left(1-\cos\left(\frac{2\pi k}{n}\right)\right)}$ vanishes. In the case of a positive scalar field, we have seen that correlations of the theory are finite in the weak and strong coupling regimes, and that they interpolate between the two. Another difference is the behaviour of the relative uncertainties in both cases. For the real scalar fields we find that the relative uncertainty for $\phi_i$ diverges, and that the relative uncertainty of $\phi_i \phi_j$ has the following behaviours in the $c_R \to 0,\infty$ limits
\[\lim_{\tilde c\to 0} \frac{\Delta(\phi_i \phi_j)}{\langle\phi_i \phi_j \rangle} = \sqrt{\frac{13}{5}}, \quad\quad\quad \lim_{\tilde c\to \infty} \frac{\Delta(\phi_i \phi_j)}{\langle\phi_i \phi_j \rangle} = \sqrt{\frac{G^2_{0}}{G^2_{i-j}}+1}.\]
Again, we see that the values are very different from the positive scalar field case.

\subsection{QGFT on $\Z_2\times \Z_2$}\label{secZ2Z2}

We cover quantisation of the $\C\Z_2$-gauge theory on a single plaquette $X=\Z_2\times \Z_2$ in Example~\ref{Z2Z2act}  as a simplified version of the square lattice $\Z^2$. 
The Euclidean case was treated in \cite{Ma:cli} but we take this further.  As we have a square, the moduli space $\CU^\times/\CG \simeq U(1)\times \R^4_{>0}$ is parametrized by a real scalar field $w$ on the edges (we write $w_1,w_2,w_3,w_4$ for the values of $w$ on the lower, right, upper and left edges respectively) and a phase $e^{\imath \theta}$ for a loop going around the square. Recall in the definition of $w_i$ in Section~\ref{sec:measureAb} that $\sqrt{w_i}$ is the real part of the transport in one direction of the edge.  We fix $\rho$ to be the unique nontrivial 1-dimensional representation of $G=\Z_2$. Then following Example \ref{Z2Z2act} the action functional simplifies to
\[
    S_{\Z_2\times \Z_2}[u] 
    = -\frac{4}{c}\left(4\sqrt{w_1 w_2 w_3w_4}\cos(\theta) - (w_1+w_3)(w_2+w_4)\right)
\]
with minima at $w_1=w_3$, $w_2=w_4$, $\theta = 0$. These minima are degenerate in two lines and correspondingly the path integral will have two scale parameters $L$, $L'$ which make it divergent. We choose the following coordinates
\[
s_1 = \frac{w_1}{w_1+w_3},\quad
s_2 = \frac{w_2}{w_2+w_4},\quad
s_3 = \frac{w_1+w_3}{L},\quad
s_4 = \frac{w_2+w_4}{L'},
\]
with inverse $w_1 = s_1 s_3 L$, $w_2 = s_2 s_4 L'$, $w_3 = (1-s_1) s_3 L$, $w_4 = (1-s_2) s_4 L'$
which allow us to write the partition function as
\[
\CZ_{\Z_2\times \Z_2} = L^2 L'^2 \int^{2\pi}_0 \extd \theta \int_{[0,1]^4} \extd s_1 \extd s_2 \extd s_3 \extd s_4 \,\,s_3\,\, s_4\,\, e^{\frac{2}{c_R} s_3s_4 (4\sqrt{(1-s_1)s_1(1-s_2)s_2}\cos(\theta)-1)} 
\]
where we introduce $c_R = \frac{c}{2 L L'}$. As before we rescale the Wilson loop variables with the respective scale $L$ or $L'$, effectively setting $L=L'=1$, and do not change the notation. So that for example $w_1$ now denoted the rescaled Wilson loop $00\to 10\to 00$ and so on. 

We will focus on 4 types of observables: the Wilson loop $w_1$, the Wilson loop around the plaquette $\sqrt{w_1 w_2 w_3 w_4} \, e^{\imath \theta}$, its pure positive scalar field $\sqrt{w_1 w_2 w_3 w_4}$ and pure phase $e^{\imath \theta}$ parts. Note that
\[
\langle \sin(\theta) \rangle = \frac{1}{\CZ_{\Z_2\times \Z_2}} \int^{2\pi}_0 \extd \theta \int_{[0,1]^4} \extd s_1 \extd s_2 \extd s_3 \extd s_4 \,\,s_3\,\, s_4\,\, \sin(\theta) e^{\frac{2}{c_R} s_3s_4 (4\sqrt{(1-s_1)s_1(1-s_2)s_2}\cos(\theta)-1)} = 0
\]
due to $\sin(\theta)$ being antisymmetric and $\cos(\theta)$ symmetric wrt $\theta$, making the expectation values of $e^{\imath \theta}$ and similarly $\sqrt{w_1 w_2 w_3 w_4} \, e^{\imath \theta}$ real valued.
Following the previous sections, we can compute the $c_R \to \infty$ limit of expectation values as
\[
\lim_{c_R\to \infty} \langle \CO \rangle = \frac{\int^{2\pi}_0 \extd \theta \int_{[0,1]^4} \extd s_1 \extd s_2 \extd s_3 \extd s_4 \,\,s_3\,\, s_4\,\, \CO}{\int^{2\pi}_0 \extd \theta \int_{[0,1]^4} \extd s_1 \extd s_2 \extd s_3 \extd s_4 \,\,s_3\,\, s_4} = \frac{2}{ \pi}\int^{2\pi}_0 \extd \theta \int_{[0,1]^4} \extd s_1 \extd s_2 \extd s_3 \extd s_4 \,\,s_3\,\, s_4\,\, \CO.
\]
For the weak coupling limit, we again use Laplace's method with respect to $s_1$, $s_2$ and $\theta$, followed by integration wrt to $s_3$, $s_4$. The minima of $S_{\Z_2\times \Z_2,R} = S_R/c_R$ are at $s_1=s_2=\frac{1}{2}$, $\theta = 0$, and the Hessian of $S_R$ wrt these coordinates at the minimum is
\[
    s_3 s_4 H_{s_1,s_2,\theta}(S_{R})\vert_{s_1=1/2,s_2=1/2,\theta=0} = 2 s_3 s_4
\begin{pmatrix}
    4 & 0 & 0\\
    0 & 4 & 0\\
    0 & 0 & 1
    \end{pmatrix} 
\]
with determinant $\det(s_3 s_4 H_{s_1,s_2,\theta}(S_{R})\vert_{s_1=1/2,s_2=1/2,\theta=0}) = 128 s^3_3s^3_4$. Using Laplace's method as in the previous sections we then have 
\[
\lim_{c_R\to 0} \langle \CO \rangle = \frac{1}{2}\int_{[0,1]^2} \extd s_3 \extd s_4 \CO\left(s_1=\frac{1}{2},s_2=\frac{1}{2},s_3,s_4,\theta=0\right) \frac{1}{\sqrt{s_3 s_4}}.
\]

\begin{table}[h!]
    \centering
    \begin{tabular}{|c|c|c|c|}
        \hline  
        $\langle \CO \rangle$ & $c_R\to 0$ & $c_R\to \infty$  \\
          \hline
          $\langle w_1\rangle$ & $1/6$ & $1/3$ \\
          \hline
          $\langle\sqrt{w_1 w_2 w_3 w_4}\rangle$ & $1/36$ & $\pi^2/144$ \\
          \hline
          $\langle\sqrt{w_1 w_2 w_3 w_4} \, e^{\imath \theta}\rangle$ & $1/36$ & $ 0$ \\
          \hline
          $\langle e^{\imath \theta}\rangle$ & $1$ & $0$ \\
          \hline
    \end{tabular}
   \caption{Expectation value of several observables for $X=\Z_2\times\Z_2$ in weak and strong coupling limits.}
    \label{tab:limitsZ2Z2}
\end{table}

The results are summarised in Table~\ref{tab:limitsZ2Z2} for the mentioned observables. We see that in the weak coupling limit, the pure phase $e^{\imath \theta}$ is dominant over the expectation of observables including the positive scalar field $w_i$. In contrast, in the strong coupling limit the expectations values of the pure phase die down, while the expectations of the positive scalar field remain non-zero. 

Of further interest is the comparison of standard LGT on the Euclidean square with our results. We restrict the Yang-Mills action to the case $w_i=1$ to recover $U(1)$-LGT with $S_{LGT}[\theta] = -\frac{8}{c_R}\cos(\theta)$, where we have dropped the constant terms. We then further restrict this to $\Z_2$-LGT by allowing only $\theta = 0,\pi$, similarly with $S_{LGT}[\sigma] = -\frac{8}{c_R}\sigma$, $\sigma = \pm 1$. One can then compute the expectation value and relative uncertainties of the Wilson loop to be
\[
\langle W \rangle_{U(1)} = \langle e^{\imath \theta} \rangle_{U(1)} = \frac{I_1\left(\frac{8}{c_R}\right)}{I_0\left(\frac{8}{c_R}\right)},
\quad \quad
\langle W \rangle_{\Z_2} = \langle \sigma \rangle_{\Z_2} = \tanh\left(\frac{8}{c_R}\right), 
\]
where $I_0$, $I_1$ are the modified Bessel functions of first kind. These have the same $c_R\to 0,\infty$ limits as the expectation $\langle e^{\imath \theta}\rangle$ of the pure phase in the noncommutative-$U(1)$-gauge theory, but their behaviour between the limits differs as can be seen Figure~\ref{fig3}, where we have also plotted the standard deviations and relative uncertainties. When we have the possibility of complex valued observables, we define the standard deviation through
\begin{equation*} \Delta \CO = \sqrt{\langle(\CO - \langle \CO \rangle )^*(\CO - \langle \CO \rangle) \rangle} = \sqrt{\langle |\CO|^2 \rangle  - \langle \CO^* \rangle \langle \CO \rangle}. \end{equation*}
In the Euclidean case we have $\langle \CO^* \rangle = \langle \CO \rangle^*$ which further simplifies the above to a real positive quantity $\Delta \CO = \sqrt{\langle |\CO|^2 \rangle  - |\langle \CO \rangle|^2}$.

A few things are noticeable from the numerical results in Figure~\ref{fig3}.  First, the expectation of Wilson loop $W= \sqrt{w_1 w_2 w_3 w_4} e^{\imath \theta}$ around the plaquette coming from NCG geometry  is much smaller than for the Wilson loop in the other theories. From the perspective of NCG, this should be the correct Wilson loop relevant to the magnetic flux through the plaquette. Even if one wishes to just consider the expectation value of the phase $e^{\imath \theta}$, nonetheless the presence of the $w$ field will lower the expectation value in comparison with the expectation computed from standard LGT, especially in the weak coupling regime. Furthermore, the relative uncertainties show that the quantum fluctuations compared to $\langle W \rangle$ are much stronger in the strong coupling regime for NCG in comparison to the $U(1)$ and $\Z_2$-LGT theories.  In the weak coupling regime, while the relative uncertainties of the Wilson loop in the two restricted theories approach zero, the limit of the relative uncertainty of the Wilson loop in the full theory
\[
\lim_{c_R\to 0} \frac{\Delta(\sqrt{w_1 w_2 w_3 w_4}e^{\imath \theta})}{\langle \sqrt{w_1 w_2 w_3 w_4}e^{\imath \theta}\rangle} = \frac{\sqrt{137}}{5} \sim 2.34
\]
shows that quantum fluctuations are still relevant in this regime.

\begin{figure}[h!]
    \centering
    \includegraphics[width=1\textwidth]{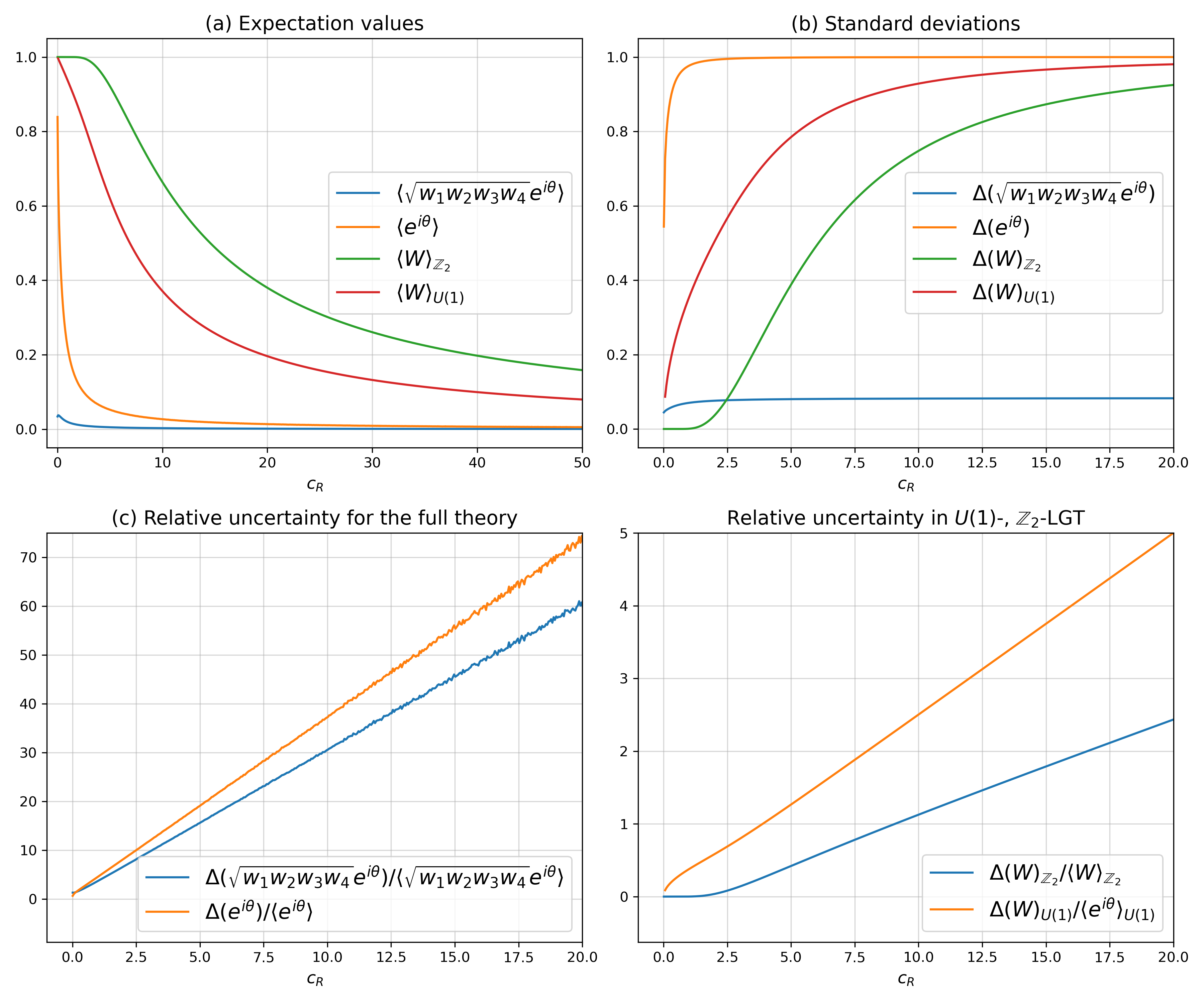}
    \caption{(a) Numerical plots for the expectation values of the Wilson loop and phase $e^{\imath \theta}$ for NCG and $U(1)$, $\Z_2$-LGT (they come out as real) on Euclidean $\Z_2\times\Z_2$. (b) Shows the standard deviations, with the relative uncertainties presented in (c) and (d).}
    \label{fig3}
\end{figure}

\subsubsection{Case of Lorentzian square $\Z_2\times \Z_2$}
\label{sec:LorZ2Z2}

We now turn our attention to the Lorentzian case. For a square lattice such as $\Z^M$ or $\Z_N^M$, namely (if we think of time as the vertical axis) we let all vertical arrows have $\lambda_{x\to y}>0$ and other (spacelike) arrows have $\lambda_{x\to y}<0$. This is because these edge weights play the role of `square-lengths'. Quantum gravity on a Lorentzian square was computed on this basis in \cite{Ma:squ,BliMa} and here we treat $\C\Z_2$-gauge theory in the same setting. We stick to the `flat case' so $\lambda_{(i,j)\to (i+1,j)}=-1$ and $\lambda_{(i,j)\to (i,j+1)}=1$ for the metric. Then returning to (\ref{SYMab}), we see that the factor comes
in pairs $\lambda_{x\to x a}\lambda_{x\to xb}$ where $a\ne b$. Here $a,b\in \{(1,0),(0,1)\}$ so there will always be one horizontal and one vertical so that these metric factors together always give -1. It follows that the action $S_{\Z_2^{1,1}}[u]$ is just the negative of what we had in the Euclidean square.  The resulting theory 
\[\CZ_{\Z_2^{1,1}} = L^2 L'^2 \int^{2\pi}_0 \extd \theta \int_{[0,1]^4} \extd s_1 \extd s_2 \extd s_3 \extd s_4 \,\,s_3\,\, s_4\,\, e^{-\frac{2 \imath }{c_R} s_3s_4 (4\sqrt{(1-s_1)s_1(1-s_2)s_2}\cos(\theta)-1)} 
\]
nevertheless has a very different character from that in the preceding section due to the additional $\imath$ in the exponent that we should use in a Lorentzian theory. 

\begin{figure}[h!]
    \centering
    \includegraphics[width=1\textwidth]{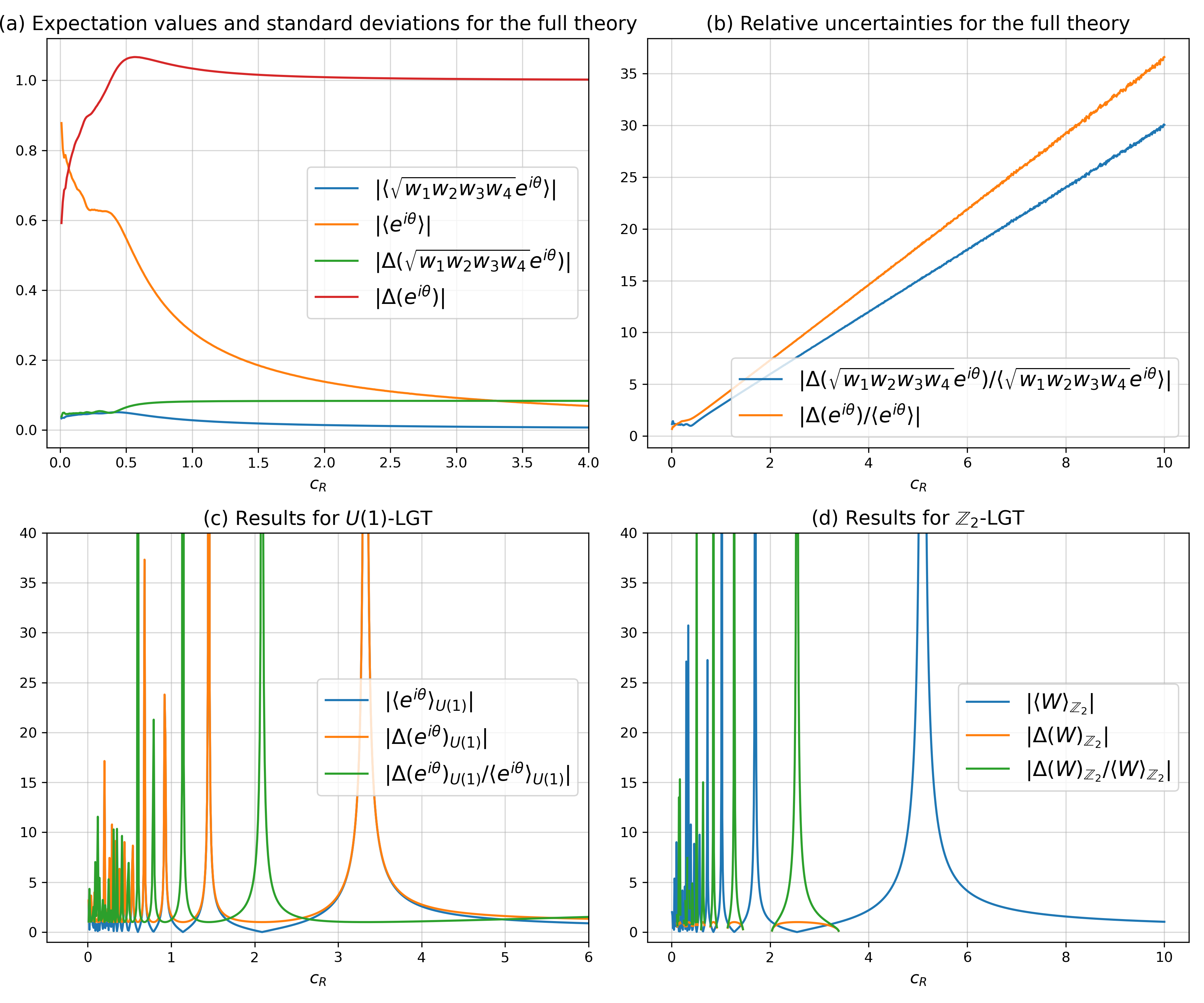}
    \caption{(a) Numerical plots for the absolute values of the expectation values and quantum fluctuations of the Wilson loop and phase $e^{\imath \theta}$ for NCG on Lorentzian $\Z_2\times\Z_2$. (b) Shows the absolute values of the relative uncertainties. Parts (c) and (d) show the much wilder results for $U(1)$ and $\Z_2$-LGT respectively.}
    \label{fig:Z2Z2Lorexpval}
\end{figure}

The limiting values of the observables in Table~\ref{tab:limitsZ2Z2} are the same as in the Euclidean theory, where the $c_R \to 0$ limit can be computed via the stationary phase formula
\[
\lim_{c_R\to 0} \int_{\R^n} g(x) e^{\frac{\imath}{c_R} f(x)} =\sum_{x_0 \in \Sigma} \frac{e^{\frac{\imath}{c_R}f(x_0)}}{\sqrt{|H(f(x_0))|}} e^{\frac{\imath \pi}{4}\mathrm{sgn}(H(f(x_0)))} \left(2\pi c_R \right)^{n/2} g(x_0)
\]
for $\Sigma$ the set of critical points of $f(x)$. The non-limit behaviour of the expectation values is plotted in Figure \ref{fig:Z2Z2Lorexpval}, together with their standard deviations $\Delta(\CO) = \sqrt{\langle |\CO|^2 \rangle - \langle \CO^* \rangle\langle \CO \rangle}$ and relative uncertainties. In the Lorentzian case these can be complex and we therefore show the absolute values. We note a stark difference between the full NCG theory and the restricted $U(1)$ and $\Z_2$ theories, where
\[
\langle W \rangle_{U(1)} = \langle e^{\imath \theta} \rangle_{U(1)} = \imath \frac{J_1\left(\frac{8}{c_R}\right)}{J_0\left(\frac{8}{c_R}\right)}
\quad \quad
\langle W \rangle_{\Z_2} = \langle\sigma \rangle_{\Z_2} = -\tan\left(\frac{8}{c_R}\right)
\]
with $J_n(z)$ are the Bessel functions of first kind. While the full theory behaves in a similar way as the Euclidean case, both $U(1)$-LGT and $\Z_2$-LGT show repeated divergences on the considered quantities. It appears  that the presence of the real positive scalar field $w$ highly regulates the behaviour of these theories.

\section{Non-universal calculi on the gauge group}
\label{secnonu}

So far, we  have considered gauge transformations $\gamma$ to be $\C G$ valued as a kind of `quantisation' (in the sense of quantum computing) of being valued in $G$. The latter is also possible and leads to much larger moduli spaces, as the notion of gauge equivalence is then much weaker.  In algebraic terms the less restrictive notion corresponds to  invertible elements of $\C(X)\tens\C G$ or equivalently to  `convolution-invertible' linear maps $\C(G)\to \C(X)$. The latter makes sense in noncommutative geometry as convolution-invertible linear maps $A\to B$ for a possibly noncommutative coordinate algebra $B$ in place of $\C(X)$ and a quantum group coordinate algebra $A$ in place of $\C(G)$. The more restrictive $G$-valued case corresponds algebraically to algebra homomorphisms $A\to B$. Both versions make sense and can be used, but the original theory in \cite{BrzMa} used the convolution-invertible definition as it is the maximal choice (and hence more things are gauge equivalent) whereas typically between unrelated algebras $A,B$ there may be very few algebra homs, i.e. this will typically be too restrictive. 

On the other hand, the more restrictive algebra hom (or in our setting, $G$-valued) version is easier to work with and was used, for example in \cite{HanLan}. Particularly, it is clearer in this case what means a differentiable gauge transform: when $A$ and $B$ are equipped with differential 1-forms $\Omega^1_A, \Omega^1$, an algebra map $A\to B$ is differentiable of it extends to $\Omega^1_A\to \Omega^1$ compatibly with $\extd$ and the bimodule structures on each side, see \cite[Sec.~1.1]{BegMa}. In the case relevant to us, a differentiable map $\gamma: X\to G$, where both $X,G$ are graphs, is a set map  where\cite[Theorem 3.1]{Ma:gra} for all arrows $x\to y$, either $\gamma_x = \gamma_y$ or $\gamma_x \to \gamma_y$ is an arrow. Such maps pull back to maps between the corresponding functions and differential forms by
\[ \gamma^*: \C(G)\to \C(X),\quad \gamma^*(f)(x)=f(\gamma_x),\quad \gamma^*:\Omega^1_G\to \Omega^1_X,\quad \gamma^*(e_{g\to h})=\sum_{x\to y} \delta_{g, \gamma_x} \delta_{h,\gamma_y} e_{x\to y}.\]
So far this has  not been a consideration since the differential calculus on $G$ as structure group was so far relevant only to the finite group `Lie algebra'  in which a connection $\alpha$ has its values, and we have so far just taken this to be $\C G^+$. The latter corresponds to the universal differential calculus on $G$, which has all possible arrows between distinct elements so that all $G$-valued $\gamma$ are automatically differentiable. In the case of $\C G$-valued $\gamma$ there is currently no concept of differentiable convolution-invertible maps, so either way we have not had  to require conditions for $\gamma$ to be differentiable. 

By contrast, we now want to look at $\alpha$ being valued on the finite group `Lie algebra' of a general (nonuniversal) calculus on $G$, and we also need to restrict to $G$-valued $\gamma$ to be able to impose its differentiability.  Recall from Section~\ref{secdiffG} that nonuniversal calculi on $G$ are induced by a subset $\CC \subsetneq G\cup\{e\}$ by right multiplication $g\to g a$, $a\in \CC$ (these are the arrows of the corresponding Cayley graph on $G$). We take $\CC$ to be $\Ad$-invariant and stable under inverses. As for a Lie group, the space of 1-forms is a free module $\Omega^1_G \simeq \C(G) \tens \Lambda^1$ generated by  the vector space of left invariant 1-forms $\Lambda^1$. It's dual $\Lambda^{1*} \simeq\C \CC$ corresponds to the `Lie algebra', where we recall that $\Lambda^{1*}$ has a natural basis $\{a-e\ |\ a\in \CC\}$ and hence can be identified with the span of $\CC$ if we are being sloppy,  but strictly speaking its elements are displaced by $-e$. This is necessary for them to live in $\C G^+$ (so the finite group `Lie algebra' for a nonuniversal calculus is a subspace of this, corresponding to the space of 1-forms being a quotient). We thus define connections as 1-forms with values in $\Lambda^{1*}$ obeying the same reality condition as before,
\[
\CA = \{\alpha \in \Lambda^{1*} \tens \Omega^1_X\ \vert\ \alpha^* = -\alpha\}.
\]
We write $\alpha \in \CA$ as $\alpha = \sum_{x\to y} \alpha_{x\to y} \tens e_{x\to y}$, then the conditions are $ \alpha_{x\to y}\in \Lambda^{1*}$ and $\alpha_{x\to y}^*=\alpha_{y\to x}$. In terms of the extended holonomies, this is  
\[ u_{x\to y}:=e+\alpha_{x\to y}\in \Lambda^{1*}\oplus\C e,\quad u_{x\to y}^*=u_{y\to x}\]
in line with what we did before to the nonuniversal case. Here $\Lambda^{1*}\oplus\C e\subseteq\C G$ still consists of elements of counit 1, so the allowed holonomies are just a subset of those for the universal calculus.

The next issue  is how to define gauge transformations so that $\alpha^\gamma$ is still in $\CA$. Recall that
\[
\alpha^\gamma = \sum_{x\to y} (\gamma_x \alpha_{x\to y}\gamma^*_y + \gamma_x \gamma_y^* - e) \tens e_{x\to y}.
\]
For this to be a connection, the expression in the parenthesis needs to be $\Lambda^{1*}$-valued. But if we want this condition to hold uniformly for all $\alpha_{x\to y}$ including $\alpha_{x\to y}=0$, then we need the terms to be individually in $\Lambda^{1*}$ for every arrow $x\to y$, 
\begin{align}
\label{eq:conditionGammaAlphaConnection}
\gamma_x \alpha_{x\to y}\gamma^*_y \in \Lambda^{1*},\quad  \gamma_x \gamma_y^* - e \in \Lambda^{1*}
\end{align}
and that (which we still assume) $\gamma_x^*=\gamma_x^{-1}$. Recalling the above basis of $\Lambda^{1*}$, the second part of the condition is equivalent to $\gamma_x \gamma^*_y \in \CC \cup\{e\}$.  We also have, since $\CC$ is assumed to contain inverses and be $\Ad$-invariant, that if $\gamma^*_x\gamma_y \in \CC \cup \{e\}$ then also $\gamma_y (\gamma^*_x\gamma_y)^* \gamma^*_y = \gamma_x \gamma^*_y \in \CC \cup \{e\}$, and conversely. Hence the second part of (\ref{eq:conditionGammaAlphaConnection}) can equally be written as  $\gamma_x^* \gamma_y - e\in \Lambda^{1*}$ or equivalently as $\gamma_x^*\gamma_y\in \CC\cup\{e\}$. We therefore take the space of local gauge transformations to be 
\[
\CG \coloneqq \{\gamma \in \C G \tens \C(X)\ \vert\ \gamma_x\in G\ \forall x\in X,\quad  \gamma^*_x\gamma_y \in \CC \cup \{e\}\ \forall x\to y\}.
\]
On the other hand, since  the calculus on $G$ is induced by $\CC$ via right multiplication, differentiability of a $G$-valued $\gamma$ is equivalent to saying that for every arrow $x\to y$ we have $\gamma_y = \gamma_x a$ for $a\in \CC \cup \{e\}$ or equivalently that $\gamma^*_x\gamma_y \in \CC \cup \{e\}$. Thus, $\CG$ is nothing other than differentiable $G$-valued maps $X\to G$. The original form of the 2nd part of (\ref{eq:conditionGammaAlphaConnection})  similarly amounts to $\gamma^*=\gamma^{-1}$ being differentiable.

The problem here, however, is that $\CG$ is not necessarily closed under pointwise multiplication. If for two local gauge transformations $\gamma$, $\eta$ and we have for $x\to y$ that  $\gamma^*_x \gamma_y, \eta^*_x\eta_y  \in \CC \cup \{e\}$, it is not clear that $\gamma \eta$ likewise obeys this since $(\gamma\eta)^*_x (\gamma \eta)_y = \eta^*_x \gamma^*_x \gamma_y \eta_y$ need not live in $\CC\cup\{e\}$. If on every connected component of $X$ either $\gamma$ or $\eta$ are global (i.e. constant) then there is no problem and $\gamma\eta\in \CG$ due to either $\gamma_x^*\gamma_y=e$ or $\eta_x^*\eta_y=e$ and Ad-invariance of $\CC$ (but the composite may not be global unless both $\gamma,\eta$ are). Restricting to global gauge transformations is an option but does not allow for more interesting local gauge transformations that could vary within some  connected components of the graph. Our proposal for the latter is as follows:

\begin{proposition} 
A sufficient condition for $\CG$ to be closed under multiplication is that \[\CC\cdot \CC \subseteq \CC \cup \{e\}.\] 
In this case,  $\alpha^\gamma\in \CA$ is a connection for any $\alpha\in \CA$. However, the only connected calculus on $G$ obeying this condition is the universal calculus where $\CC=G\setminus\{e\}$. 
\end{proposition}

\begin{proof} Note that the condition on $\CC$ stated also implies that $(\CC\cup \{e\})^2, \CC\cdot(\CC\cup\{e\})\subseteq \CC\cup\{e\}$. 
  (1)   Let $\gamma$, $\eta$ be two gauge transformations, so for every arrow $x\to y$ we have $\gamma^*_x \gamma_y, \eta^*_x \eta_y \in \CC \cup\{e\}$. But $y\to x$ is also an arrow, so we also have $\eta_y^*\eta_x\in\CC\cup\{e\}$. As explained above, this also equivalent under our assumptions  to $\eta_y\eta_x^*\in \CC\cup\{e\}$. Then $\eta_x(\eta_x^*\gamma_x^*\gamma_y\eta_y)\eta_x^*=\gamma_x^*\gamma_y\eta_y\eta_x^*\in\CC\cup\{e\}$ using the condition on $\CC$. Hence by Ad-invariance of $\CC$, we also have $\eta_x^*\gamma_x^*\gamma_y\eta_y\in \CC\cup\{e\}$ so that $\gamma\eta\in \CG$.  

(2)  If $x\to y$ then so does $y\to x$ and hence $\gamma_y^*\gamma_x\in \CC\cup\{e\}$. Then $\gamma_x^*(\gamma_x \alpha_{x\to y}\gamma_y^*)\gamma_x=\alpha_{x\to y}\gamma_y^*\gamma_x$ and we show that this lives in $\Lambda^{1*}$.  Writing $\gamma_y^*\gamma_x=b\in \CC\cup\{e\}$ and $\alpha_{x\to y}=\sum_{a\in\CC}\alpha_{x\to y}^a(a-e)$, we have for each term that $(a-e)b=ab-b$ is either $c-b=(c-e)-(b-e)$ for $c\in \CC$ or $e-b$ (by the condition on $\CC$). Either way, the result for each term lives in $\Lambda^{1*}$ as required. Hence, by Ad-invariance of $\Lambda^{1*}$,  $\gamma_x \alpha_{x\to y}\gamma_y^*\in \Lambda^{1*}$ as required. 

(3) If $\CC$ generates $G$ via right multiplication then for all $g\in G$, there is a set $a_i\in \CC$, $i=1,\dots,n$ for some $n\in \mathbb{N}$, such that $g = a_1 \cdot \ldots \cdot a_n$. Repeatedly using the condition on $\CC$ then implies that $a_1, a_{1}a_{2}, ...,  g\in \CC \cup \{e\}$.  Applying this to all $g\ne e$, we see that  $G \setminus \{e\}\subseteq \CC$. But $\CC\subseteq G\setminus\{e\}$ as part of the specification of a calculus, so we have equality.  \end{proof}

Thus, the condition stated does the job and we have a local gauge theory but the calculus on $G$ for this to work cannot be a connected one. We will see this in our examples. 

\begin{example} For $G = \Z$ (note that this is with its additive group structure), we can fix a positive integer $k\in \N$ and take the calculus  $\CC = k \Z\setminus \{0\}$, i.e. all nonzero multiples of $k$. Then $nk,mk\in \CC$ imply that $nk+mk=(n+m)k\in \CC\subset\CC\cup\{0\}$ as required.  For this calculus,  the integer line is split into $k$ connected components labelled by $q = 0,\dots,k-1$, with vertices $\{q + n k\vert n\in \Z\}$, i.e. the set of connected components is just $\Z_k$ defined as integers mod $k$. Each component is the complete graph, i.e. has the universal calculus, since $(q+nk)-(q+mk)=(n-m)k\in \CC$ so that there is an arrow $(q+mk)\to (q+nk)$. The zeroth de Rham cohomology class (i.e. the functions with zero $\extd$) can be computed as 
\[ H^0(\Z,\CC) \simeq \C(\Z_k)=\C^k,\]
 i.e. functions that are $k$ periodic, meaning they are constant on the connected components. In this specific example every connected component is bijective to the integer line by $q+nk$ corresponding to $k$, so we have $k$ copies of $\Z$ each with the universal calculus. 

For $G=\Z_N$, we take $k \in \Z_N$ and set $\CC = k \Z_N \setminus \{0\}$. If $k$ is not a divisor of $N$ then $\CC = k \Z_N \setminus \{0\}$ is the universal calculus. Otherwise, we can list its elements as $\CC=\{jk\ |\ j=1,\cdots, N/k-1\}$ which is bijective to $\Z_{N/k}\setminus\{0\}$. If $nk, mk\in \CC$ then in $\Z_N$ (i.e. working modulo $N$), we have $nk+mk=[n+m]k$ where $[\ ]$ denotes modulo $N/k$, so $nk+mk\in \CC\cup\{e\}$ and hence our condition again holds. By similar arguments as before, the calculus on $\Z_N$ has $k$ connected components, each also bijective to $\Z_{N/k}$ with its universal graph calculus, and $H^0(\Z_N, \CC) \simeq \C(\Z_{k})=\C^k$.
\end{example}

\begin{example}\label{exS3non} For a nonAbelian example, we take $G = S_3$. There are only two nontrivial conjugacy classes, $\{u,v,w\}$ and $\{uv,vu\}$, leading to two choices of calculus. The connected one generated by $\CC = \{u,v,w\}$ cannot obey the condition in the proposition and forces us to restrict gauge transformations to be constant, i.e. global. The other choice $\CC = \{uv,vu\}$ obeys $\CC\cdot \CC \subseteq \CC \cup \{e\}$ since $uv \cdot \CC = \{vu,e\}$, $vu \cdot \CC = \{e,uv\}$, so we have a local gauge theory, but with a non-connected calculus. Indeed, the graph is
\[
\begin{tikzcd}
{} & {u \arrow[dd,leftrightarrow] \arrow[rrd,leftrightarrow]} & {uv \arrow[dd,leftrightarrow] } & {} \\
{e \arrow[urr,leftrightarrow]\arrow[drr,leftrightarrow] } & {} & {} & {w\arrow[lld,leftrightarrow] } \\
{} & {v } & {vu  } & {} \\
\end{tikzcd}.
\]
with two connected components  $S_3 = \{e,uv,vu\} \cup \{u,v,w\}$, each with their universal (triangle) calculus. The first component here is the subgroup $\Z_3$ and  
$H^0(S_3, \CC) = \C(S_3/\Z_3)=\C^2$.
\end{example}

We now compute the moduli space $\CU^\times/\CG$ for the case of $G=\Z_N$ above and the $A_2$ graph $\bullet-\bullet$. 

\begin{example}
    Take the base graph $X$ to be the two vertex graph $x\leftrightarrow y$ and the gauge group $G = \Z_N$ with calculus $\CC =\{jk\ |\ j=1,\cdots,N/k-1\}$  as above, where $k$ is a fixed divisor of $N$. There is only one edge so $u_{x\to y}=u$ and $u_{y\to x}=u^*$ for some element $u\in \Lambda^{1*}\oplus\C e$. To avoid confusion we retain  $e$ to refer to 0 as an element of the group $\Z_N$. We also assume $u$ is invertible as an element of the group algebra $\C\Z_N$, i.e. we look at $\CU^\times$. We can expand and take adjoints in the group algebra $\C \Z_N$ as 
 \[ u=e+\sum_{j=1}^{N/k-1}u^j(jk-e),\quad  u^*= e+\sum_{j=1}^{N/k-1}\overline{u^j}(N-jk-e)= e+\sum_{j=1}^{N/k-1}\overline{u^j}\left(\left({N\over k}-j\right)k-e\right)\]
 for some coefficients $\{u^j\}$, which we can choose to regard as defining a corresponding element $\tilde u=\sum_{j=0}^{N/k-1}u^j j\in \C \Z_{N/k}$, where we choose to set $u^0:=1-\sum_{j=1}^{N/k-1}u^j$. Then $u^*$ corresponds to $(\tilde u)^*$ in this group algebra. To impose invertibility in $\C\Z_N$, we write 
  \[ u=e\left(1-\sum_{j=1}^{N/k-1}u^j\right)+\sum_{j=1}^{N/k-1} u^j (jk)=\sum_{j=0}^{N/k-1}u^j(jk).\]
We think of the coefficients here as a function $U$ on $\Z_N$ with support on elements of the form $jk$, where $j=0,\cdots, N/k-1$, and we then Fourier transform this function on $\Z_N$ to give
\[ \hat U(m)= \sum_{i=0}^{N-1}e^{{2\pi\imath\over N}im}U^i=\sum_{j=0}^{N/k-1}q^{jm}u^j=\widehat {\tilde U}([m]),\]
 where $q=e^{2\pi\imath\over N/k}$ and $\tilde U$ denotes the coefficients of $\tilde u$ regarded as a function on $\Z_{N/k}$, which we then Fourier transformed on this group. Here, $\hat U$ is a function on $\Z_N$ which is periodic with period $N/k$ and one period is the $\Z_{N/k}$ Fourier transform of $\tilde U$. But, invertibility of $u$ in the group algebra corresponds to point-wise invertibility of the Fourier transform of $U$ and hence to pointwise invertibility of the Fourier transform of $\tilde U$ and hence to invertibility of $\tilde u$ in the group algebra $\C \Z_{N/k}$. In other words, $\CU^\times$ is the same as for $\C\Z_{N/k}$-gauge theory with its universal calculus, where we would associate $\tilde u\in \C\Z_{N/k}$  to $x\to y$ and its $*$ to $y\to x$, and would want $\epsilon(\tilde u)=1$ (which was ensured by our choice of $u^0$) and $\tilde u$ invertible. 

Next, a gauge transformation  means  pair $\gamma=(\gamma_x,\gamma_y)$ with $\gamma_x,\gamma_y\in \Z_N$ and (say) $\gamma_x\gamma_y^*=j_\gamma k\in \CC\cup\{e\}$. The gauge transformed extended holonomy is $\gamma_x u\gamma_y^*= (j_\gamma k).u= \sum_{j=0}^{N/k-1} u^j((j+j_\gamma)k)$. The corresponding elements of $\C \Z_{N/k}$ is $\widetilde{u^\gamma}=\sum_{j=0}^{N/k-1}u^j [j+j_\gamma]$ where we work modulo $N/k$. In other words, the group of gauge transformations $\CG$ is bigger but only the associated $j_\gamma$, which we can view as in $\Z_{N/k}$ enters the gauge transformation. Moreover, in terms of the corresponding element $\tilde u$, this transformation is the same as for $\C\Z_{N/k}$-gauge theory with universal calculus {\em but with gauge transformations restricted to group-valued}.  

To finish, working entirely in this $\C\Z_{N/k}$-gauge theory with universal calculus but $\Z_{N/k}$-valued gauge transforms, the action of $j_\gamma\in \Z_{N/k}$ is to multiply the value of $\tilde u$ in $\C\Z_{N/k}$. The corresponding function $\tilde U$ gets translated by $j_\gamma$, which means that its Fourier transform $\widehat{\tilde U}(m)$ gets a phase factor $q^{m j_\gamma}$. The value of $\widehat{\tilde U}(0)=1$ is unchanged but the value at $m=1$ can be changed by any power of $q$ which we can use to bring it to the wedge
\[ \widehat{\tilde U}(1)=r e^{\imath \varphi},\quad r\in(0,\infty),\quad \varphi\in \left[0, {2\pi \imath k\over N}\right).\]
Having used up our discrete gauge freedom, the other values $\widehat{\tilde U}(m)$ for $m\ge 2$ can be chosen freely other than they need to be invertible complex numbers. Hence 
\[ \CU^\times/\CG\simeq \left\{r e^{\imath \varphi}\ \left|\ r\in(0,\infty),\  \varphi\in \left[0, {2\pi \imath k\over N}\right)\right.\right\} \times (\C^{\times})^{N/k-2}.\]
This is in stark contrast to $\CU^\times/\CG=\R_{>0}^{N/k-1}$ if we had allowed $\C\Z_{N/k}$-valued gauge transformations as in Section~\ref{secsimply}. 

For a concrete choice, $N=6$ and $k=2$ gives $\CC=\{2,4\}$. The 6 vertices of $\Z_6$ with this calculus form a graph with two connected components, each a triangle (this is coincidentally the same graph as in Example ~\ref{exS3non}). The equivalent theory is a $\C\Z_3$-gauge theory with universal calculus but  $\Z_3$-valued gauge transforms. The moduli space for $X=\bullet-\bullet$ is then 2-dimensional, namely a $120^\circ$ wedge region as above $\times \C^\times$, in contrast to $\R_{>0}^2$ for $\C \Z_3$-valued gauge transforms as in previous sections. 
\end{example}

\section{Concluding remarks}\label{secrem}

We have presented a new and systematic approach to gauge theory on a graph $X$ with gauge `group' replaced by the group algebra $\C G$ of a group $G$, typically a finite group. Here $\C G$ is a Hopf algebra, and we used noncommutative geometry and Hopf algebra methods that can be generalised to other Hopf algebras such as q-deformed  $U_q(su_2)$, its finite-dimensional quotient $u_q(su_2)$ or the quantum double of a finite group. Conversely, if such q-deformations are relevant to physics (as appears to the case in 2+1 quantum gravity or the Kitaev model in condensed matter) then the $\C G$-valued theory in the present work {\em is} the relevant warm up.  Noncommutative geometry is used also for the base $X$ since, even though the `coordinate algebra' is just functions on the set of vertices,  differential 1-forms are literally linear combinations of graph arrows and  can never commute with all functions. 

The $\C G$-valued theory has a very different character from usual LTG based on $G$-valued `holonomies' (parallel transports) associated to arrows and $G$-valued gauge fields, in that both are allowed to be formal linear combinations of group elements. We found that the enlargement of both somewhat cancels out in that the moduli space $\CU^\times/\CG$ are what we expect in the Abelian case regarding the phase degrees of freedom (namely trivial when the $X$ is simply connected) but with some residual `tropical' $\R_{>0}$ modes. At the same time, we found that the $\C G$-valued theory can be viewed as a collection of noncommutative $U(d_\rho)$-valued gauge theories for the different nontrivial irreps $\rho$ of $G$. For example, our $\C\Z_2$-valued gauge theory can be viewed as a noncommutative $U(1)$-gauge theory of a more conventional form that would be usual $U(1)$-gauge theory if the graph or lattice were replaced by a smooth manifold. 
We studied the associated quantum gauge theories, including correlation functions and expectation values for $G=\Z_2$ and some different graphs $X$, including the non-simply connected case of a square. In particular, we also saw new effects in our detailed study of the resulting quantum gauge field theory on the finite chain and the square. They included boundary vs bulk effects  in the finite chain in the spirit of \cite{ArgMa2} and a well-behaved theory of the Lorentzian square completely unlike usual LGT which does not behave well with $\imath$ in front of the action. It would be interesting to look at such effects for a larger class of graphs. 

Some directions for further work are as follows.  Firstly, one can say rather more about how our theory relates to quantisation of usual LTG both from the spin networks approach and in other approaches. Secondly, since we used well-established Hopf algebra methods, the same detailed analysis could be carried out for other (non-group algebra) quantum groups. (Doing the same for the enveloping algebra of a Lie algebra also fits our methods but is not fundamentally different from the group algebra case.) Next, we focussed mainly on discrete $G$ with the geometry of $G$ being the `universal calculus', which corresponds to equipping $G$ with its complete graph. The last main section of the paper looked at the theory with non-universal (i.e. other Cayley graph) differential structures on $G$ but only solved the problem of what are the differentiable gauge transformations in the case where these are $G$-valued. What means `differentiable' in the $\C G$-valued case remains an open problem.

In fact,  the Hopf algebra methods we used are capable of handling geometrically nontrivial bundles, where a connection is not just a $G$-valued or (in our case) $\C G$-valued 1-form on the base. The notion of quantum group principal bundles $P$ with Hopf algebra fibre, and connections on them, was introduced in \cite{BrzMa} and is by now very well established in algebra. Gauge transformations can then be understood as bundle automorphisms both of a `convolution-invertible' type (generalising our $\C G$-valued ones) or of `algebra homomorphism' type (generalising the $G$-valued case in Section~\ref{secnonu}), with  `differentiability' clear for the latter type. There are also interesting links to Hopf algebroids, see \cite{HanLan}.  Even the case of $\C G$ structure Hopf algebra as  here but analysing the resulting gauge theories for some examples of  nontrivial bundles would be interesting for further work. 

Finally, returning to the physics, several topics of interest in the usual LGT case can be looked at in our version for group algebras $\C G$ (or beyond). Notably, the analysis in \cite{Don:lat} of what happens to the theory when a connected graph is disconnected by removal of some edges, in terms of gauge theories on the two halves and physics on the boundary, could be revisited in our setting. It would also be interesting to look at various concepts of entropy on graphs in relation to the physics in this context. The `tropical' modes also have some similarity to quantum gravity on graphs using noncommutative methods\cite{Ma:squ,ArgMa,BliMa} which should be investigated further. Last but not least, coupling gauge fields with matter fields in this context could be studied systematically as sections of an associated bundle (which makes sense in noncommutative geometry). Scalar matter can be understood in this context by taking $\phi\in A = \C(X)$ and the other relevant constructions accordingly, see \cite[Corollary 5.14]{BegMa}, but the spinor case on graphs still needs to be developed in order to connect to physics, and here the `spectral triples' approach of Connes can be a starting point. This goes back to a different style of noncommutative geometry coming out of operator algebras\cite{Con}, but can be adapted to the more explicit kind of `quantum Riemannian geometry' used here, see \cite[Chap.~8]{BegMa}.  A related problem (for which matter fields would be a good place to start) is  the lack of a theory of noncommutative variational calculus, which represents a gap between the functional integral methods we used vs the equations of motion and Hamiltonian formulation. Both these aspects will be looked at elsewhere. We note that classically, variational calculus goes via the cohomology of the total space of the jet bundle, which is only starting to be understood in noncommutative geometry\cite{MaSi,Flo}.


\begin{thebibliography}{ggghhh}

\bibitem{ArgMa} J.N. Argota-Quiroz and S Majid, Quantum gravity on polygons and FLRW model, Class. Quant. Grav. 37.24 (2020) 245001

\bibitem{ArgMa2} J.N. Argota-Quiroz and S Majid,  Quantum Riemannian geometry of the discrete interval and q-deformation, J. Math. Phys. 64 (2023) 051701

\bibitem{ACI:dis} P. Aschieri, L. Castellani and A. P. Isaev, Discretized Yang--Mills and Born--Infeld actions on finite group geometries, Int. J. Mod. Phys. A 18 (2003) 3555--3585


%\bibitem{AzMa:cro} R. Aziz and S. Majid, Quantum differentials on cross product Hopf algebras”, J. Algebra  563 (2020)  303--351

%\bibitem{And:var} I. M. Anderson, The variational bicomplex, Utah State Tech.  Rep. 1989.
% URL: \url{http://math.uni.lu/?michel/data}.

\bibitem{Bae:spi} JC Baez, Spin networks in gauge theory, Adv. Math. 117 (1996) 253--272

 \bibitem{BegMa} E.J. Beggs and S. Majid, {\em Quantum Riemannian Geometry},  Grundlehren der mathematischen Wissenschaften, Vol. 355, Springer (2020) 809pp
 
\bibitem{Sho} S. Beigi, P. W. Shor, and D. Whalen, The quantum double model with boundary: Condensations and symmetries, Commun. Math. Phys. 306 (2011) 663--694
 
\bibitem{BliMa} S. Blitz and S. Majid,  Quantum curvature fluctuations and the cosmological constant in a single plaquette quantum gravity model, in press, Class. Quantum Grav. Lett.  (2025) (6pp)


\bibitem{BrzMa} T. Brzezinski and S. Majid, Quantum group gauge theory on quantum spaces, Commun. Math. Phys.  157 (1993) 591--638

%\bibitem{Brz:tra} T. Brzezinski, Translation map in quantum principal bundles, J. Geom. Phys. 20 (1996) 349--370

\bibitem{CasPa:dis} L. Castellani, and C. Pagani, Finite group discretization of Yang--Mills and Einstein actions, Ann. Phys.  297 (2002) 295--314



% \bibitem{Cre:tra} M. Creutz, Gauge fixing, the transfer matrix, and confinement on a lattice, Phys. Rev. D 15 (1977) 1128

\bibitem{Con} A. Connes, {\em Noncommutative Geometry}, Academic Press 1994. 

\bibitem{CowMa} A. Cowtan and S. Majid, Quantum double aspects of surface code models, J. Math. Phys. 63 (2022) 042202 (49pp)

\bibitem{Die} R.  Diestel, {\em Graph Theory}, Graduate Texts in Mathematics, Vol. 173, Springer (2012)


\bibitem{DMHS:lgt} A. Dimakis, F. Mueller-Hoissen and T. Striker, Noncommutative differential calculus and lattice gauge theory, J. Phys. A  26 (1993) 1927

\bibitem{DMH:fin}  A. Dimakis, F. Mueller-Hoissen, Differential calculus and gauge theory on finite sets, J. Phys. A 27 (1994) 3159

\bibitem{Don:lat} W. Donnelly, Decomposition of entanglement entropy in lattice gauge theory, Phys. Rev. D 85 (2012)  085004

% \bibitem{Eli:imp} S. Elitzur, Impossibility of spontaneously breaking local symmetries, Phys. Rev. D 12 (1975) 3978

\bibitem{Flo}K.J. Flood, M. Mantegazza, H. Winther,  Jet functors in noncommutative geometry.
arXiv:2204.12401v1 (math.qa)

\bibitem{Gho} A. Ghobadi, Isotropy quotients of Hopf algebroids and the fundamental groupoid of digraphs, J. Algebra 610 (2022) 591--631

\bibitem{GJM} A. Grigorayan, R. Jimenez, and Y. Muranov, Fundamental groupoids of digraphs and graphs, Czech. Math. J.  68  (2018) 35--65

\bibitem{HanLan} X. Han and G. Landi, Gauge groups and bialgebroids, Lett. Math. Phys. 111 (2021) 1--43

\bibitem{Kit}A. Kitaev, Fault-tolerant quantum computation by anyons, Ann. Phys. 303, 2--30 (2003).

\bibitem{KogSus:ham} J. Kogut and L. Susskind, Hamiltonian formulation of Wilson lattice gauge theories, Phys. Rev. D   11 (1975) 395.

%\bibitem{LMR}J. Lopez-Pena, S. Majid, and K. Rietsch, Lie theory of finite simple groups and the Roth property,  Math. Proc. Cam. Phil. Soc. 163 (2017) 301--340

%\bibitem{Ma:fin} S. Majid, Riemannian geometry of quantum groups and finite groups with nonuniversal differentials,  Commun. Math. Phys. 225 (2002) 131--170

\bibitem{Ma:cli} S. Majid, Noncommutative Physics on Lie Algebras,$(\Z_2)^n$ Lattices and Clifford Algebras, in {\em Clifford Algebras: Applications to Mathematics, Physics, and Engineering}, Springer (2004) 491--518

\bibitem{Ma:gra} S. Majid, Noncommutative Riemannian geometry on graphs, J. Geom. Phys. 69 (2013) 74--93

\bibitem{Ma:hod} S. Majid, Hodge star as braided Fourier transform, Alg. Repn. Theory, 20 (2017) 695--733

%\bibitem{Ma:boo} S. Majid, Quantum geometry of Boolean algebras and de Morgan duality, J. Noncomm. Geom. 17 (2023) 37--79

\bibitem{Ma:squ} S. Majid, Quantum gravity on a square graph, Class. Quant. Grav. 36  (2019) 245009

%\bibitem{Ma:par} S. Majid, Quantum Riemannian geometry and particle creation on the integer line, Class. Quant. Grav. 36.13 (2019) 135011

%\bibitem{Ma:dia} S. Majid, Diagrammatics of braided group gauge theory, J.  Knot Theory Ramif. 8 (1999) 731--771

\bibitem{MaRai} S. Majid, E. Raineri, Electromagnetism and gauge theory on the permutation group S3, J. Geom. Phys., 44 (2002) 129--155

\bibitem{MaRie} S. Majid and K. Rietsch, Lie theory and coverings of finite groups, J. Algebra, 389 (2013) 137--150

\bibitem{MaSi} S. Majid and F. Sim\~{a}o, Quantum jet bundles,  Lett. Math. Phys. (2023) 113:120 (60pp)

\bibitem{Smi} J. Smit, Introduction to quantum fields on a lattice, Cambridge University Press, 2003

\bibitem{Wil:con} K.G. Wilson, Confinement of quarks, Phys.  Rev. D 10 (1974) 2445

\bibitem{FaIva:fms}  M. Faber, and A. N. Ivanov. On free massless (pseudo) scalar quantum field theory in 1+ 1-dimensional space-time, Euro.  Phys. J. C 24 (2002) 653--663.

\end{thebibliography}
\end{document}